\documentclass[journal,letterpaper,onecolumn,twoside,nofonttune]{IEEEtran}

\usepackage[utf8]{inputenc}
\usepackage[T1]{fontenc}
\usepackage{url}
\usepackage{ifthen}
\usepackage{cite}
\usepackage[cmex10]{amsmath}
\usepackage{balance}

\usepackage{amsmath}
\usepackage{amsthm}
\usepackage{amsfonts}
\usepackage{amssymb}
\usepackage{mathtools}
\usepackage{color}
\usepackage{verbatim}
\usepackage[normalem]{ulem}
\usepackage{enumitem}
\usepackage{hyperref}
\usepackage[ruled,vlined]{algorithm2e}
\usepackage{algpseudocode}
\usepackage{mathrsfs}
\usepackage{lipsum}

\usepackage{caption}

\usepackage{float} 

\definecolor{purple}{rgb}{0.5, 0.0, 0.5}
\definecolor{dark_green}{rgb}{0.0, 0.5, 0.0}
\definecolor{dyellow}{rgb}{0.99, 0.93, 0.0}
\definecolor{amaranth}{rgb}{0.9, 0.17, 0.31}

\newcommand{\white}{\color{white}}

\newcommand{\N}{\mathbb{N}}
\newcommand{\Z}{\mathbb{Z}}
\newcommand{\Q}{\mathbb{Q}}
\newcommand{\R}{\mathbb{R}}

\newcommand{\E}{\mathbb{E}}

\newcommand{\rpm}{\raisebox{.2ex}{$\scriptstyle\pm$}}

\newcommand{\D}{\mathcal{D}}

\newcommand{\I}{\mathcal{I}}
\newcommand{\J}{\mathcal{J}}
\newcommand{\K}{\mathcal{K}}
\newcommand{\Kbar}{\bar{\mathcal{K}}}

\newcommand{\Nn}{\mathcal{N}}

\newcommand{\lag}{\mathscr{L}}
\newcommand{\Scal}{\mathcal{S}}
\newcommand{\Tcal}{\mathcal{T}}
\newcommand{\Uu}{\mathcal{U}}

\newcommand{\Sbar}{\bar{\mathcal{S}}}
\newcommand{\Tbar}{\bar{\mathcal{T}}}
\newcommand{\ab}{\bold{a}}
\newcommand{\bb}{\bold{b}}

\newcommand{\dt}{\tilde{d}}

\newcommand{\eb}{\bold{e}}

\newcommand{\gb}{\bold{g}}

\newcommand{\qbar}{\bar{q}}
\newcommand{\Ub}{\bold{U}}
\newcommand{\Ubt}{\tilde{\Ub}}
\newcommand{\ub}{\bold{u}}
\newcommand{\Vb}{\bold{V}}

\newcommand{\vb}{\bold{v}}
\newcommand{\xb}{\bold{x}}
\newcommand{\xbt}{\tilde{\xb}}
\newcommand{\xbh}{\hat{\xb}}

\newcommand{\yb}{\bold{y}}

\newcommand{\ybt}{\tilde{\yb}}

\newcommand{\Ab}{\bold{A}}
\newcommand{\Bb}{\bold{B}}

\newcommand{\Cb}{\bold{C}}
\newcommand{\Db}{\bold{D}}
\newcommand{\Dbwt}{\widetilde{\Db}}
\newcommand{\Eb}{\bold{E}}

\newcommand{\Ebwt}{\widetilde{\bold{E}}}

\newcommand{\Gb}{\bold{G}}

\newcommand{\Ib}{\bold{I}}
\newcommand{\Hbh}{\hat{\bold{H}}}

\newcommand{\Mb}{\bold{M}}
\newcommand{\Mbh}{\hat{\Mb}}

\newcommand{\Qb}{\bold{Q}}
\newcommand{\Qt}{\tilde{Q}}

\newcommand{\Sb}{\bold{S}}
\newcommand{\Sigb}{\bold{\Sigma}}

\newcommand{\Sbwt}{\widetilde{\Sb}}

\newcommand{\Abh}{\hat{\Ab}}

\newcommand{\Abwh}{\widehat{\Ab}}

\newcommand{\Sbw}{\Sbwt_{\wb}}

\newcommand{\wb}{\bold{w}}
\newcommand{\Wb}{\bold{W}}
\newcommand{\Wbb}{\bar{\bold{W}}}

\newcommand{\Zb}{\bold{Z}}
\newcommand{\Xb}{\bold{X}}
\newcommand{\Yb}{\bold{Y}}
\newcommand{\At}{\tilde{A}}
\newcommand{\Abt}{\tilde{\Ab}}

\newcommand{\Bbh}{\hat{\Bb}}
\newcommand{\bt}{\tilde{b}}
\newcommand{\bbt}{\tilde{\bb}}

\newcommand{\bbwh}{\widehat{\bb}}

\newcommand{\Dt}{\tilde{\mathcal{D}}}
\newcommand{\Rh}{\hat{R}}
\newcommand{\Rt}{\tilde{R}}
\newcommand{\rh}{\hat{r}}
\newcommand{\rt}{\tilde{r}}

\newcommand{\gh}{\hat{g}}
\newcommand{\Hh}{\hat{H}}

\newcommand{\Thb}{\bold{\Theta}}
\newcommand{\Omb}{\bold{\Omega}}
\newcommand{\Ombb}{\bar{\bold{\Omega}}}
\newcommand{\Ombwt}{\widetilde{\Omb}}
\newcommand{\Ombw}{\Omb_{\wb}}
\newcommand{\Pib}{\bar{\Pi}}

\newcommand{\Pis}{\Pi^{\star}}
\newcommand{\Pit}{\tilde{\Pi}}
\newcommand{\PsiB}{\bold{\Psi}}

\newcommand{\psiv}{\vec{\psi}}
\newcommand{\wbt}{\tilde{\wb}}
\newcommand{\Wbt}{\tilde{\Wb}}

\newcommand{\Ft}{\tilde{F}}
\newcommand{\ellt}{\tilde{\ell}}

\newcommand{\rank}{\mathrm{rank}}
\newcommand{\lcm}{\mathrm{lcm}}
\newcommand{\diag}{\mathrm{diag}}
\newcommand{\err}{\mathrm{err}}
\newcommand{\image}{\mathrm{im}}
\newcommand{\SVD}{\mathrm{SVD}}
\newcommand{\QR}{\mathrm{QR}}

\newcommand{\qvec}[1]{\textbf{\textit{#1}}}

\DeclarePairedDelimiter\floor{\lfloor}{\rfloor}

\DeclarePairedDelimiter\bceil{\big\lceil}{\big\rceil}
\DeclarePairedDelimiter\bfloor{\big\lfloor}{\big\rfloor}

\newtheorem{Thm}{Theorem}
\newtheorem{Cor}{Corollary}
\newtheorem{Prop}{Proposition}
\newtheorem{Lemma}{Lemma}

\newtheorem{Rmk}{Remark}

\DeclareMathOperator*{\argmin}{arg\,min}

\begin{document}
\title{Gradient Coding with Iterative Block Leverage Score Sampling}

\author{$\textbf{Neophytos Charalambides}^{\mu}$, $\textbf{Mert Pilanci}^{\sigma}$, \textbf{and} $\textbf{Alfred O. Hero III}^{\mu}$\\
$\text{\white.}^{\mu}$EECS Department University of Michigan $\text{\white.}^{\sigma}$EE Department Stanford University\\
  Email: \texttt{neochara@umich.edu}, \texttt{pilanci@stanford.edu}, \texttt{hero@umich.edu}
\vspace{-4mm}
}

\maketitle

\begin{abstract}
Gradient coding is a method for mitigating straggling servers in a centralized computing network that uses erasure-coding techniques to distributively carry out first-order optimization methods. Randomized numerical linear algebra uses randomization to develop improved algorithms for large-scale linear algebra computations. In this paper, we propose a method for distributed optimization that combines gradient coding and randomized numerical linear algebra. The proposed method uses a randomized $\ell_2$-subspace embedding and a gradient coding technique to distribute blocks of data to the computational nodes of a centralized network, and at each iteration the central server only requires a small number of computations to obtain the steepest descent update. The novelty of our approach is that the data is replicated according to importance scores, called block leverage scores, in contrast to most gradient coding approaches that uniformly replicate the data blocks. Furthermore, we do not require a decoding step at each iteration, avoiding a bottleneck in previous gradient coding schemes. We show that our approach results in a valid $\ell_2$-subspace embedding, and that our resulting approximation converges to the optimal solution.
\end{abstract}

\begin{IEEEkeywords}
Low-Rank Approximations, Least-Squares Regression, Randomized Algorithms, Randomized Numerical Linear Algebra, Coded Computing, Stragglers, Erasure-Coding, Replication-Coding.
\end{IEEEkeywords}

\section{Introduction}

In this work we bridge two disjoint areas, to accelerate first-order methods such as steepest descent distributively, while focusing on linear regression. Specifically, we propose a framework in which Randomized Numerical Linear Algebra (RandNLA) sampling algorithms can be used to devise Coded Computing (CC) schemes. We focus on the task of $\ell_2$-subspace embedding through an importance sampling procedure known as leverage score sampling; which scores measure the relative importance of each data point within a given dataset, and distributed gradient computation in the presence of stragglers; which is referred to as gradient coding (GC).

Traditional numerical linear algebra algorithms are deterministic. For example, inverting a full-rank matrix $\Mb\in\R^{N\times N}$ requires $O(N^3)$ arithmetic operations by performing Gaussian elimination, as does naive matrix multiplication. The fastest known algorithm which multiplies two $N\times N$ matrices, requires $O(N^{\omega})$ operations, for $\omega<2.372$ \cite{PV20,WXXZ24}. Other important problems are computing the determinant, singular and eigenvalue decompositions, $\SVD$, $\QR$ and Cholesky factorizations of large matrices.

Although these deterministic algorithms run in polynomial time and are numerically stable, their exponents make them prohibitive for many applications in scientific computing and machine learning, when $N$ is in the order of millions or billions \cite{MLG08,Mah12}. To circumvent this issue, one can perform these algorithms on a significantly smaller approximation. Specifically, for a matrix $\Sb\in\R^{r\times N}$ with $r\ll N$, we apply the deterministic algorithm on the surrogate $\Mbh=\Sb\Mb\in\R^{r\times d}$ for $\Mb\in\R^{N\times d}$. The matrix $\Sb$ is referred to as a dimension-reduction or a sketching matrix, and $\Mbh$ is a sketch of $\Mb$, which contains as much information about $\Mb$ as possible. For instance, when computing the product of $\Ab\in\R^{N\times L}$ and $\Bb\in\R^{N\times M}$, we apply a carefully chosen $\Sb\in\R^{r\times N}$ on each matrix to get
$$ \overbrace{\begin{pmatrix} & & & & & & \\ & & & \Ab^\top & & & \\ & & & & & & \end{pmatrix}}^{L\times N} \overbrace{\begin{pmatrix} & & \\ & & \\ & & \\ & \Bb & \\ & & \\ & & \\ & & \end{pmatrix}}^{N\times M} \approx \overbrace{\begin{pmatrix} & & & & \\ & & \Abh^\top & & \\ & & & & \end{pmatrix}}^{L\times r} \overbrace{\begin{pmatrix} & & \\ & & \\ & \Bbh\ & \\ & & \\ & & \end{pmatrix}}^{r\times M} $$
where $\Abh=\Sb\Ab$ and $\Bbh=\Sb\Bb$. Thus, naive matrix multiplication now requires $O(LMr)$ operations, instead of $O(LMN)$. Such approaches have been motivated by the Johnson-Lindenstrauss lemma \cite{JL84,DG03}, and require low complexity.

A multitude of other problems, such as $k$-means clustering \cite{BZMD14,CEMMP15,Cha20} and computing the $\SVD$ of a matrix \cite{DFKVV99,DFKVV04,DKM06a,DKM06b}, make use of this idea in order to accelerate computing accurate approximate solutions. We refer the reader to the following monographs and comprehensive surveys on the rich development of RandNLA \cite{HMJ11,Mah12,Mah16,MM20,Wan15,Woo14,DM16,KV17,murray2023}, an interdisciplinary field that exploits randomization as a computational resource, to develop improved algorithms for large-scale linear algebra problems.

The problem of $\ell_2$-subspace embedding, a form of spectral approximation of a matrix, has been extensively studied in RandNLA. The main techniques for constructing appropriate $\ell_2$-subspace embedding sketching matrices, are performing a random projection or row-sampling. Well-known choices of $\Sb$ for reducing the effective dimension $N$ to $r$ include: i) \textit{Gaussian projection}; for a matrix $\Thb$ where $\Thb_{ij}\sim\Nn(0,1)$ define $\Sb=\frac{1}{\sqrt{r}}\Thb$, ii) \textit{leverage score sampling}; sample with replacement $r$ rows from the matrix according to its normalized leverage score distribution and rescale them appropriately, iii) \textit{Subsampled Hadamard Transform} (SRHT); apply a Hadamard transform and a random signature matrix to judiciously make the leverage scores approximately uniform and then follow similar steps to the leverage score sampling procedure.

In this paper, we first generalize ii) to appropriately sample submatrices instead of rows to attain an $\ell_2$-subspace embedding guarantee. We refer to such approaches as \textit{block sampling}. Throughout this paper, sampling is done with replacement. Sampling blocks has been explored in block-iterative methods for solving systems of linear equations \cite{Elf80,Gut06,NT14,RN20}. Our motivation in dealing with blocks rather than individual vectors, is the availability to invoke results that can be used to characterize the approximations of distributed computing networks, to speed up first-order methods, as sampling individual rows/columns is prohibitive in real-world environments. This in turn leads to an \textit{iterative sketching} approach, which has been well studied in terms of second-order methods \cite{PW16,PW17,LLDP20}. By iterative sketching, we refer to an iterative algorithm which uses a new sketching matrix $\Sb_{[s]}$ at each iteration, i.e., throughout such a procedure the sketches $\big(\Sbwt_{[1]}\Ab,\Sbwt_{[2]}\Ab,\ldots,\Sbwt_{[n]}\Ab\big)$ are obtained. The scenario where a single sketching matrix $\Sb$ is applied before the iterative process, is referred to as the \textit{sketch-and-solve paradigm} \cite{Sar06}, which also induces bias towards the samples considered through $\Sb$.

Second, we propose a general framework which incorporates our sketching algorithm into a CC approach. This framework accommodates a central class of sketching algorithms, that of importance (block) sampling algorithms (e.g., $CUR$ decomposition \cite{OJXE18}, $CR$-multiplication \cite{CPH20c}). Coded computing is a  novel computing paradigm that utilizes coding theory to effectively inject and leverage data and computation redundancy to mitigate errors and slow or non-responsive servers; known as \textit{stragglers}, among other fundamental bottlenecks, in large-scale distributed computing. In our setting, the straggling effect is due to computations being communicated over erasure channels, whose erasure probability follows a probability distribution which is central to the CC probabilistic model. The seminal work of \cite{LLPPR17} which first introduced CC, focused on exact matrix-vector multiplication and data shuffling. More recent works deal with recovering good approximations, while some have utilized techniques from RandNLA, e.g., \cite{BP23,CPH20a,CPH20c,GWCR18,GKCMR20,JM21,THRD21,RCHV23}. Our results are presented in terms of the standard CC probabilistic model proposed in \cite{LLPPR17}, though they extend to any computing network comprised of a central server and computational nodes, referred to as \textit{servers}.

To mitigate stragglers, we appropriately encode and replicate the data blocks, which leads to accurate CC estimates. In contrast to previous works which simply replicate each computational task or data block the same number of times \cite{ZKJKS08,TLDK17,CMH20,CMH21}, we replicate blocks a number of times that corresponds to their \textit{block leverage scores}. Consequently, this induces a non-uniform sampling distribution in the aforementioned CC model, which is an approximation to the normalized block leverage scores. Furthermore, this eliminates the need for a decoding step, which can be a high complexity bottleneck in CC. A drawback of using RandNLA techniques is that exact computations are not recoverable. For more details on the approximation error and computational complexity of CC, the reader is referred to the monographs \cite{LA20,ng2021}. A CC schematic, focusing on approximate GC, is presented in Figure 1.
\begin{figure*}[t]
  \centering
  \includegraphics[scale=.18]{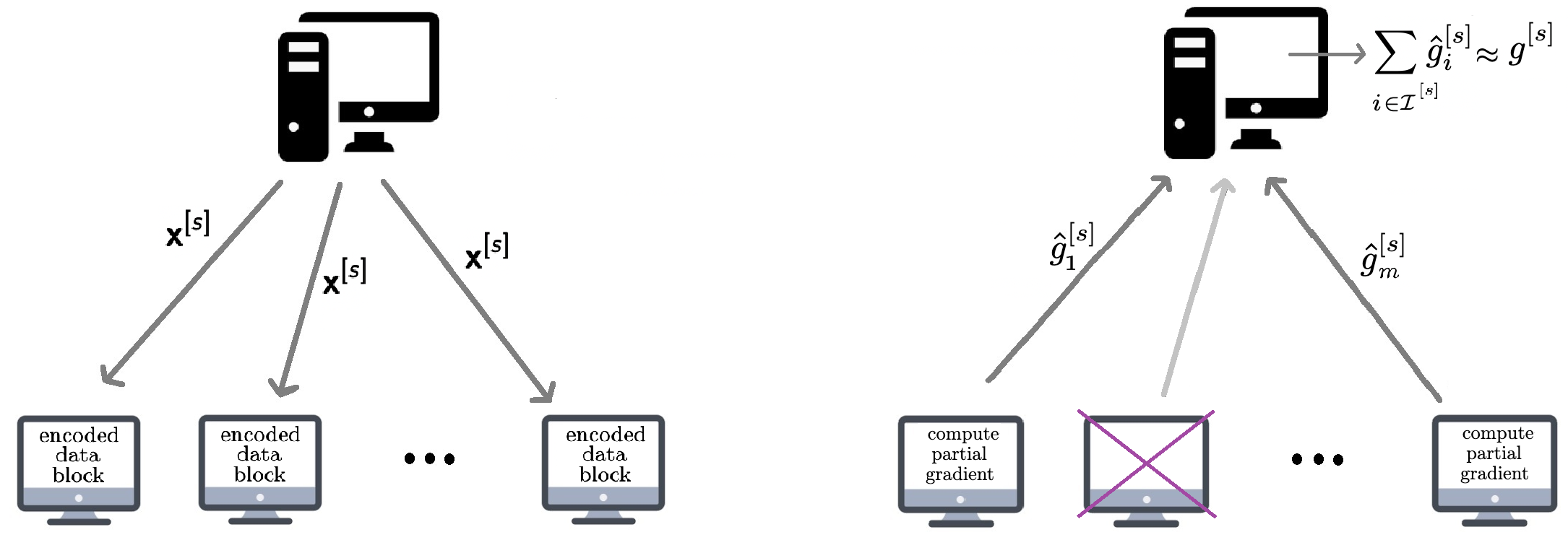}
  \caption{Schematic of our approximate GC scheme, at iteration $s$. Each server has an encoded block of data, of which they compute the gradient once they receive the updated parameters vector $\xb^{[s]}$. The central server then aggregates all the received partial gradients indexed by $\I^{[s]}$, i.e., $\big\{\gh_i^{[s]}\big\}_{i\in\I^{[s]}}$, to approximate the gradient $g^{[s]}$. At each iteration we expect a different index set $\I^{[s]}$, which leads to iterative sketching.}
  \label{CG_schematic}
\end{figure*}

The central idea behind our approach is that non-uniform importance sampling can be emulated, by replicating tasks across the network's computational nodes, who communicate through an erasure channel with the central server. By emulate, we mean that after replicating elements according to their importance sampling distribution and then sampling uniformly at random (without replacement) from the set containing redundant elements, we mimic sampling with replacement from the initial set, according to the importance sampling distribution. The tasks' computation times follow a runtime distribution \cite{LLPPR17}, which, along with a prespecified gradient transmission \textit{ending time} $T$, determine the number of replications per task across the network. Though similar ideas have been proposed \cite{GWCR18,GKCMR20,CPH20a,CPH20c} where sketching has been incorporated into CC, this is the first time redundancy is introduced through RandNLA, as compared to compression, in order to reduce complexity and provide approximation guarantees. Although uniform replication is a powerful technique, it does not capture the non-equal relevance of different rows of the data matrix to the information in the dataset. We capture such information through replication and rescaling according to the block leverage scores. By performing uniform sampling within these blocks, we attain a spectral approximation of the data matrix. One can compare this method to doing an iterative $\ell_2$-subspace embedding sketching of $\Ab$ in standard CC. The shortcoming of this approach, is that a large amount of servers may be required when the underlying sampling distribution is non-uniform. We are able to overcome this obstacle, by approximating the importance sampling distribution as close as possible, under the physical constraint imposed by the network that the number of servers is fixed.

In Appendix \ref{weighting_sec} we discuss how further compression can be attained by introducing weighting, while guaranteeing the same approximation error when first and second order methods are used for linear regression in the sketch-and-solve paradigm (Proposition \ref{prop_wght_lr} and Corollary \ref{unb_est_Sw}). We also show that, in terms of the expected reduction in sketching dimension, that the benefit decreases as the block leverage score distribution approaches the uniform distribution (Theorem \ref{thm_exp_dim}). This further justifies the fact that higher diversity in leverage scores leads to more accurate algorithms \cite{PKB14}.

At each iteration of the proposed descent algorithm, all completed jobs that are received by the central server are aggregated to form the final gradient approximation. Thus, unlike many other CC schemes, the proposed method does not store completed jobs which will not be accounted for. Our method sacrifices accuracy for speed of computation, and the inaccuracy is quantified in terms of the resulting $\ell_2$-subspace embedding (Theorem \ref{subsp_emb_thm_lvg}). Specifically, the computations of the responsive servers will correspond to sampled block computation tasks enabled by leverage score sampling, as summarized in Algorithm \ref{app_GC_alg}. Approximate CC is of interest since it could lead to faster inexact but accurate solutions, at a reduced computational cost \cite{LA20}.

To summarize, our main contributions are:
\begin{enumerate}
  \item propose \textit{block leverage score sketching}, to accommodate block sampling for $\ell_2$-subspace embedding, which makes leverage score sampling techniques suitable for distributed computations,
  \item provide theoretical guarantees for the sketching performance,
  \item show the significance of weighting, when our \textit{weighted} sketching algorithm is used in iterative first and second order methods,
  \item propose the use of \textit{expansion networks}, which use the sampling distribution to determine how to replicate and distribute blocks in the CC framework --- this unifies the disciplines of RandNLA and CC --- where replication and uniform sampling (without a random projection) result in a spectral approximation,
  \item show how expansion networks can be used for approximate distributed \textit{steepest descent} (SD), which approaches the optimal solution using unbiased gradient estimators in a similar manner to \textit{batch stochastic steepest descent} (SSD),
  \item experimental validation of the performance of our algorithm on artificial datasets with non-uniform induced sampling distributions.
\end{enumerate}

The paper is organized as follows. In Subsection \ref{rel_work_CC_subs} we present related works, in terms of CC, and in Subsection \ref{appr_GC_related_work_subs} we give an overview of the approximate GC literature. In Section \ref{not_backgr_sec} we present the notation which will be used, and review necessary background. We also list the notational conventions used in the paper, in Appendix \ref{notation_app}. In Subsection \ref{lvg_sec} we present our sketching algorithm and its theoretical approximation guarantees. In Subsections \ref{exp_netw_sec} and \ref{opt_ind_distr_subsec} we give a framework for which our sketching algorithm and potentially other importance sampling algorithms, can be used to devise CC schemes. This is where we introduce redundancy through RandNLA, which has not been done before. In Subsection \ref{sk_lin_regr} we summarize our GC scheme, and in Subsection \ref{synopsis_subsec} we give a brief synopsis of our main results by tying everything together. A concise description of our distributed procedure can be found in Algorithm \ref{app_GC_alg} of Subsection \ref{sk_lin_regr}. We conclude with experimental evaluations in Section \ref{exper_sec} on simulated data with highly non-uniform underlying sampling distributions, as empirical validation.

\subsection{Related Work on CC}
\label{rel_work_CC_subs}

There are several related works \cite{FD16,JM19,GWCR18,GKCMR20} to the GC approach we present in this paper. The authors of \cite{FD16} propose replicating subtasks according to what they define as the priority of the job, while \cite{GKCMR20} and \cite{JM19} incorporate sketching into CC. It is worth noting that even though the focus of this work is on first order gradient methods, our approach also applies to second-order methods, as well as approximate matrix products through the $CR$-multiplication algorithm \cite{DK01,DFKVV99,CPH20c}. We briefly discuss this in Section \ref{concl_sec}.

The work of \cite{FD16} deals with matrix-vector multiplication. Similar to our work, the authors of \cite{FD16} also replicate the computational tasks a certain number of times; and stop the process at a prespecified instance, but they do not use block leverage scores. Here, the computation $\Ab\xb$ for $\Ab\in\R^{N\times N}$ and $\xb\in\R^N$ is divided into $C$ different tasks, prioritizing the smaller computations. The $m$ servers are split into $c$ groups, and are asked to compute one of the tasks $\yb_j=\big(\sum_{i\in\J_j}\sigma_i\ub_i\vb_i^\top\big)\xb$, for $\Ab=\sum_{l=1}^N\sigma_l\ub_l\vb_l^\top$ the $\SVD$ representation of $\Ab$. Each task $\yb_j$ is computed by the servers of the respective group, and $\N_N=\bigsqcup_{j=1}^s\J_j$ is a disjoint partitioning of the rank-1 outer-products of the $\SVD$ representation. The size of the $j^{th}$ task is $|\J_j|=p_j$, which in our work is determined by the normalized block scores. The scores in our proposed scheme is motivated and justified by RandNLA, in contrast to the selection of the sizes $p_j$ which is not discussed in \cite{FD16}. Furthermore, the scheme of \cite{FD16} requires a separate maximum distance separable code for each job $\yb_j$; thus requiring multiple decodings, while we do not require a decoding step. Another drawback of \cite{FD16} is that an integer program is formulated to determine the optimal ending time, which the authors do not solve explicitly. In contrast, we determine a solution for any desired ending time. Lastly, we note that the $\ell_2$-subspace embedding approximation guarantee of our method, depends on the ending time $T$.

In terms of sketching and RandNLA, the works of \cite{GWCR18} and \cite{JM19} utilize the \textit{Count-Sketch} \cite{CCFC04}; which relies on hashing. In ``CodedSketch'' \cite{JM19}, Count-Sketches are incorporated into the design of a variant of the improved ``Entangled Polynomial Code'' \cite{YMAA20}, to combine approximate matrix multiplication with straggler tolerance. 
The code approximates the submatrix blocks $\{\Cb_{i,j}:(i,j)\in\N_{k_1}\times\N_{k_2}\}$ of the final product matrix $\Cb=\Ab\cdot\Bb$ (matrix $\Cb$ is partitioned $k_1$ times across its rows, and $k_2$ times across its columns), with an accuracy that depends on $\|\Cb_{i',j'}\|_F$ for all $(i',j')\in\N_{k_1}\times\N_{k_2}$. This hinders its practical applications, as it requires accuracy guarantees without oracle knowledge of the outcome of each submatrix of the matrix product $\Cb$. This approach permits each block of $\Cb$ to be approximately recovered, if only a subset of the servers complete their tasks.

In ``OverSketch'' \cite{GWCR18}, redundancy is introduced through additional Count-Sketches, to mitigate the effect of stragglers in distributed matrix multiplication. In particular, the Count-Sketches $\breve{\Ab}=\Ab\Sb$ and $\breve{\Bb}=\Sb^\top\Bb$ of the two inputs $\Ab$ and $\Bb$ are computed, and are partitioned into submatrices of size $b\times b$. The $b\times b$ submatrices of the final product $\Cb=\Ab\cdot\Bb$ are then approximately calculated, by multiplying the corresponding row-block of $\breve{\Ab}$ and column-block of $\breve{\Bb}$, each of which is done by a single server. The ``OverSketch'' idea has also been extended to distributed Newton Sketching \cite{GKCMR20} for convex optimization problems.

\subsection{Related Work on Approximate GC}
\label{appr_GC_related_work_subs}

Many approaches have been proposed for approximate GC, which introduces redundancy in order to accelerate the convergence rate of distributed optimization algorithms \cite{RTTD17,BWE20,KKR19,SH22,CMPH22,CMPH23,CPE17,GW21,CP18,WCP19,WLS19,CHZP18,CPH20a,HYKM19,MSM19}. These approaches can be grouped into two main categories:
\begin{enumerate}[label=(\arabic*)]
  \item \underline{\textit{fixed coefficient decoding}}: at each iteration the same decoding step takes place, which is a linear combination of the received computations with predetermined coefficients \cite{RTTD17,CPE17,KKR19,BWE20,CMPH23,CMPH22,SH22},
  \item \underline{\textit{optimal coefficient decoding}}: at each iteration another least squares optimization problem is solved, in order to determine the optimal coefficients that will be used in the linear combination of the received computations, as part of the decoding \cite{CPE17,WCP19,WLS19,CP18,CPH20a,GW21}.
\end{enumerate}
We note that optimal coefficient decoding does not necessarily lead to optimal convergence, though we use this term to be consistent with the literature. 

The scheme proposed in this paper falls in the category of fixed coefficient decoding, instead of optimal coefficient decoding. Even though we do not consider a decoding step, our approach can be viewed as having a fixed decoding step, which we explain in Subsection \ref{appr_GC_subsec}. The main drawback of optimal coefficient decoding schemes is that a decoding vector has to be determined at every iteration, which is usually done by solving a linear system of equations, putting an additional computational burden on the central server. In Subsection \ref{appr_GC_subsec} we show how our gradient approximations can be quantified in terms of the optimal coefficient decoding error, and also discuss how approaches from both fixed and optimal coefficient decoding can be accelerated by reducing the total number of iterations.

Similarly to our data replication approach is the work of \cite{BWE20}, which replicates each data point $\Ab_{(i)}$ proportionally to its norm $\|\Ab_{(i)}\|_2$, to attain what they call a \textit{pairwise-balanced} distribution scheme. The ``Stochastic Gradient Coding'' algorithm of \cite{BWE20} also converges in expectation to the optimal solution with appropriately selected step-sizes. We use a similar analysis to that of \cite{BWE20} to show that our method achieves convergence rate of $O(1/S)$ of the parameters, where $S$ is the total number of iterations which take place. Furthermore, we also show that our expected regret converges to zero at a rate of $O(1/\sqrt{S}+r/S)$, and, unlike \cite{BWE20}, we quantify the error of our gradient approximation in terms of the optimal gradient decoding residual error. The allocation of the data points in the Stochastic Gradient Coding algorithm \cite{BWE20} is done randomly, and the servers need to compute a weighted sum of the gradients corresponding to each of their allocated points, while we compute a partial gradient corresponding to the block and rescale it according to its block leverage score. The encoding of our scheme corresponds to the weighted coefficients of \cite{BWE20}, since both are rescalings of the individual and block partial gradients. A benefit of our proposed replication of the data according to the block leverage scores is that it is more appropriate for spectral approximations of the data matrix, as it preserves the subspace spanned by the basis of the data matrix $\Ab$. This in turn allows us to also draw comparisons to approximations from optimal coefficient decoding schemes (Theorem \ref{app_GC_thm}).

The authors of \cite{KKR19,SH22} base their constructions on balanced incomplete block designs (BIBDs), which are combinatorial structures that allow resiliency to potential adversarial stragglers. A drawback of the methods proposed is \cite{KKR19,SH22}, is that they do not provide a convergence guarantee to the optimal solution $\xb^\star$ of the linear system, for their proposed iterative approximate gradient based scheme. A similar line of work to \cite{BWE20} that also considers adversarial stragglers is that of \cite{CPE17}, which considers ``Fractional Repetition Codes'' \cite{TLDK17,CMH20} with approximate decoding. This work also proposes ``Bernoulli Gradient Codes'' and a regularized variant (rBGC), where each data point $\Ab_{(i)}$ is randomly assigned to a fixed number of distinct servers. This GC scheme is an approximation of the pairwise-balanced scheme of \cite{BWE20}, and can be viewed as a case of SSD when all the data points have the same norm.

In \cite{CMPH23}, an optimal (in expectation) step-size $\xi_s^{\star}$ was determined for each iteration of GC methods with unbiased gradients, which can be applied by the central server. The idea proposed in \cite{CMPH23} applies to all fixed coefficient decoding GC methods, where the step-size is modified in order to reduce the overall number of iterations needed for SD to converge. The optimal step-size $\xi_s^{\star}$ is determined by the central server at each iteration once it receives sufficiently many computations to do so, from the computational nodes. This computation replaces the decoding step performed by the central server, which is of high complexity.

Another line of work considers using low density matrices as encoding matrices for distributed optimization \cite{HYKM19,MSM19}. In \cite{HYKM19} the authors distribute the data to the servers using a Low Density Generator Matrix code. As with optimal coefficient decoding, a drawback of \cite{HYKM19} is that at each iteration the central server has to run a decoding algorithm to decode the sum of the received partial gradients. The work of \cite{MSM19} focuses on linear loss functions, in which the data sent to the computational nodes is encoded through a Low Density Parity Check code. As is the case for all exact schemes, in \cite{MSM19} there is a threshold on the number of servers that need to respond in order to recover the exact gradient by performing the decoding step. In the CC literature, this is referred to as the \textit{recovery threshold}, and has been extensively studied \cite{DFHJCG20}. If this threshold is not met in the scheme of \cite{MSM19}, the central server leverages the Low Density Parity Check code to compute an estimate of the gradient. To this extent, it is worth mentioning that both our proposed method and that of \cite{BWE20} are flexible in terms of the required number of responses, as in both approaches the gradient approximation degrades gracefully if more stragglers than expected occur. This is an important consideration, as in practice there is a great deal of variability in the number of stragglers over time.

It is important to also mention that there are many papers which consider distributed stochastic first-order methods which do not take account of stragglers \cite{RM51}. There are also many works on CC that focus on the issue of numerical stability, and the problems of exact and approximate matrix multiplication, e.g., \cite{DCG16,FC18,BP19,JSM19,DFHJCG20,JDCC21,KD22,SHN19,DRV21,RT21,DR21}. We summarize the most relevant results from the works discussed above on approximate GC in Table \ref{rel_work_table}.

\begin{center}
\begin{table}[h]
\centering
\begin{tabular}{ |p{2.59cm}||p{1.9cm}|p{1.2cm}|p{1.4cm}| }
\hline
\multicolumn{4}{|c|}{\textbf{Comparison of Related Work}} \\
\hline
\hline
\textbf{Gradient Code} & \textbf{Decoding $\quad$ Coefficients} & \textbf{Residual Error} & \textbf{Convergence of} $\xb^{[s]}$ \\
\hline
\hline
Expander Codes \cite{RTTD17} & Fixed & --- & Yes \\ \hline 
Expander Codes \cite{GW21} & Optimal & $Np^{\bar{r}-o(\bar{r})}$ & Yes \\ \hline
Pairwise Balanced \cite{BWE20} & Fixed & --- & Yes \\ \hline
BIBD \cite{KKR19} & Fixed \& Optimal & --- & No \\ \hline
rBGC \cite{CPE17} & Fixed & --- & No \\ \hline
\textit{Leverage Scores GC} & Fixed & $4\epsilon^2$ & Yes \\ \hline
\end{tabular}

\caption{The second column indicates the category of the GC. The third column contains the average case normalized gradient residual error $\frac{\|g^{[s]}-\gh^{[s]}\|_2^2}{\|\Ab\|_2\cdot\|\Ab\xb^{[s]}-\bb\|_2}$, where $p=\frac{m-q}{m}$ is the fraction of stragglers, $\bar{r}$ is the average replication factor, $N$ is the number of data points in the data matrix, and $\epsilon$ is the $\ell_2$-subspace embedding error of our sketching algorithm. The fourth column, indicates whether a convergence proof of the parameters vector $\xb^{[s]}$ is provided in the corresponding work. Our work is presented in the last row.}
\label{rel_work_table}
\end{table}
\end{center}

\section{Notation and Background}
\label{not_backgr_sec}

We denote $\N_n\coloneqq\{1,2,\ldots,n\}$, and $X_{\{n\}}=\{X_i\}_{i=1}^n$, where $X$ is an arbitrary variable. We use $\Ab,\Bb$ to denote real matrices, $\bb,\xb$ real column vectors, $\Ib_n$ the $n\times n$ identity matrix, $\bold{0}_{n\times m}$ and $\bold{1}_{n\times m}$ respectively the $n\times m$ all zeros and all ones matrices, and by $\eb_i$ the standard basis column vector whose dimension will be clear from the context. The largest eigenvalue of a matrix $\Mb$, is denoted by $\lambda_1(\Mb)$. By $\Ab_{(i)}$ we denote the $i^{th}$ row of $\Ab$, by $\Ab^{(j)}$ its $j^{th}$ column, by $\Ab_{ij}$ the value of $\Ab$'s entry in position $(i,j)$, and by $\xb_i$ the $i^{th}$ element of $\xb$. The rounding function to the nearest integer is expressed by $\left\lfloor\cdot\right\rceil$, i.e., $\left\lfloor a\right\rceil=\floor{a+1/2}$ for $a\in\R$. The cardinality of a set is expressed by $|\cdot|$, e.g., $|\{1,2,5,9\}|=4$. Disjoint unions are represented by $\bigsqcup$, e.g., $\Z=\{j:j\text{ is odd}\}\bigsqcup\{j:j\text{ is even}\}$, and we define $\biguplus$ as the union of multisets, e.g., $\{1,2,3\}\biguplus\{3,4\}=\{1,2,3,3,4\}$. The diagonal matrix with real entries $a_{\{n\}}$ is expressed as $\diag\big(a_{\{n\}})$. The Moore–Penrose inverse of a rank $m$ matrix $\Mb\in\R^{n\times m}$ with $n>m$, is defined as $\Mb^{\dagger}\coloneqq(\Mb^\top\Mb)^{-1}\Mb^\top$.

We partition vectors and matrices across their rows into $K$ submatrices, so that any two submatrices either have the same number of rows or one of them has at most one additional row. Such a partitioning is always attainable through simple arithmetic \cite{CMH21}. For example, if we are to partition 100 rows into 7 blocks, 5 blocks are comprised of 14 rows and 2 blocks are comprised of 15 rows, and $100=5*14+2*15$. This allows us to assume that each node has the same computational load, and expected response times. For simplicity, we assume that $K$ divides the number of rows $N$, denoted $K\mid N$. That is, for an $\ell_2$-subspace embedding of $\Ab\in\R^{N\times d}$ with target $\bb\in\R^N$, we assume $K\mid N$ and the ``size of each partition’’ is $\tau=N/K$.\footnote{If $K\nmid N$, we appropriately append zero vectors/entries until this is met. It is not required that all blocks have the same size, though we discuss this case to simplify the presentation. One can easily extend our results to blocks of varying sizes, and use the analysis from \cite{CMH20} to determine the optimal size of each partition.} We partition $\Ab$ and $\bb$ across their rows:
\begin{equation}
\label{part_matrix}
  \Ab = \Big[\Ab_1^\top \ \cdots \ \Ab_K^\top\Big]^\top \qquad \text{and} \qquad \bb = \Big[\bb_1^\top \ \cdots \ \bb_K^\top\Big]^\top
\end{equation}
where $\Ab_i\in\R^{\tau\times d}$ and $\bb_i\in\R^{\tau}$ for all $i\in\N_K$. Partitions, are referred to as \textit{blocks}. Throughout the paper we consider the case where $N\gg d$. For $\Ab$ full-rank, its $\SVD$ is $\Ab=\Ub\Sigb\Vb^\top$, where $\Ub\in\R^{N\times d}$ is its reduced left orthonormal basis.

Matrix $\Ab$ represents a dataset $\D$ of $N$ samples with $d$ features, and $\bb$ the corresponding labels of the data points. The partitioning \eqref{part_matrix} corresponds to $K$ sub-datasets $\D_{\{K\}}$, i.e., $\D=\bigsqcup_{j=1}^K\D_j$. Our results are presented in terms of an arbitrary partition $\N_N=\bigsqcup_{\iota=1}^K\K_\iota$, for $\N_N$ the index set of the rows of $\Ab$ and $\bb$. The index subsets $\K_{\{K\}}$ indicate which data samples are in each sub-dataset. By $\Ab_{(\K_\iota)}$, we denote the submatrix of $\Ab$ comprised of the rows indexed by $\K_\iota$. That is, for $\Ib_{(\K_\iota)}$ the restriction of $\Ib_N$ to only include its rows corresponding to $\K_\iota$, we have $\Ab_{(\K_\iota)}=\Ib_{(\K_\iota)}\cdot\Ab$. By $\K_{\iota}^i$, we indicate that the $\iota^{th}$ block was sampled at trial $i$, i.e., the superscript $i$ indicates the sampling trial and the subscript $\iota\in\N_K$ which block was sampled at that trial. An explicit example where this notation is used, can be found in Appendix \ref{app_example}. Also, by $j(i)$ we denote the index of the submatrix which was sampled at the $i^{th}$ sampling trial, i.e., $\K_{j(i)}=\K_{j(i)}^i$. The complement of $\K_j$ is denoted by $\Kbar_j$, i.e., $\Kbar_j=\N_N\backslash\K_j$, for which $\Ub_{(\K_j)}^\top\Ub_{(\K_j)}=\Ib_d-\Ub_{(\Kbar_j)}^\top\Ub_{(\Kbar_j)}$. By $t\gets T$, we mean $T$ is a realization of the time variable $t$.

Sketching matrices are represented by $\Sb$ and $\Sbwt_{[s]}$. The script $[s]$ indexes an iteration $s=0,1,2,\ldots$ which we drop when it is clear from the context. Throughout the paper, we will be reducing dimension $N$ of $\Ab$ to $r$, i.e., $\Sb\in\R^{r\times N}$. Sampling matrices are denoted by $\Omb\in\{0,1\}^{r\times N}$, and diagonal rescaling matrices by $\Db\in\R^{N\times N}$.

\subsection{Least Squares Approximation}

Linear least squares approximation, abbreviated to least squares, is a technique to find an approximate solution to a system of linear equations that has no exact solution \cite{DMMS11}. Consider the system $\Ab\xb=\bb$, for which we want to find an approximation to the best-fitted
\begin{equation}
\label{x_star_lr}
  \xb^\star = \argmin_{\xb\in\R^d}\Big\{L_{ls}(\Ab,\bb;\xb)\coloneqq\|\Ab\xb-\bb\|_2^2\Big\},
\end{equation}
where the objective function $L_{ls}$ has gradient 
\begin{equation}
\label{gr_ls}
  g^{[s]} = \nabla_{\xb}L_{ls}(\Ab,\bb;\xb^{[s]}) = 2\Ab^\top(\Ab\xb^{[s]}-\bb).
\end{equation}
We refer to the gradient of the block pair $(\Ab_i,\bb_i)$ from \eqref{part_matrix} as the $i^{th}$ \textit{partial gradient}
\begin{equation}
\label{part_gr_ls}
  g_i^{[s]} = \nabla_{\xb}L_{ls}(\Ab_i,\bb_i;\xb^{[s]}) = 2\Ab_i^\top(\Ab_i\xb^{[s]}-\bb_i).
\end{equation}
Existing exact methods find a solution vector $\xb^\star$ in $O(Nd^2)$ time, where $\xb^\star=\Ab^{\dagger}\bb$. In Subsection \ref{sk_lin_regr} we focus on approximating the optimal solution $\xb^\star$ by using our methods, via distributive SD/SSD and iterative sketching. What we present also accommodates regularizers of the form $\gamma\|\xb\|_2^2$, though to simplify the presentation of our expressions, we only consider \eqref{x_star_lr}.

\subsection{Steepest Descent}

When considering a minimization problem with a convex differentiable objective function $L\colon\R^d\to\R$, we select an initial $\xb^{[0]}\in\R^d$ and repeat at iteration $s+1$:
$$ \xb^{[s+1]}\gets\xb^{[s]}-\xi_s\cdot\nabla_{\xb}L(\xb^{[s]}) $$
for $s=0,1,2,\ldots$, until a prespecified termination criterion is met. This iterative procedure is called steepest descent, also known as gradient descent. The parameter $\xi_s>0$ is the corresponding step-size, which may be adaptive or fixed. To guarantee convergence of $L_{ls}$, one can select $\xi_s=2/\sigma_{\text{max}}(\Ab)^2$ for all iterations, though this is too conservative.

\subsection{Leverage Scores}

In high dimensional least squares problems subsampling of the rows of $\Ab$ is often used to reduce computational load. One commonly used method for subsampling uses what are known as \textit{leverage scores} \cite{MMY15,DMMW12}  which characterize the importance of the data points (rows of $\Ab$). Common applications of leverage scores include: approximation by low rank \cite{DM08,MD09}, approximate matrix multiplication \cite{Neo23}, tensor products and approximations \cite{FGF20,MS21,WZ22}, matrix completion \cite{HLD20}, kernel ridge regression \cite{AM15}, $k$-means clustering \cite{BDM09}, graph sparsification \cite{SS11}, and Nystr\"{o}m-based low-rank  approximations \cite{MM17}. The leverage scores of $\Ab$ measure the extent to which the vectors of its orthonormal basis $\Ub$ are correlated with the standard basis, and define the key structural non-uniformity that underlies fast randomized matrix algorithms. Leverage scores are defined as $\ell_i\coloneqq\|\Ub_{(i)}\|_2^2$, and are agnostic to any particular basis, as they are equal to the diagonal entries of the projection matrix $P_\Ab=\Ab\Ab^\dagger=\Ub\Ub^\top$. The normalized leverage scores $\pi_{\{N\}}$ of $\Ab$:
\begin{equation}
\label{norm_lvg_sc}
  \pi_i \coloneqq \left\|\Ub_{(i)}\right\|_2^2\big/\|\Ub\|_F^2 = \left\|\Ub_{(i)}\right\|_2^2\big/d
\end{equation}
for each $i\in\N_N$, form a sampling probability distribution, as $\sum_{i=1}^N\pi_i=1$ and $\pi_i\geqslant0$ for all $i$. This induced distribution has proven to be useful in linear regression \cite{DMMW12,Woo14,Mah16,ERNM22,BDN15}.

The \textit{normalized block leverage scores}, introduced independently in \cite{CPH20a,OJXE18}, are the sum of the normalized leverage scores of the subset of rows constituting each row-block of $\Ab$ \eqref{part_matrix}. Similar to leverage scores, these scores measure the correlation between the data block's corresponding submatrices in $\Ub$ and $\Ib_N$, as well as how much a particular row-block contributes to the approximation of $\Ab$. Analogous to \eqref{norm_lvg_sc}, considering the partitioning of the dataset $\D$ according to $\K_{\{K\}}$, the normalized block leverage scores of $\Ab$ are defined as
\begin{equation}
\label{norm_block_lvg_sc}
  \Pi_l \coloneqq \left\|\Ub_{(\K_l)}\right\|_F^2\big/\|\Ub\|_F^2 = \left\|\Ub_{(\K_l)}\right\|_F^2\big/d = \sum_{j\in\K_l}\pi_j
\end{equation}
for each $l\in\N_K$. This is a direct generalization of \eqref{norm_lvg_sc}, as 
\begin{equation*}
  \|\Ub_{(\K_l)}\|_F^2 = \left\|\big[(\Ub_{(\K_l)})_{(1)}\ \dots \ (\Ub_{(\K_l)})_{(|\K_l|)}\big]\right\|_2^2 .
\end{equation*}
Much like individual leverage scores, block leverage scores measure how much a particular block contributes to the approximation of the matrix $\Ab$ \cite{OJXE18}.

A related notion is that of the \textit{Frobenius block scores}, which, in the case of a partitioning as in \eqref{part_matrix}, are $\|\Ab_\iota\|_F^2$ for each $\iota\in\N_K$. Such scores have been used for approximate matrix multiplication, e.g., $CR$-multiplication \cite{DK01,DKM06a,DKM06b}. In our context, the block leverage scores of $\Ab$ are the Frobenius block scores of $\Ub$.

A drawback of using leverage scores, is that calculating them requires $O(Nd^2)$ time. To alleviate this, one can instead settle for relative-error approximations which can be approximated much faster, e.g., \cite{DMMW12} does so in $O(Nd\log N)$ time. In particular, one can consider approximate normalized scores $\Pit_{\{K\}}$ where $\Pit_i\geqslant \beta\Pi_i$ for all $i$, for some misestimation factor $\beta\in(0,1]$. Since $\Pi_{\{K\}}$ and $\Pit_{\{K\}}$ are identical if and only if $\beta=1$, a higher $\beta$ implies the approximate distribution is more accurate. We specify that a misestimation factor $\beta$ to $\Pi_{\{K\}}$ is for a specific approximate distribution $\Pit_{\{K\}}$, by denoting it as $\beta_{\Pit}$.

Approximate block sampling distributions to $\Pi_{\{K\}}$ are denoted by $\Pit_{\{K\}}$, and the distributions induced through \textit{expansion networks}; which we introduce in Subsection \ref{exp_netw_sec}, are denoted by $\Pib_{\{K\}}$. We quantify the difference between distributions $\Pi_{\{K\}}$ and $\Pit_{\{K\}}$ by the distortion metric $d_{\Pi,\Pit}\coloneqq\frac{1}{K}\sum_{i=1}^K|\Pi_i-\Pit_i|$, which is the $\ell_1$ distortion between $\Pi_{\{K\}}$ and $\Pit_{\{K\}}$ \cite{IWN22}.

\subsection{Subspace Embedding}

Our approach to approximating \eqref{x_star_lr}, is to apply an $\ell_2$-subspace embedding sketching matrix $\Sb\in\R^{r\times N}$ on $\Ab$. Recall that $\Sb\in\R^{r\times N}$ is a $(1\rpm\epsilon)$ $\ell_2$-\textit{subspace embedding} of the column-space of $\Ab$, if 
\begin{equation}
\label{subsp_emb_Ax}
  \Pr\big[(1-\epsilon)\|\Ab\xb\|_2^2 \leqslant \|\Sb\Ab\xb\|_2^2 \leqslant (1+\epsilon)\|\Ab\xb\|_2^2\big]\geqslant 1-\delta
\end{equation}
for all $\xb\in\R^d$, where $\delta\geqslant0$ is the failure probability \cite{Woo14}. Notice that such an $\Sb$, is also a $(1\rpm\epsilon)$ $\ell_2$-subspace embedding of $\Ub$, as $\{\Ab\xb:\xb\in\R^{d}\}=\{\Ub\yb:\yb\in\R^d\}$. This implies that the event whose probability is bounded in \eqref{subsp_emb_Ax} is equivalent to simultaneously satisfying
\begin{equation}
\label{subsp_emb_lower}
  (1-\epsilon)\|\yb\|_2^2 = (1-\epsilon)\|\Ub\yb\|_2^2 \leqslant \|\Sb\Ub\yb\|_2^2 
\end{equation}
and
\begin{equation}
\label{subsp_emb_upper}
  (1+\epsilon)\|\yb\|_2^2 = (1+\epsilon)\|\Ub\yb\|_2^2 \geqslant \|\Sb\Ub\yb\|_2^2
\end{equation}
for all $\yb\in\R^{d}$. The lower \eqref{subsp_emb_lower} and upper \eqref{subsp_emb_upper} bounds on $\|\Sb\Ub\yb\|_2^2$ respectively imply
\begin{equation*}
  \yb^\top\big(\Ib_d-(\Sb\Ub)^\top\Sb\Ub\big)\yb \leqslant \epsilon\|\yb\|_2^2
\end{equation*}
and
\begin{equation*}
  \yb^\top\big((\Sb\Ub)^\top\Sb\Ub-\Ib_d\big)\yb \leqslant \epsilon\|\yb\|_2^2
\end{equation*}
thus, a simplified condition for an $\ell_2$-subspace embedding of $\Ab$ is
\begin{equation}
\label{eq_form}  
  \Pr\big[\|\Ib_d-\Ub^\top\Sb^\top\Sb\Ub\|_2\leqslant\epsilon\big]\geqslant 1-\delta
\end{equation}
for a small $\delta\geqslant0$. Further details on the derivation of \eqref{eq_form} can be found in Appendix \ref{notation_app}.

For the overdetermined system $\Ab\xb=\bb$, we require $r>d$, and in the sketch-and-solve paradigm the objective is to determine an $\xbh$ that satisfies
\begin{equation}
\label{appr_xbh}  
  (1-\epsilon)\|\Ab\xb^\star-\bb\|_2 \leqslant \|\Ab\xbh-\bb\|_2 \leqslant (1+\epsilon)\|\Ab\xb^\star-\bb\|_2
\end{equation}
where $\xbh$ is an approximate solution to the modified least squares problem
\begin{equation}
\label{mod_OLS}  
  \xbh = \argmin_{\xb\in\R^d} \Big\{L_\Sb(\Sb,\Ab,\bb;\xb)\coloneqq\|\Sb(\Ab\xb-\bb)\|_2^2\Big\}.
\end{equation}
If \eqref{eq_form} is met, we get with high probability the approximation characterizations:
\begin{enumerate}
  \item $\|\Ab\xbh-\bb\|_2 \leqslant \frac{1+\epsilon}{1-\epsilon}\|\Ab\xb^\star-\bb\|_2 \leqslant (1+O(\epsilon))\|\Ab\xb^\star-\bb\|_2$
  \item $\|\Ab(\xb^\star-\xbh)\|_2\leqslant\epsilon\|(\Ib_N-\Ub\Ub^\top)\bb\|_2=\epsilon\|\bb^{\perp}\|_2$
\end{enumerate}
where $\bb^\perp=\bb-\Ab\xb^\star$ is orthogonal to the column-span of $\Ab$, i.e., $\Ab^\top\bb^\perp=\bold{0}_{d\times 1}$.

\subsection{Coded Computing Probabilistic Model}
\label{CC_model}

In this subsection, we describe the GC problem and CC probabilistic model we will be considering, which were respectively introduced in \cite{TLDK17} and \cite{LLPPR17}. In GC (Figure \ref{CG_schematic}), there is a central server that shares the $K$ disjoint subsets $\D_{\{K\}}$ of $\D$ among $m$ \textit{homogeneous}\footnote{This means that they have the same computational power, independent and identically distributed statistics for the computing time of similar tasks, and expected response time.}
servers, to facilitate computing the solution of the following minimization problem, of the form \eqref{x_star_lr}:
\begin{equation*}
\label{sep_ls_obj}  
  \xb^\star = \arg\min_{\xb\in\R^d}\bigg\{\|\Ab\xb-\bb\|_2^2=\sum_{j=1}^KL_{ls}(\D_j;\xb)\bigg\}.
\end{equation*}
Since the objective function $L_{ls}(\Ab,\bb;\xb)$ is  differentiable and additively separable, it follows that $g^{[s]}=\sum_{j=1}^Kg_j^{[s]}$. The objective function's gradient is updated in a distributed manner while only requiring $q$ servers to respond, i.e., it is robust to $m-q$ stragglers. This is achieved through an encoding of the computed partial gradients by the servers, and a decoding step once $q$ servers have sent back their encoded computation. In our approach, we consider \textit{approximate} GC by requesting updated partial gradients according to a sketched version of the objective problem given in \eqref{mod_OLS}.

The probabilistic computational model introduced in \cite{LLPPR17}, which is the standard CC paradigm, is central to our framework. In this model one assumes the existence of a \textit{mother runtime distribution} $F(t)$, with a corresponding probability density function $f(t)$. Let $T_0$ be the time it takes a single machine to complete its computation, and define the cumulative distribution function $F(t)\coloneqq\Pr[T_0\leqslant t]$. When the computation budget is $\tau$, and the computations are partitioned into $K$ subtasks and distributed, each of size $\frac{\tau}{N}$, the runtime distribution of each subtask is assumed to follow the distribution $\Ft(t)\coloneqq F(t\tau/N)=\Pr[T_i\leqslant t]$, where $T_i$ is the random completion time of the $i^{th}$ subtask. Moreover, since the $m$ servers are homogeneous and the subtasks are of the same size, the distribution is the same for each $i\in\N_K$. In this work, we view the computations as being communicated to the central server over erasure channels, where the $l^{th}$ server $W_l$ has an erasure probability\footnote{This is also known as the \textit{survival function}: $\phi(t)=\int_{t}^{\infty}\big(1-\tilde{f}(u)\big)du=1-\Ft(t)=\Pr[T_i>t]$, for $\tilde{f}(t)$ the PDF corresponding to $\Ft(t)$. The function $\phi(t)$ is monotonically decreasing.}
\begin{equation}
\label{eras_prob}
   \phi(t)\coloneqq1-\Ft(t)=1-\Pr\big[W_l \text{ responds by time } t\big]
\end{equation}
i.e., the probability that $W_l$ is a straggler at time $t$ is $\phi(t)$. All servers have the same erasure probability, as we are assuming they are homogeneous.

In our setting, there are two hyperparameters required for determining an expansion network. First, one needs to determine a time instance $t\gets T$ after which the central server will stop receiving servers' computations. This may be decided by factors such as the system's limitations, number of servers, or an upper bound on the desired waiting time for the servers to respond. At time $T$, according to $\Ft(t)$, the central server receives roughly $q(T)\coloneqq\floor{\Ft(T)\cdot m}$ server computations, where $m$ is the total number of computational nodes in the network. We refer to the prespecified time instance $T$ after which the central server stops receiving computations, as the \textit{ending time}. If the sketching procedure of the proposed sketching algorithm were to be carried out by a single server, there would be no benefit in setting $T$ such that $q(T)\tau>N$, as the exact calculation could have taken place in the same amount of time. In distributed networks though there is no control over which servers respond, and it is not a major concern if $q(T)\tau$ is slightly above $N$, as this still results in acceleration of the computation. The trade-off between accuracy and waiting time $t$ is captured in Theorem \ref{subsp_emb_thm_lvg}, for $q\gets q(t)$ sampling trials. The second hyperparameter we need in order to design an expansion network is the block size $\tau$, which is determined by the number of partitions $K$ illustrated in \eqref{part_matrix}. Together, $q(T)$ and $\tau$ determine the ideal number of servers needed for our framework to perfectly emulate sampling according to the normalized block leverage scores $\Pi_{\{K\}}$.

\section{Coded Computing from RandNLA}
\label{cc_fr_rnla_sec}

In this section, we introduce our \textit{block} leverage score sampling algorithm, which is more practical and can be carried out more efficiently than its vector-wise counterpart. Our $\ell_2$-subspace embedding result is presented in Theorem \ref{subsp_emb_thm_lvg}. By setting $\tau=1$ and $\beta=1$, we recover a known result for (exact) leverage score sampling.
\begin{figure*}[t]
  \centering
  \includegraphics[scale=.16]{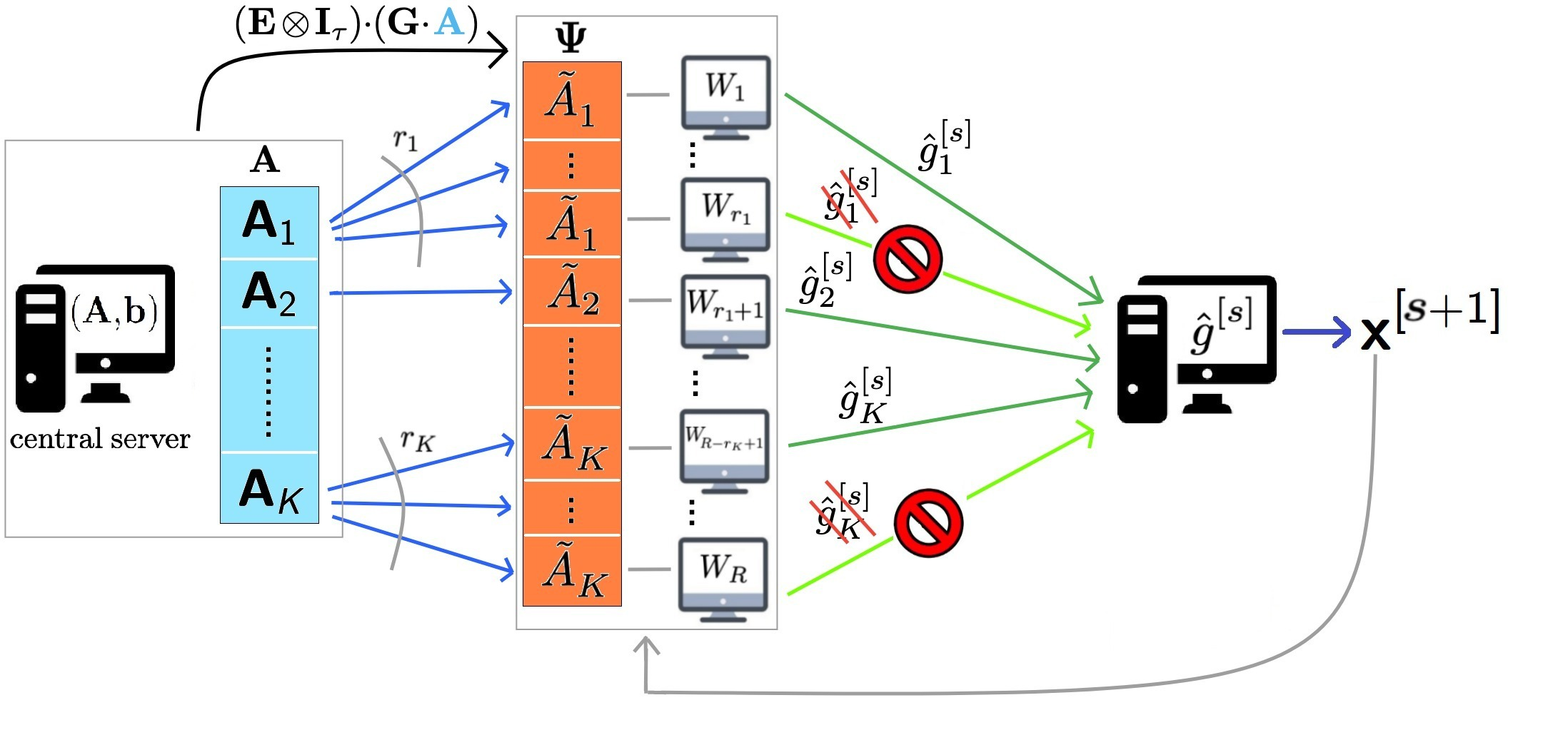}
  \caption{Illustration of our GC approach, at iteration $s+1$. The blocks of $\Ab$ (and $\bb$) are encoded through $\Gb$ and then replicated through $\Eb\otimes\Ib_\tau$, where each block of the resulting $\PsiB$ is given to a single server, before the iterative SD procedure takes place. At the illustrated iteration, servers $W_{r_1}$ and $W_{R}$ are stragglers, and their computations are not received. The central server determines the estimate $\gh^{[s]}$, and then shares $\xb^{[s+1]}$ with all the computational nodes. The resulting estimate is the gradient of the induced sketch, i.e., $\gh^{[s]}=\nabla_\xb L_\Sb(\Sbwt_{[s]},\Ab,\bb;\xb^{[s]})$.}
  \label{GC_block_lvg_schematic}
\end{figure*}

In Subsection \ref{exp_netw_sec} we incorporate our block sampling algorithm into the CC probabilistic model described in Subsection \ref{CC_model}, in which we leverage task redundancy to mitigate stragglers. Specifically, we show how to replicate computational tasks among the servers, under the integer constraints imposed by the physical system and the desired waiting time. This is then used to approximate the gradient at each iteration in a way that emulates the sampling procedure of the sketching method presented in Algorithm \ref{alg_1_pseudocode}. In Subsection \ref{opt_ind_distr_subsec} we further elaborate on when a perfect emulation is possible, and how emulated block leverage score sampling can be improved when it cannot be done perfectly, through the proposed networks. In Subsection \ref{sk_lin_regr} we present our GC approach and relate it to SD and SSD, which in turn implies convergence guarantees with appropriate step-sizes. In Subsection \ref{sk_lin_regr}, we also provide an algorithmic description of our GC procedure, to tie together the algorithms discussed in this section. Furthermore, at each iteration we have a different induced sketch, hence our procedure lies under the framework of iterative sketching. Specifically, we obtain gradients of multiple sketches of the data
$\big(\Sbwt_{[1]}\Ab,\Sbwt_{[2]}\Ab,\ldots,\Sbwt_{[n]}\Ab\big)$ and iteratively refine the solution, where $n$ can be chosen logarithmic in $N$. A schematic of our approach is presented in Figure \ref{GC_block_lvg_schematic}, and in Appendix \ref{app_example} we provide a concrete example of the induced sketching matrices resulting from the iterative process.

\subsection{Block Leverage Score Sampling}
\label{lvg_sec}

In the leverage score sketch \cite{DMMW12,ERNM22,MMY15,Mah16,Woo14} we sample with replacement $r$ rows according $\pi_{\{N\}}$ \eqref{norm_lvg_sc}, and then rescale each sampled row by $1/\sqrt{r\pi_{i_j}}$, where $\pi_{i_j}$ is the sampling probability of the row that was selected at the $j^{th}$ trial. Instead, we sample with replacement $q$ blocks from \eqref{part_matrix} according to $\Pi_{\{K\}}$ \eqref{norm_block_lvg_sc}, and rescale them by $1/\sqrt{q\Pi_{i_j}}$. The pseudocode of the block leverage score sketch is given in Algorithm \ref{alg_1_pseudocode}, where we consider an approximate distribution $\Pit_{\{K\}}$ such that $\Pit_i\geqslant \beta\Pi_i$ for all $i$, for $\beta\in(0,1]$ a dependent loss in accuracy \cite{DMMW12,DMM06,Mah16}. The spectral guarantee of the sketching matrix $\Sbwt$ of Algorithm \ref{alg_1_pseudocode} is presented in Theorem \ref{subsp_emb_thm_lvg}. Iterative sketching in our distributed GC approach through Algorithm \ref{alg_1_pseudocode}, corresponds to selecting a new sampling matrix $\Ombwt^{[s]}$ at each iteration through the servers' responses, i.e., $\Sbwt_{[s]}=\Dbwt\cdot\Ombwt^{[s]}$ for each $s$.

\begin{algorithm}[h]
\caption{Block Leverage Score Sketch}
\label{alg_1_pseudocode}
\SetAlgoLined
{\small
  \KwIn{$\Ab\in\R^{N\times d}$, $\tau=\frac{N}{K}$, $q=\frac{r}{\tau}>\frac{d}{\tau}$}
  \KwOut{$\Sbwt\in\R^{r\times N}$, $\widehat{\Ab}\in\R^{r\times d}$}
  \textbf{Initialize:} $\Omb=\bold{0}_{q\times K}$, $\Db=\bold{0}_{q\times q}$\\
  \textbf{Compute:} (approximate) distribution $\Pit_{\{K\}}$ \eqref{norm_block_lvg_sc}\\
  \For{$j=1$ to $q$}
    {
      sample with replacement $i_j$ from $\N_K$, according to $\Pit_{\{K\}}$\\
      $\Omb_{j,i_j}=1$ \Comment{equivalently $\Omb_{(j)}=\eb_{i_j}^\top$}\\
      $\Db_{j,j}=\sqrt{\frac{\tau}{r\Pit_{i_j}}}=\sqrt{\frac{1}{q\Pit_{i_j}}}$\\
    }
  $\Ombwt\gets\Omb\otimes\Ib_\tau$\\
  $\Dbwt\gets\Db\otimes\Ib_\tau$\\
  $\Sbwt\gets\Dbwt\cdot\Ombwt$ \Comment{$\Sbwt=(\Db\cdot\Omb)\otimes\Ib_{\tau}$}\\
  $\widehat{\Ab}\gets\Sbwt\cdot\Ab$\\
}
\end{algorithm}

\begin{Thm}
\label{subsp_emb_thm_lvg}
The sketching matrix $\Sbwt$ of Algorithm \ref{alg_1_pseudocode} is a $(1\rpm\epsilon)$ $\ell_2$-subspace embedding of $\Ab$, according to \eqref{eq_form}. Specifically, for $\delta>0$ and $q=\Theta\left(\frac{d}{\tau}\log{(2d/\delta)}/(\beta\epsilon^2)\right)$:
\begin{equation*}
  \Pr\big[\|\Ib_d-\Ub^\top\Sbwt^\top\Sbwt \Ub\|_2\leqslant\epsilon\big]\geqslant 1-\delta.
\end{equation*}
\label{alg_lvg_b}
\end{Thm}

\begin{proof}
The main tool is a matrix Chernoff bound \cite[Fact 1]{Woo14}. We define random matrices corresponding to the sampling process and bound their norm and variance, in order to apply the matrix Chernoff bound. The complete proof can be found in Appendix \ref{lvg_app}.
\end{proof}

Theorem \ref{subsp_emb_thm_lvg} applies to more general problems than leverage score sampling. Specifically, one can apply a random projection to ``flatten'' the block leverage scores, i.e., to make them all approximately equal, and then sample with replacement uniformly at random. This is the main idea behind the analysis of the SRHT \cite{AC06,AC09}. The trade-off between such algorithms and Algorithm \ref{alg_1_pseudocode}, is computing the leverage scores explicitly vs. the computational reduction due to applying a random projection. Sketching approaches similar to the SRHT which do not directly utilize the data, are referred to as ``data oblivious sketches''. Theses approaches are better positioned for handling high velocity streams as well as highly unstructured and arbitrarily distributed data \cite{MOW21}. Multiplying the data by a random matrix spreads the information in the rows of the matrix, such that all rows are of equal importance, and the new matrix is ``incoherent''. In Appendix \ref{comp_SRHT_app}, we show when Algorithm \ref{alg_1_pseudocode} and the corresponding block sampling counterpart of the SRHT \cite{CMPH22} achieve the same asymptotic guarantees for the same number of sampling trials $q$.

Next, we provide a sub-optimality result for non-iterative sketching of the block leverage score sketch for ordinary least squares
\begin{equation}
\label{nonit_OLS}
  \xbt = \argmin_{\xb\in\R^d} \Big\{L_\Sb(\Sbwt,\Ab,\bb;\xb)\coloneqq\|\Sbwt(\Ab\xb-\bb)\|_2^2\Big\}
\end{equation}
which follows from the results of \cite{PW16}. The following Corollary is established by adapting the proof of {\cite[Theorem 1]{PW16}}, which is based on a reduction from statistical minimax theory combined with information-theoretic bounds and an application of Fano’s inequality, yielding an upper bound to $\big\|\E\big[\Sbwt^\top\big(\Sbwt\Sbwt^\top\big)^{-1}\Sbwt\big]\big\|_2$. The proof of Corollary \ref{sub_opt_cor} can be found in Appendix \ref{lvg_app}. 

\begin{Cor}
\label{sub_opt_cor}
For any full-rank data matrix $\Ab\in\R^{N\times d}$ with a noisy observation model $\bb=\Ab\xb^\bullet+\wb$ where $\xb^\bullet$ is an arbitrary vector of length $d$ and $\wb\sim\Nn\left(0,\sigma^2\Ib_N\right)$, the optimal least squares solution $\xb^\star$ of \eqref{x_star_lr} has prediction error $\E\big[\|\Ab(\xb^\bullet-\xb^\star)\|_2^2\big]\lesssim\frac{\sigma^2d}{N}$. Furthermore, any estimate $\xbt$ based on the sketched system $(\Sbwt\Ab,\Sbwt\bb)$ produced by Algorithm \ref{alg_1_pseudocode} with sampling probabilities $\Pi_{\{K\}}$, has a prediction error lower bound of
\begin{equation}
\label{sub_opt_error}
  \E\big[\|\Ab(\xb^\bullet-\xbt)\|_2^2\big]\gtrsim\frac{\sigma^2d}{\min\{r,N\}} .
\end{equation}
\end{Cor}

Even though Corollary \ref{sub_opt_cor} considers sampling according to the exact block leverage scores, its proof can be modified to accommodate approximate sampling. Additionally, the above corollary holds for constrained least squares, though we do not explicitly state it, as this is not a focus of the work presented in this paper. From \eqref{sub_opt_error}, it is clear that for a smaller $r$ with $r<N$ the proposed approach produces a less efficient sketch and poorer approximation $\xbt$, though when considering a higher $r$ approaching $N$, there is an improvement in the accuracy of $\xbt$ at the cost of a higher computational load.

\subsection{Expansion Networks}
\label{exp_netw_sec}

The framework we propose emulates the sampling with replacement procedure of Algorithm \ref{alg_1_pseudocode} in distributed CC environments. Even though we focus on $\ell_2$-subspace embedding and descent methods in this paper, the proposed framework applies to any matrix algorithm which utilizes importance sampling with replacement. In contrast to other CC schemes in which RandNLA was used to compress the network, e.g., \cite{CPH20a,CPH20c,RCHV23}, here the networks are expanded according to $\Pi_{\{K\}}$. That is, the computations corresponding to the row-blocks of $\Ab$ are replicated by the servers proportionally to $\Pi_{\{K\}}$. It is unlikely that we can exactly emulate this distribution, as the number of replications per task need to be integers. Instead, we mimic the exact probabilities with an induced distribution $\Pib_{\{K\}}$ through expansion networks, which are determined by $\Ft(t)$ at a prespecified $T$, after which the central server stops receiving computations for that iteration.

The proposed method solves the minimization problem \eqref{appr_opt_sol}, whose approximate solution $\rh_{\{K\}}$ \eqref{appr_distr} determines the number of replicas $r_{\{K\}}$ of each block in our expansion network. We note that \eqref{appr_opt_sol} is a surrogate to the integer program \eqref{int_program}, whose solution can achieve an accurate realizable distribution $\Pib_{\{K\}}$ to $\Pi_{\{K\}}$ through the distributed network, by appropriately replicating the blocks. Unfortunately, the integer program \eqref{int_program} is not always solvable. Nonetheless, when we have an approximation to \eqref{appr_opt_sol} or \eqref{int_program}, through uniform sampling we can minimize with high probability the $\ell_2$-subspace embedding condition for $\Ab$ \eqref{eq_form}, up to a small error, given the integer constraints imposed by the physical system --- $r_i\in\Z_+$ for all $i$ and $R=\sum_{l=1}^Kr_l$ such that $R\approx m$. In the CC context we want $m=R$, i.e., the total number of replicated blocks is equal to the number of servers. Next, we describe the desired induced distribution $\Pib_{\{K\}}$, in order to set up the optimization problem \eqref{appr_opt_sol}.

Assume without loss of generality that $\Pi_j\leqslant\Pi_{j+1}$ for all $j\in\N_{K-1}$, thus $r_j\leqslant r_{j+1}$. The sampling distribution through the expansion network translates to
\begin{equation}
\label{exp_netw_distr}
  \Pib_i \coloneqq \Pr\left[\At_i \text{ is sampled}\right] = r_i/R \approx \Pi_i
\end{equation}
for all $i\in\N_K$ and $R=\sum_{l=1}^Kr_l$, where $\At_i$ is the $i^{th}$ encoded partition of our GC scheme. Our objective is to determine $r_{\{K\}}$ such that $\Pib_i\approx\Pi_i$ for all $i$. Furthermore, for an erasure probability determined by $\phi(t)$ at a specified time $t$ \eqref{eras_prob}, the probability that the computation corresponding to the $i^{th}$ block is sampled with replacement through the erasure channels at time $t$ is
\begin{equation}
\label{sample_er_channel}
  \Pr\left[\text{sample } \At_i \text{ through the channels}\right]=1-\phi(t)^{\rho_i(t)}
\end{equation}
for some $\rho_i(t)\in\R_{>0}$,\footnote{To be realizable, through replications, we need $\rho_i(t)\in\Z_+$.} and the network emulates the sampling distribution $\Pi_{\{K\}}$ exactly when
\begin{equation}
\label{lvg_score_eras_prob}
  \Pi_i=1-\phi(t)^{\rho_i(t)}
\end{equation}
for all $i$.

The replications which take place can be interpreted as the task allocation through a directed bipartite graph $G=(\mathcal{L},\mathcal{R},\mathcal{E})$, where $\mathcal{L}$ and $\mathcal{R}$ correspond to the $K$ encoded partitions $\At_{\{K\}}$ and $m$ servers respectively, where $\text{deg}(x_i)=r_i$ for all $x_i\in\mathcal{L}$ and $\text{deg}(y_j)=1$ for all $y_j\in\mathcal{R}$, with $\{x_i,y_j\}\in\mathcal{E}$ only if the $j^{th}$ server $W_j$ is assigned $\At_i$. This allocation is done only once, before the iterative procedure begins.
\begin{figure}[h]
  \centering
  \includegraphics[scale=.2]{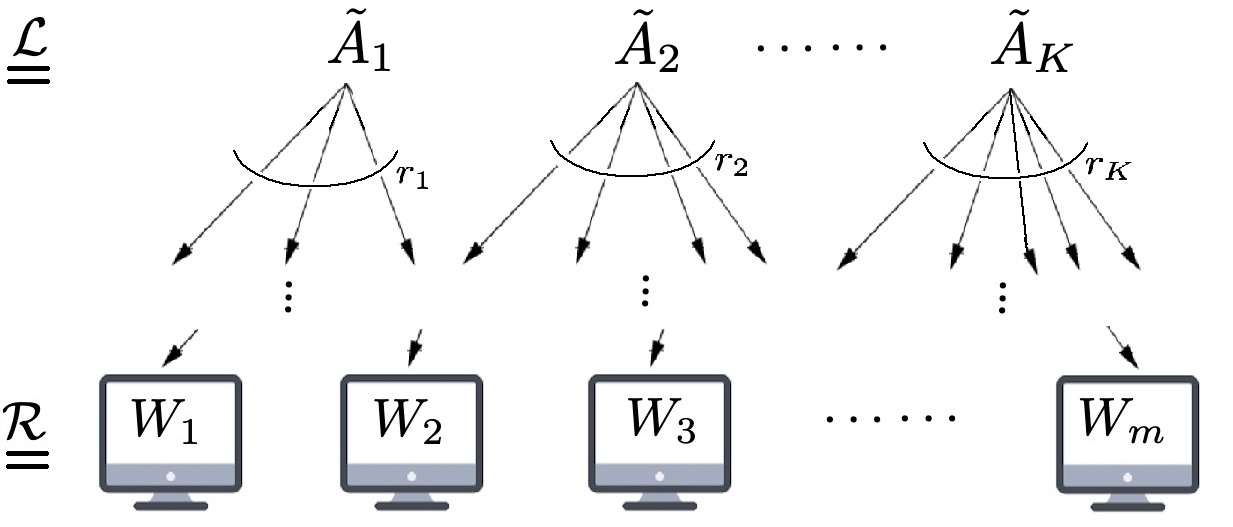}
  \caption{Depiction of an expansion network as a bipartite graph, for $m=\sum_{l=1}^Kr_l$.}
  \label{bip_exp_network}
\end{figure}

Our goal is to determine $r_{\{K\}}$, which minimize the error in the emulated distribution $\Pib_{\{K\}}$. Under the assumption that we have an integer number of replicas per block, from \eqref{sample_er_channel} and \eqref{lvg_score_eras_prob} we deduce that $\Pi_i\approx1-\phi(t)^{r_i}$ for $r_i\in\Z_+$, which lead to the minimization problem
\begin{equation}
  \arg\min_{\substack{r_{\{K\}}\subsetneq\Z_+}} \left\{\Delta_{\Pi,\Pib}\coloneqq\frac{1}{K}\sum_{i=1}^K\big|\Pi_i-\big(1-\phi(t)^{r_i}\big)\big|\right\} = \arg\left\{ \sum_{i=1}^K\min_{r_i\in\Z_+} \Big\{\big|\Pi_i-\big(1-\phi(t)^{r_i}\big)\big|\Big\} \right\} . \label{appr_opt_sol}
\end{equation}
Further note that $\Delta_{\Pi,\Pib}\equiv d_{\Pi,\Pit}$, where $\Pit_{i}=1-\phi(t)^{r_i}$ for each $i\in\N_K$. For our proposed distribution $\Pib_{\{K\}}$, we may have $d_{\Pi,\Pib}\neq\Delta_{\Pi,\Pib}$. By combining \eqref{exp_netw_distr} and \eqref{sample_er_channel}, we then solve for the approximate replications $\rh_{\{K\}}$ at time $t$:
\begin{equation}
\label{appr_distr}
  \Pib_i\approx\Pi_i=1-\phi(t)^{\rho_i(t)} \quad \implies \quad \rh_i=\left\lfloor\frac{\log(1-\Pi_i)}{\log(\phi(t))}\right\rceil=\left\lfloor\rho_i(t)\right\rceil
\end{equation}
which result in the induced distribution $\Pib_i=\rh_i/\Rh$, for $\Rh\coloneqq\sum_{l=1}^K\rh_l$. In our context, we also require that $\Rh\approx m$.

Ideally, the above procedure would result in replication numbers $\rh_{\{K\}}$ for which $\Rh=m$. This though is unlikely to occur, as $\Pi_{\{K\}}$ and $R$ are determined by the data, and $m$ is the fixed number of servers. To address this issue we redefine $\rh_{\{K\}}$ to $\rt_{\{K\}}$ by $\rt_i=\rh_i\rpm\alpha_i$ for $\alpha_i$ small integers such that $\sum_{l=1}^K\rt_l=m$, and $\sum_{l=1}^K|\Pi_l-\rt_l/m|$ is minimal. If $m\gg \Rh$ for a large enough $\tau$, we can set the number of replicas to be $\rt_i\approx\left\lfloor m/\Rh\right\rceil\cdot \rh_i$. Furthermore, the block size $\tau$ can be selected such that $\Rh$ is approximately equal to the system's parameter $m$. We address the fact that $\Rh$ is only approximately equal to $m$ in Subsection \ref{opt_ind_distr_subsec}.

\begin{Lemma}
\label{appr_err_rounding}
The approximation $\rh_{\{K\}}$ according to \eqref{appr_distr} of the minimization problem \eqref{appr_opt_sol} at time $t$, satisfies
$$ \Delta_{\Pi,\Pib}\leqslant\Big(1-\sqrt{\phi(t)}\Big)\cdot\left(\sum_{l=1}^K\phi(t)^{\min\limits_{i\in\N_K}\left\{\rh_i,\rho_i(t)\right\}}\right). $$
\end{Lemma}

\begin{proof}
We break the proof into the cases where we round $\rho_i(t)$ to both the closest integers above and below. In either case, we know that $\big(\rho_i(t)-\rh_i(t)\big)\in[-1/2,1/2]$, for each $i\in\N_K$. Denote the respective individual summands of $\Delta_{\Pi,\Pib}$ by $\Delta_i$. In the case where $r_i=\lfloor\rho_i(t)\rfloor$, we have $\rho_i(t)=\rh_i+\eta$ for $\eta\in[0,1/2]$, hence
\begin{align*}
  \Delta_i &= \left|\big(1-\phi(t)^{\rho_i(t)}\big)-\big(1-\phi(t)^{\rh_i}\big)\right| \\
  &= \left|\phi(t)^{\rh_i}-\phi(t)^{\rho_i(t)}\right| \\
  &= \left|\phi(t)^{\rh_i}-\phi(t)^{r_i+\eta}\right| \\
  &= \left|\phi(t)^{\rh_i}\cdot\big(1-\phi(t)^\eta\big)\right| \\
  &\leqslant \left|\phi(t)^{\rh_i}\cdot\big(1-\phi(t)^{1/2}\big)\right| \\
  &= \phi(t)^{\rh_i} \cdot \Big(1-\sqrt{\phi(t)}\Big) .
\end{align*}
Similarly, in the case where $r_i=\lceil\rho_i(t)\rceil$ we have $\rh_i=\rho_i(t)+\eta$ for $\eta\in[0,1/2]$, and
\begin{align*}
  \Delta_i \leqslant \phi(t)^{\rho_i(t)} \cdot \Big(1-\sqrt{\phi(t)}\Big) .  
\end{align*}
Considering all summands, it follows that
$$ \Delta_{\Pi,\Pib} = \sum_{l=1}^K \Delta_l \leqslant \sum_{l=1}^K\left(\phi(t)^{\min\limits_{i\in\N_K}\left\{\rh_i,\rho_i(t)\right\}} \cdot \Big(1-\sqrt{\phi(t)}\Big)\right) . $$
\end{proof}

We note that all terms involved in the upper bound of Lemma \ref{appr_err_rounding} are positive and strictly less than one. Furthermore, for a larger $t$ we have a smaller $\phi(t)$, while for a smaller $t$ we have a smaller $\rh_i$ for each $i$. This bound further corroborates the importance of the hyperparameter $t$ and the distribution $F(t)$, in designing expansion networks.

The replication of blocks which takes place, can be described through a corresponding $R\tau\times K\tau$ binary \textit{expansion matrix}:
\begin{equation}
\label{expansion_matrix}
  \Ebwt = \Eb\otimes\Ib_\tau = \begin{pmatrix} \bold{1}_{r_1\times 1} & & & \\ & \bold{1}_{r_2\times 1} & & \\ & & \ddots & \\ & & & \bold{1}_{r_K\times 1} \end{pmatrix} \otimes\Ib_\tau
\end{equation}
where $\Eb\in\{0,1\}^{R\times K}$ is the adjacency matrix of the bipartite graph $G$ portrayed in Figure \ref{bip_exp_network} (up to a permutation of the rows/server indices). It follows that $(\Ebwt\cdot\Ab,\Ebwt\cdot\bb)$ are comprised of replicated blocks of the partitioning in \eqref{part_matrix}, with replications according to $r_{\{K\}}$.

For the proposed networks, the multiplicative misestimation factor in Theorem \ref{subsp_emb_thm_lvg} is $\beta_{\Pib}=\min_{i\in\N_K}\{\Pi_i/\Pib_i\}\leqslant1$. In the case where $\Rt=\sum_{l=1}^K\rt_l>m$ and $\Pit_i\coloneqq\rt_i/\Rt$, Algorithm \ref{alg_replicas_mod_R} takes $\rt_{\{K\}}$ as an input and determines $r_{\{K\}}$ such that $R=\sum_{l=1}^Kr_l=m$. The updated distribution $\Pib_{\{K\}}$ where $\Pib_i=r_i/R$ for each $i$, also has a more accurate misestimation factor, i.e., $\beta_{\Pib}>\beta_{\Pit}$. To establish sampling guarantees in relation to $d_{\Pi,\Pib}$, one would need to invoke an additive approximation error to the scores, i.e., $\Pit_i\leqslant \Pi_i+\epsilon$ for all $i$ where $\epsilon\geqslant0$ is a small constant \cite{CLMMPS15,DMMW12}. In our distributed networks, the additive error would be $\epsilon_{\Pib}=\max_{i\in\N_K}\{|\Pi_i-\Pib_i|\}$.

\subsection{Optimally Induced Distributions}
\label{opt_ind_distr_subsec}

Recall that \eqref{appr_opt_sol} is a surrogate to the integer program
\begin{equation}
\label{int_program}
  r^{\star}_{\{K\}}=\arg\min_{\substack{r_1,\ldots,r_K\in\Z_+\\ R=\sum_{l=1}^Kr_l}} \bigg\{d_{\Pi,\Pis}=\frac{1}{K}\sum_{i=1}^K\big|\Pi_i-\overbrace{r_i/R}^{\Pis_i}\big|\bigg\}
\end{equation}
where $R$ approximates the total number of servers. The solution $r^{\star}_{\{K\}}$ approximates the block leverage score distribution $\Pi_{\{K\}}$ through expansion networks. Similar to $\Delta_{\Pi,\Pib}$ from \eqref{appr_opt_sol}, the distortion metric $d_{\Pi,\Pis}$ is a measure of closeness between the distributions $\Pi_{\{K\}}$ and $\Pib_{\{K\}}$, under the network imposed constraints. In the case where $\Pi_{\{K\}}\subsetneq(0,1)\backslash\Q_+$, i.e., $\Pi_{\{K\}}$ are not necessarily all rational, the integer constraints of the physical network may make $\{r_i^\star/R\}_{i=1}^K$ differ from $\Pi_{\{K\}}$. The integer program \eqref{int_program} cannot be solved exactly when $\Pi_{\{K\}}\subsetneq(0,1)\backslash\Q_+$, as we can always get finer approximations, e.g., through a continued fraction approximation. This is specific to the ending time $T$, which is not made explicit in \eqref{int_program}, when considering erasures over the communication channels according to \eqref{eras_prob}, which $T$ we do not include in \eqref{int_program}, to simplify notation. Furthermore, \eqref{int_program} can also be considered for centralized distributed settings which differ from the system model proposed in \cite{LLPPR17}. We note that solvers of \eqref{int_program} exist when $R$ is fixed and we remove the constraint $R=\sum_{l=1}^Kr_l$. The proof of Corollary \ref{cor_exact_sol} below is a constructive solution of \eqref{int_program} when $\Pi_{\{K\}}\subsetneq[0,1]\cap\Q_+$, in which case $\{r_i^\star/R\}_{i=1}^K$ can be made exactly equal to $\Pi_{\{K\}}$, a condition we call \textit{perfect emulation}.

\begin{Prop}
\label{perf_emulation_prop}
A perfect emulation occurs when $d_{\Pi,\Pis}=0$. This is possible if and only if $\Pi_i\in[0,1]\cap\Q_+$ for all $i$ and the denominators of $\Pi_{\{K\}}$ in reduced form are factors of $R$, i.e., $R\cdot\Pi_i\in\Z_+$.
\end{Prop}

\begin{proof}
If $d_{\Pi,\Pis}=0$, then $\Pi_i=\Pis_i$ for all $i\in\N_K$, thus $\Pi_{\{K\}}$ and $\Pis_{\{K\}}$ are the same sampling distributions. For the reverse direction, assume that for all $i$ we have $\Pi_i=a_i/b_i$ for coprime integers $a_i,b_i\in\Z_+$, and that $R=\mu_i b_i$ for some $\mu_i\in\Z_+$; thus $R\cdot\Pi_i=\mu_i a_i\in\Z_+$. Let $r_i=\mu_i a_i$. It follows that $\Pis_i=\frac{r_i}{R}=\frac{\mu_i a_i}{\mu_i b_i}=\frac{a_i}{b_i}=\Pi_i$ for all $i$, hence $d_{\Pi,\Pis}=0$. Now, assume for a contradiction that there is a $j\in\N_K$ for which $\Pi_j\in(0,1)\backslash\Q_+$. Then, by definition, $\Pi_j$ cannot be expressed as a fraction $r_j/R$ for $r_j,R\in\Z_+$, thus $d_{\Pi,\Pis}\geqslant|\Pi_j-\Pis_j|>0$.
\end{proof}

\begin{Cor}
\label{cor_exact_sol}
When $\Pi_{\{K\}}\subsetneq[0,1]\cap\Q_+$, we can solve \eqref{int_program}, so that $d_{\Pi,\Pis}=0$.
\end{Cor}

\begin{proof}
From Proposition \ref{perf_emulation_prop}, the smallest $R$ in order to attain $r_{\{K\}}$ for which $d_{\Pi,\Pis}=0$ when considering $\Pi_i=a_i/b_i$ in reduced form, is the least common multiple $R=\lcm(b_1,\ldots,b_K)$. For each $i\in\N_K$, we then have $R=\mu_i b_i$ for $\mu_i\in\Z_+$, and $r_i=\mu_i a_i$. Hence $\Pis_i=\frac{r_i}{R}=\frac{\mu_i a_i}{\mu_i b_i}=\frac{a_i}{b_i}=\Pi_i$, for which $d_{\Pi,\Pis}=0$.
\end{proof}

\begin{Lemma}
If for a set of integers $\rt_{\{K\}}$ we have $\Rt=\sum_{l=1}^K\rt_l$, $m=\Rt$, and $\Pit_i=\rt_i/\Rt$ for all $i\in\N_K$, then:
\begin{equation*}
\label{bds_distortion_lower}
  d_{\Pi,\Pit} \geqslant \frac{1}{m}\cdot\min_{i\in\N_K}\left\{\bfloor{|m\cdot\Pi_i-\rt_i|}\right\} 
\end{equation*}
and
\begin{equation}
\label{bds_distortion_upper}
  d_{\Pi,\Pit} \leqslant \frac{1}{m}\cdot\max_{i\in\N_K}\left\{\bceil{|m\cdot\Pi_i-\rt_i|}\right\}.
\end{equation}
\end{Lemma}

\begin{proof}
Let $\dt_i=|\Pi_i-\Pit_i|=|\Pi_i-\rt_i/m|$ for each $i$, hence
\begin{equation*}
  \rt_L \coloneqq \frac{1}{m}\cdot\min_{i\in\N_K}\left\{\bfloor{|m\cdot\Pi_i-\rt_i|}\right\} \leqslant \dt_i
\end{equation*}
and
\begin{equation*}
  \rt_U \coloneqq \frac{1}{m}\cdot\max_{i\in\N_K}\left\{\bceil{|m\cdot\Pi_i-\rt_i|}\right\} \geqslant \dt_i 
\end{equation*}
for all $i\in\N_K$. By rescaling the sum over all $\dt_{\{K\}}$ by $1/K$, we get
\begin{equation*}
 \rt_L = \frac{K\cdot \rt_L}{K} \leqslant \frac{1}{K}\cdot\sum_{i=1}^K \dt_i \leqslant \frac{K\cdot \rt_U}{K} = \rt_U
\end{equation*}
which completes the proof.
\end{proof}

Next, we give a simple approximation to \eqref{int_program}, for when we do not consider the erasure channel characterization through \eqref{lvg_score_eras_prob}; nor an ending time $T$. Given $\Pi_{\{K\}}$, the replication numbers are $\rt_i = \left\lfloor\Pi_i/\Pi_1\right\rceil$ and $\Rt=\sum_{l=1}^K\rt_l$. Further note that for more accurate approximations, we can select an integer $\nu>1$ and take $\rt_i = \left\lfloor\nu\cdot\Pi_i/\Pi_1\right\rceil$. From the proof of Corollary \ref{cor_exact_sol}, it follows that if $\Pi_{\{K\}}\subsetneq[0,1]\cap\Q_+$ and $\nu=\mu_1 a_1$, we get $\rt_i = \Pi_i/\Pi_1 = \mu_i a_i\in\Z_+$, which solves \eqref{lvg_score_eras_prob}. The drawback of designing an expansion network with this solution, is that as $\nu$ increases; $\Rt$ also increases.

We drop the constraint $R=\sum_{l=1}^Kr_l$ of \eqref{int_program} and the assumption that $m=R$, by proposing a procedure in Algorithm \ref{alg_replicas_mod_R} for determining $r_{\{K\}}$ from a given set $\rt_{\{K\}}$ (e.g., those proposed in \eqref{appr_distr}) to get the induced distribution $\{\Pib_i=r_i/m\}_{i=1}^K$, where $m=\sum_{l=1}^K r_l$. In Algorithm \ref{alg_replicas_mod_R}, $\chi=1$ and $\chi=0$ indicate whether $\Rt\geqslant m$ or $\Rt<m$ respectively.

\begin{algorithm}[h]
\caption{Determine $r_{\{K\}}$ from $\rt_{\{K\}}$}
\label{alg_replicas_mod_R}  
\SetAlgoLined
{\small
  \KwIn{$m$, $\Pi_{\{K\}}$, $\rt_{\{K\}}$, $\Rt=\sum_{i=1}^K\rt_i$}
  \KwOut{replication numbers $r_{\{K\}}$}
  \textbf{Initialize:} $\dt_{\{K\}}=\left\{\dt_i\coloneqq\Pi_i-\rt_i/m\right\}_{i=1}^K$, $r_{\{K\}}=\rt_{\{K\}}$, $R=\Rt$, $\chi=1$, $\tilde{j}=0$\\
  \If{$R<m$}
    {
      $\chi\gets 0$ \Comment{$\chi$ indicates: $\Rt\geqslant m$ or $\Rt<m$}\\
      $\dt_{\{K\}}\gets \left\{-\dt_i\right\}_{i=1}^K$\\
    }
  \While{$R\neq m$}
    {
      $j\gets \arg\min\limits_{i\in\N_K}\big\{\dt_{\{K\}}\big\} \qquad (\blacklozenge)$\\
        \If{$(-1)^{\chi+1}\cdot\left(\Pi_j-\frac{1}{m}\left(r_j+(-1)^{\chi}\right)\right)>0$ \textbf{and} $\tilde{j} \equiv j$}
          {
            $\dt_j\gets 1$
          }
        \Else
          {
            $r_j\gets r_j+(-1)^{\chi}$\\
            $R\gets R+(-1)^{\chi}$\\
            $\dt_j\gets(-1)^{\chi+1}\cdot\left(\Pi_j-r_j/m\right)$
          }
       $\tilde{j}\gets j$
    }
}
\end{algorithm}

\begin{Rmk}
\label{rmk_repl_nums}
The objective of Algorithm \ref{alg_replicas_mod_R} is to reduce the upper bound \eqref{bds_distortion_upper} when $m<\Rt$, while guaranteeing that $\sum_{l=1}^Kr_l=m$. In practice, the more concerning and limiting case is when $m<\Rt$. The bottleneck of Algorithm \ref{alg_replicas_mod_R} is retrieving the index $j$ in the \textnormal{\textbf{while}} loop, which takes $O(K)$ time. In order to reduce the number of instances we solve $(\blacklozenge)$, we ensure that we only reduce the replica numbers $\rh_{\{K\}}$ in the case where $R>m$; and increase them when $R<m$, by the inner \textnormal{\textbf{if}} statement. Moreover, this is carried out once before sharing the replicated blocks. The more practical and realistic case is when $R>m$, as we can get a closer approximation with a greater $R$, and $\lcm(b_1,\ldots,b_K)$ will likely be large when $\Pi_{\{K\}}\subsetneq[0,1]\cap\Q_+$. We note that the integers $\rh_{\{K\}}$ of \eqref{appr_distr} and their sum $\Rh$ are a byproduct of the prespecified ending time $t\gets T$, the mother runtime distribution $F(t)$, and the block size $\tau$, which can be selected so that $\Rh>m$.
\end{Rmk}

\begin{Prop}
\label{prop_reduct_alg}  
Assume we are given $\rt_{\{K\}}$ (not necessarily according to \eqref{appr_distr}), for which $m<\Rt=\sum_{l=1}^K\rt_l$. Denote by $\Pit_{\{K\}}$ the corresponding sampling distribution $\left\{\Pit_i=\rt_i/\Rt\right\}_{i=1}^K$, for which $d_{\Pi,\Pit}\leqslant\tilde{\epsilon}$. Then, the output $r_{\{K\}}$ of Algorithm \ref{alg_replicas_mod_R} produces an induced distribution $\{\Pib_i=r_i/m\}_{i=1}^K$ which satisfies $d_{\Pi,\Pib}< d_{\Pi,\Pit}\leqslant\tilde{\epsilon}$.
\end{Prop}

\begin{proof}
The case where $\Rt>m$, corresponds to $\chi=1$, in which case Algorithm \ref{alg_replicas_mod_R} returns $r_{\{K\}}$ for which $r_i\leqslant\rt_i$ for all $i\in\N_K$. In this case, the optimization problem $(\blacklozenge)$ assigns to $j$ a partition index for which $\Pi_j<r_j/m$. Since $\dt_j=\Pi_j-r_j/m<0$, we have $\Pi_j<r_j/m$. The \textbf{if} statement guarantees that we did not produce an $r_j$ for which $\Pi_j>r_j/m$, when we previously had $\Pi_j<\rt_j/\Rt$ (or $\Pi_j<r_j/R$ after a reassignment of $\rt_j$). Along with the fact that $\frac{r_j-1}{R-1}<\frac{r_j}{R}$, it follows that for the updated difference $d_j'$:
\begin{align*}
  \left|d_j'\right| &= \left|\Pi_j-\frac{r_j-1}{R-1}\right|\\
  &= \frac{r_j-1}{R-1}-\Pi_j\\
  &< \frac{r_j}{R}-\Pi_j\\
  &= \left|\Pi_j-\frac{r_j}{R}\right|\\
  &= \left|\dt_j\right|
\end{align*}
i.e., at each iteration of the \textbf{else} statement, we decrease the distortion metric, thus $d_{\Pi,\Pib}< d_{\Pi,\Pit}\leqslant\tilde{\epsilon}$.

The \textbf{else} statement is carried out $\Rt-m$ times, producing $r_{\{K\}}$ for which $\sum_{l=1}^Kr_l=\Rt-(\Rt-m)=m$, hence the normalizing integer for the distribution $\Pib_{\{K\}}$ is $R=m$.
\end{proof}

By incorporating Algorithm \ref{alg_replicas_mod_R} into the proposed GC method, our procedure is more versatile in comparison to other methods when it comes to the code design parameters $m$ and $q$. Other methods require a decoding recovery threshold of $q$ responses in order to apply their decoding step, and the pre-selected parameters need to satisfy $(m-q+1)\mid m$ \cite{TLDK17,HASH17,CMH21}. In exact GC schemes, the average replication number of sub-datasets depends on the number of stragglers it can tolerate, as such codes are maximum distance separable. Specifically, to tolerate $(m-q)$ stragglers, each data block needs to be replicated at least $(m-q+1)$ times across the servers. In contrast, our method can be deployed when this is condition is not met, and does not require a recovery threshold.

\subsection{GC with Leverage Score Sampling}
\label{sk_lin_regr}
 
Next, we derive the server computations of our GC scheme, so that the central server retrieves the gradient of the modified objective function $L_\Sb(\Sbwt_{[s]},\Ab,\bb;\xb^{[s]})=\|\Sbwt_{[s]}(\Ab\xb^{[s]}-\bb)\|_2^2$:
\begin{equation}
\label{appr_gr}
  \gh^{[s]}\coloneqq\nabla_\xb L_\Sb(\Sbwt_{[s]},\Ab,\bb;\xb^{[s]})
\end{equation}
for $\Sbwt_{[s]}$ the sketching matrix according to Algorithm \ref{alg_1_pseudocode} and $\xb^{[s]}$ our current parameters update; both at iteration $s$, to iteratively approximate the solution $\xb^\star$ to \eqref{x_star_lr}. Furthermore, we show that with a diminishing step-size, our updates $\xb^{[s]}$ converge in expectation to the optimal solution $\xb^\star$, and the expected regret of the least squares objective converges to zero. First though, in Algorithm \ref{alg_SD_iter_sketching} we provide the pseudocode of how SD is performed with iterative sketching. This corresponds to the `Iterative Procedure at Iteration $s$' in Algorithm \ref{app_GC_alg}.

\begin{algorithm}[h]
\caption{Steepest Descent with Iterative Sketching}
\label{alg_SD_iter_sketching}
\SetAlgoLined
  \KwIn{$\Ab\in\R^{N\times d}$, $\bb\in\R^N$, $\xb^{[0]}\in\R^d$, $\tau$, $q=\frac{r}{\tau}>\frac{d}{\tau}$}
  \KwOut{$\xbh\in\R^d$ s.t. $\Ab\xbh\approx\bb$}
  \For{$s=0,1,2,\ldots$}
    {
    \begin{itemize}[label={--}]
      \item  run Algorithm \ref{alg_1_pseudocode}, to obtain $\Sbwt_{[s]}\in\R^{r\times N}$\\
      $\quad \triangleright$ approximate $\Pit_{\{K\}}$ {only at} $s=0$
      \item $\Abwh_{[s]}\gets\Sbwt_{[s]}\cdot\Ab$
      \item $\bbwh_{[s]}\gets\Sbwt_{[s]}\cdot\bb$
      \item {\footnotesize$\gh^{[s]} \gets \nabla_\xb L_{ls}\big(\Abwh_{[s]},\bbwh_{[s]};\xb^{[s]}\big) \equiv \nabla_\xb L_\Sb\big(\Sbwt_{[s]},\Ab,\bb;\xb^{[s]}\big)$}\\
      \item \textbf{Select:} step-size $\xi_s>0$\\
      \item \textbf{Update:} $\xb^{[s+1]}\gets\xb^{[s]}-\xi_s\cdot\gh^{[s]}$
    \end{itemize}
    }
\end{algorithm}

The blocks of the proposed leverage score sampling procedure are those of the encoded data $\Abt\coloneqq\Gb\cdot\Ab$ and $\bbt\coloneqq\Gb\cdot\bb$, for
\begin{equation}
\label{scalar_enc_matrix}
  \Gb=\diag\Big(\Big\{1\big/\sqrt{q\Pib_i}\Big\}_{i=1}^K\Big)\otimes\Ib_\tau\in\R_{\geqslant0}^{N\times N} .
\end{equation}
Specifically, the encoding carried out by the central server corresponds to the rescaling through $\Gb$. We partition both $\Abt$ and $\bbt$ across their rows analogous to \eqref{part_matrix}:
\begin{equation}
\label{part_matr_expr_LR_1}
  \Abt=\Gb\Ab=\Big[\At_1^\top \cdots \At_K^\top\Big]^\top
\end{equation}
\begin{equation}
\label{part_matr_expr_LR_2}
  \bbt=\Gb\bb=\Big[\bt_1^\top \cdots \bt_K^\top\Big]^\top
\end{equation}
where $\At_i\in\R^{{\tau}\times d}$ and $\bt_i\in\R^{\tau}$ for all $i\in\N_K$. Furthermore, all the data across the expansion network after the scalar encoding, is contained in aggregated expanded and encoded matrix-vector pairs:
$$ (\PsiB,\psiv)\coloneqq(\Ebwt\cdot\Abt,\Ebwt\cdot\bbt)\in\R^{R\tau\times d}\times\R^{R\tau} $$
as illustrated in Figure \ref{GC_block_lvg_schematic}. For the encoded objective function $L_{\Gb}(\xb)\coloneqq\|\Gb(\Ab\xb-\bb)\|_2^2$, we have:
\begin{enumerate}[label=(\roman*)]
  \item $L_{\Gb}(\xb)=\xb^\top\left(\sum\limits_{i=1}^K\At_i^\top\At_i\right)\xb+\left(\sum\limits_{i=1}^K(\bt_i^\top-2\xb^\top\At_i^\top)\bt_i\right)$
  \item $\nabla_{\xb}L_{\Gb}(\xb)=2\sum\limits_{i=1}^K\At_i^\top\left(\At_i\xb-\bt_i\right)$
  \item $\nabla_{\xb}L_{\Gb}(x_\Gb^\star)=0 \ \ \implies \ \ x_\Gb^\star=\left(\sum\limits_{i=1}^K\At_i^\top\At_i\right)^{-1}\cdot\left(\sum\limits_{i=1}^K\At_i^\top\bt_i\right)\ $.
\end{enumerate}

We make use of (ii) to approximate the gradient distributively. Each server is provided with a partition $\Dt_j=(\At_j,\bt_j)$, and computes the respective summand of the gradient from (ii), which is the encoded partial gradient on $\D_j=(\Ab_j,\bb_j)$:
\begin{equation}
\label{gt_unif_id}
  \gh_i^{[s]}\coloneqq\nabla_{\xb}L_{ls}(\Dt_i;\xb^{[s]})=\nabla_{\xb}L_{ls}\left(1\big/\sqrt{q\Pib_i}\cdot\Ab_i,1\big/\sqrt{q\Pib_i}\cdot\bb_i;\xb^{[s]}\right)=\frac{1}{q\Pib_i}\cdot g_i^{[s]} .
\end{equation}
Once a server computes its assigned partial gradient, it sends it back to the central server. When the central server receives $q$ responses, it sums them in order to obtain the approximate gradient $\gh^{[s]}$.

Denote the index multiset corresponding to the encoded pairs $(\At_j,\bt_j)$ of the received computations at iteration $s$ with $\Scal^{[s]}$, for which $|\Scal^{[s]}|=q$. The parameters vector update computed locally at the central server once sufficiently many servers respond is $\xb^{[s+1]}\gets\xb^{[s]}-\xi_s\cdot\gh^{[s]}$, where
\begin{equation}
\label{gr_update}
  \gh^{[s]} = \sum_{i\in\Scal^{[s]}}\nabla_{\xb}L_{ls}(\Dt_j;\xb^{[s]}) = 2\sum\limits_{i\in\Scal^{[s]}}\At_i^\top\left(\At_i\xb^{[s]}-\bt_i\right)
\end{equation}
and $\xi_s$ is an appropriate step-size. In the case where $q$ is not determined a priori or varies at each iteration, the scaling corresponding to $1/\sqrt{q}$ in the encoding through $\Gb$ could be done by the central server, after that iteration's computations are aggregated. We consider the case where $q$ is the same for all iterations.

Next, we present the performance guarantees of the proposed GC method.

\begin{Thm}
\label{SGD_nonunif_thm}
The proposed iterative block leverage score sketching method presented in Algorithm \ref{alg_SD_iter_sketching}, results in a SSD procedure for minimizing $L_{ls}(\PsiB,\psiv;\xb)$. Furthermore, at each iteration it produces an unbiased estimator of the least squares objective gradient \eqref{gr_ls}, i.e., $\E\left[\gh^{[s]}\right]=g^{[s]}$.
\end{Thm}

\begin{proof}
The computations of the $q$ fastest servers indexed by $\I^{[s]}$ (which corresponds to $\Ombwt^{[s]}$), are added to produce $\gh^{[s]}$, and the sampling of Algorithm \ref{alg_1_pseudocode} is according to the block leverage score distribution $\Pib_{\{K\}}$. The application of $\Ombwt^{[s]}$ (in Algorithm \ref{alg_1_pseudocode}) has a direct correspondence to the index set $\I^{[s]}$ of the $q(T)$ non-stragglers at iteration $s$, which can be viewed as drawing  $\I^{[s]}$ uniformly at random from $\{\I\subseteq\N_m:|\I|=q(T)\}$. It follows that each $\I^{[s]}$ has equal chance of occurring, which is precisely the stochastic step of SSD, i.e., each group of $q$ encoded block pairs has an equal chance of being selected.

Since the servers are homogeneous and respond independently of each other, it follows that at iteration $s$, each $\gh_i$ is received with probability $\Pib_i$. Therefore
\begin{align*}
  \E\left[\gh^{[s]}\right] &= \E\left[\sum_{i\in\I^{[s]}}\gh_i^{[s]}\right]\\
  &= \sum_{i\in\I^{[s]}}\E\left[\gh_i^{[s]}\right]\\
  &= \sum_{i\in\I^{[s]}}\sum_{j=1}^K\Pib_j\cdot\gh_j^{[s]}\\
  &= q\cdot\sum_{j=1}^K\Pib_j\cdot\gh_j^{[s]}\\
  &\overset{\flat}{=} q\cdot\sum_{j=1}^K\Pib_j\cdot\frac{1}{q\Pib_j}\cdot g_j^{[s]}\\
  &= \sum_{j=1}^Kg_j^{[s]}\\
  &= g^{[s]}
\end{align*}
where in $\flat$ we invoked \eqref{gt_unif_id}.
\end{proof}

\begin{Lemma}
\label{eq_opt_sols}  
The optimal solution of the modified least squares problem $L_{ls}(\PsiB,\psiv;\xb)$, is equal to the optimal solution $\xb^\star$ of the least squares problem \eqref{x_star_lr}.
\end{Lemma}

\begin{proof}
Note that the modified objective function $L_{ls}(\PsiB,\psiv;\xb)$ is $\|\Ebwt\Gb\cdot(\Ab\xb-\bb)\|_2^2$. Denote its optimal solution by $x^{\star}\in\R^d$. Further note that $\Ebwt$ is comprised of $\tau\times\tau$ identity matrices in such a way that it is full-rank, and $\Gb$ corresponds to a rescaling of these $\Ib_\tau$ matrices, thus $\Ebwt\Gb$ is also full-rank. It then follows that
$$ x^{\star} = \big((\Eb\Gb)\cdot\Ab\big)^\dagger\cdot\big((\Eb\Gb)\cdot\bb\big) = \Ab^\dagger\cdot\big((\Eb\Gb)^\dagger\cdot(\Eb\Gb)\big)\cdot\bb = \Ab^\dagger\cdot\Ib_N\cdot\bb = \xb^\star. $$
\end{proof}

\begin{Cor}
\label{conv_regret_cor}
The expected regret of the least squares objective resulting from Algorithm \ref{alg_SD_iter_sketching} with a diminishing step-size $\xi_s$, converges to zero at a rate of $O(1/\sqrt{s}+r/s)$, i.e.,
\begin{equation*}
  \E\big[L_{ls}(\Ab,\bb;\xb^{[s]})-L_{ls}(\Ab,\bb;\xb^\star)\big] \leqslant O(1/\sqrt{s}+r/s) .
\end{equation*}
\end{Cor}

\begin{proof}
This follows by directly applying Theorem \ref{SGD_nonunif_thm} and Lemma \ref{eq_opt_sols} to \cite[Theorem 6.3]{Bub15}.
\end{proof}

\begin{Thm}
\label{conv_x_thm}
Suppose that there exists a constant $C$ so that $\big\|\nabla_\xb L_{ls}(\Ab_{(i)},\bb_{(i)};\xb)\big\|_2^2\leqslant C$ for all $\xb\in\R^d$ and $i\in\N_N$, and that we run our Algorithm \ref{alg_SD_iter_sketching} with step-size $\xi_s=1/(\eta s)$ for a fixed $\eta>0$. Then, the error after $S$ iterations is bounded by
\begin{equation*}
\label{conv_error_x}
  \E\big[\|\xb^{[S]}-\xb^\star\|\big] \leqslant \frac{4NrC^2}{S\eta^2q\Pi_1} = O(1/S).
\end{equation*}
\end{Thm}

\begin{proof}
Following a similar analysis to the proof of \cite[Theorem 5]{BWE20}, we have
$$ \E\left[\left\|\nabla_\xb L_{ls}(\Ab_{(i)},\bb_{(i)};\xb)\right\|_2^2\right] \leqslant \E\Big[\max_{\Scal^{[s]}\subseteq{\N_K}^q}\Big\{\sum_{i\in\Scal^{[s]}}\nabla_{\xb}L_{ls}(\Abt_{(i)},\bbt_{(i)};\xb)\Big\}\Big] \leqslant C^2\cdot\frac{Nr}{q\Pi_1} . $$
By then applying \cite[Lemma 1]{RSS12}, it follows that 
\begin{equation*}
  \E\Big[\big\|\xb^{[S]}-\xb^\star\big\|_2^2\Big] \leqslant \frac{4NrC^2}{S\eta^2q\Pi_1} = O(1/S) .
\end{equation*}
\end{proof}

The crucial aspect of our expansion network; incorporated in Theorem \ref{SGD_nonunif_thm}, which allowed us to use block leverage score sampling in Algorithm \ref{alg_1_pseudocode}, is that uniform sampling of $L_{ls}(\PsiB,\psiv;\xb^{[s]})$ is $\beta_{\Pib}$-approximately equivalent to block sampling of $L_{ls}(\Abt,\bbt;\xb^{[s]})$ according to the block leverage scores of $\Ab$. Since the two objective functions are differentiable and additively separable, the resulting gradients are equal, under the assumption that we use the same $\xb^{[s]}$ and sampled index set $\Scal^{[s]}$. Note that the index set of the sampled blocks from $L_{ls}(\PsiB,\psiv;\xb^{[s]})$, corresponds to an index multiset of the sampled blocks from $L_{ls}(\Abt,\bbt;\xb^{[s]})$, as in the latter we are considering sampling with replacement. As previously mentioned, the main drawback is that in certain cases we need significantly more servers to accurately emulate $\Pi_{\{K\}}$.

The significance of Theorem \ref{SGD_nonunif_thm}, is that well-known established SD and SSD results directly apply to Algorithm \ref{alg_SD_iter_sketching} under the assumption that the approximate gradient is an unbiased estimator \cite[Chapter 14]{SB14}. Even though Algorithm \ref{alg_SD_iter_sketching} does not guarantee a descent step at every iteration, such stochastic descent methods are more common in practice when dealing with large datasets, as empirically they outperform regular SD. This is also validated in our experiments of Section \ref{exper_sec}.

Finally, we include a succinct description of the proposed approximate GC approach which utilizes block leverage score sampling, in the pseudocode of Algorithm \ref{app_GC_alg}, an illustration of which is also provided in Figure \ref{GC_block_lvg_schematic}. Below, we also give an overview of how everything is tied together.

Recall that our GC objective is to distributively approximate the gradient $g^{[s]}$. We do so by computing the gradient $\gh^{[s]}$ of the modified problem \eqref{mod_OLS}, where the sketching matrix $\Sb\gets\Sbwt_{[s]}$ is the one proposed in Algorithm \ref{alg_1_pseudocode}. In our distributed setting, the sketching matrix $\Sbwt_{[s]}$ is never explicitly computed, but we instead recover the gradient of the modified objective function \eqref{mod_OLS}, where $\Sb\gets\Sbwt_{[s]}$ is induced by the proposed procedure. Specifically, to compute the sketch's gradient $\gh^{[s]}$ we assume that the computational nodes are homogeneous, and that the central server expects a response from $q$ servers after time $T$, and we replicate each encoded pair $(\At_i,\bt_i)$ according to $r_{\{K\}}$ from Algorithm \ref{alg_replicas_mod_R}. The number of replications per task $r_{\{K\}}$ are determined such that uniformly sampling from the replicated multiset of the encoded block pairs
\begin{equation*}
  \biguplus_{i=1}^K\left\{\biguplus_{\iota=1}^{r_i}\Big\{(\At_\iota,\bt_\iota)\Big\}\right\}
\end{equation*}
approximates leverage score sampling according to $\Pit_{\{K\}}$ from Algorithm \ref{alg_1_pseudocode}. This requires that the cardinality in the above multiset union matches the total number of servers in the distributed network, i.e., $m=\sum_{j=1}^Kr_j$. In order to recover an approximate solution to \eqref{x_star_lr}, the central server locally performs a SD update which is then communicated to every server, and this process is repeated until a prespecified termination criterion is met.

\begin{algorithm}[h]
\caption{Approximate GC Approach Summary}
\label{app_GC_alg}
\SetAlgoLined
{\small
  \underline{\textbf{Central Server:}}
  \begin{enumerate}
    \item partitions $\Ab,\bb$ according to \eqref{part_matrix}, and decides $q$
    \item estimates $\Pit_{\{K\}}$
    \item determines $\rt_{\{K\}}$ from \eqref{appr_distr}
    \item passes $m,\Pit_{\{K\}},\rt_{\{K\}},\Rt=\sum_{l=1}^K\rt_l$ to Algorithm \ref{alg_replicas_mod_R}, to determine $r_{\{K\}}$, and sets $r_0=0$
    \item applies the encodings according to \eqref{scalar_enc_matrix}, \eqref{part_matr_expr_LR_1}, \eqref{part_matr_expr_LR_2}
    \item delivers the encoded block $(\At_i,\bt_i)$ to the servers indexed between $\left(1+\sum_{j=0}^{i-1}r_j\right)$ and $\left(\sum_{j=0}^{i}r_j\right)$, for each $i\in\N_K$
  \end{enumerate}

  The following iterative procedure corresponds to Algorithm \ref{alg_SD_iter_sketching}\\
  \vspace{1mm}
  While the central server's termination criterion is not met, repeat:\\
  \vspace{1mm}
  \underline{\textbf{Iterative Procedure at Iteration $s$:}}\\
  \begin{itemize}[label={--}]
    \item each server computes their corresponding partial gradient \eqref{gt_unif_id}: $\gh_i^{[s]} = \nabla_{\xb}L_{ls}(\At_i,\bt_i;\xb^{[s]})$\\
    \item \underline{\textbf{Central Server:}}
    \begin{enumerate}[label=\roman*)]
      \item sums the $q$ fastest responses indexed by $\Scal^{[s]}$, according to \eqref{gr_update}: $\gh^{[s]} = \sum_{i\in\Scal^{[s]}} \gh_i^{[s]}$\\
       {\small$\ \ \triangleright\ \Scal^{[s]}$ determines $\Sbwt_{[s]}$ and $\gh^{[s]} = \nabla_\xb L_\Sb\big(\Sbwt_{[s]},\Ab,\bb;\xb^{[s]}\big)$}\\
      \item performs the update: $\xb^{[s+1]}\gets\xb^{[s]}-\xi_s\cdot\gh^{[s]}$
      \item communicates $\xb^{[s+1]}$ to all the servers
    \end{enumerate}
  \end{itemize}

}
\end{algorithm}

\subsection{Synopsis}
\label{synopsis_subsec}

Here, we give a summary of the main theorems provided in the previous subsections. Firstly, as summarized in Subsection \ref{opt_ind_distr_subsec}, we mimic block leverage score sampling with replacement of $(\Ab,\bb)$ (from $L_{ls}(\Ab,\bb;\xb^{[s]})$) through uniform sampling, by approximately solving \eqref{appr_opt_sol} through the implication of \eqref{appr_distr} (Lemma \ref{appr_err_rounding}). This is done implicitly by communicating computations over erasure channels. Secondly, by Theorem \ref{subsp_emb_thm_lvg} we know that the proposed block leverage score sketching matrices satisfy \eqref{eq_form}, where the approximate sampling distribution $\Pib_{\{K\}}$ is determined through the proposed expansion network associated with $\Pi_{\{K\}}$. Hence, at each iteration, we approach a solution $\xbh^{[s]}$ of the induced sketched system
$$ \Sbwt_{[s]}\left(\Ab\xb^{[s]}\right)=\Sbwt_{[s]}\bb $$
which $\xb^{[s]}$ satisfies \eqref{appr_xbh} with overwhelming probability. Thirdly, in Corollary \ref{conv_regret_cor} we showed that with a diminishing step-size $\xi_s$, the expected regret of the least squares objective converges to zero at a rate of $O(1/\sqrt{s}+r/s)$ \cite{DGSX12,Bub15}. Lastly, by Theorem \ref{conv_x_thm}, our updates $\xb^{[s]}$ converge to $\xb^\star$ in expectation. A synopsis is given Figure \ref{synopsis_figure}.
\begin{figure*}[ht]
\centering
\begin{equation*}
  {\Large \left\{\substack{\text{$L_\Sb(\Sbwt_{[s]},\Ab,\bb;\xb)$ sol'ns} \\ \text{satisfy } \eqref{eq_form}\text{ and }\eqref{appr_xbh}}\right\} \ \ \xleftarrow{\text{Thm}\ \ref{subsp_emb_thm_lvg}} \ \ \left\{\substack{\text{Solve $L_{ls}(\PsiB,\psiv;\xb^{[s]})$} \\ \text{through sketched-GC}}\right\} \ \ \xrightarrow{\text{Thm}\ \ref{conv_x_thm}} \ \ \left\{\substack{\text{With a diminishing } \xi_s: \\ \lim\E[\xb^{[s]}]\to\xb^\star}\right\}}
\end{equation*}
\caption{Synopsis of our main results.}
\label{synopsis_figure}
\end{figure*}

\subsection{Approximate GC from $\ell_2$-subspace embedding}
\label{appr_GC_subsec}

In this subsection, we quantify the approximation error of the recovered gradient using Algorithm \ref{app_GC_alg}. In conventional GC, the objective is to construct an encoding matrix $\Gb\in\R^{m\times K}$ and decoding vectors $\ab_\I\in\R^{1\times q}$, such that $\ab_\I\Gb_{(\I)}=\vec{\bold{1}}$ for any set of non-straggling servers $\I$. It follows that the optimal coefficient decoding vector for a set $\I$ of size $q$ in approximate GC \cite{CPE17,KKR19,SH22} is
\begin{equation}
\label{opt_dec_err}
  \ab_\I^\star = \arg\min_{\ab\in\R^{1\times q}}\big\{\|\ab\Gb_{(\I)}-\vec{\bold{1}}\|_2^2\big\} \ \ \implies \ \ \ab_\I^\star = \vec{\bold{1}}\Gb_{(\I)}^{\dagger}
\end{equation}
for $\Gb_{(\I)}^{\dagger}=\big(\Gb_{(\I)}^\top\Gb_{(\I)}\big)^{-1}\Gb_{(\I)}^\top\in\R^{K\times q}$.

\begin{Prop}
\label{prop_opt_dec_err}
The error in the approximated gradient $\grave{g}^{[s]}$ of an approximate optimal coefficient decoding linear regression GC scheme $(\Gb,\ab_\I^\star)$, satisfies
\begin{equation}
\label{appr_GC_error_prop} 
  \big\|g^{[s]}-\grave{g}^{[s]}\big\|_2 \leqslant 2\sqrt{K}\cdot\err(\Gb_{(\I)})\cdot\|\Ab\|_2\cdot\|\Ab\xb^{[s]}-\bb\|_2
\end{equation}
for $\err(\Gb_{(\I)}) \coloneqq \big\|\Ib_K-\Gb_{(\I)}^{\dagger}\Gb_{(\I)}\big\|_2$.
\end{Prop}

\begin{proof}
Consider the optimal coefficient decoding vector of an approximate GC scheme $\ab_\I^\star$ \eqref{opt_dec_err}. In the case where $q\geqslant K$, it follows that $\ab_\I^\star = \vec{\bold{1}}\Gb_{(\I)}^{\dagger}$.

Let $\gb^{[s]}$ be the matrix comprised of the transposed exact partial gradients at iteration $s$, i.e.,
\begin{equation*}
  \gb^{[s]} \coloneqq {\begin{pmatrix} g_1^{[s]} & g_2^{[s]} & \hdots & g_K^{[s]} \end{pmatrix}}^\top \in \R^{K\times d} .
\end{equation*}
Then, for a GC encoding-decoding pair $(\Gb,\ab_\I)$ satisfying $\ab_\I\Gb_{(\I)}=\vec{\bold{1}}$ for any $\I$, it follows that
$$ \ab_\I\left(\Gb_{(\I)}\gb^{[s]}\right) = \vec{\bold{1}}\gb^{[s]} = \sum_{j=1}^K\big(g_j^{[s]}\big)^\top = \big(g^{[s]}\big)^\top . $$
Hence, the gradient can be recovered exactly. Considering an optimal approximate scheme $(\Gb,\ab_\I^\star)$ which recovers the gradient estimate $\grave{g}^{[s]}=\left(\ab_\I^\star\Gb_{(\I)}\right)\gb^{[s]}$, the error in the gradient approximation is:
\begin{align*}
  \big\|g^{[s]}-\grave{g}^{[s]}\big\|_2 &= \left\|\big(\vec{\bold{1}}-\ab_\I^\star\Gb_{(\I)}\big)\gb^{[s]}\right\|_2\\
  &= \left\|\vec{\bold{1}}\big(\Ib_K-\Gb_{(\I)}^{\dagger}\Gb_{(\I)}\big)\gb^{[s]}\right\|_2\\
  &\leqslant \|\vec{\bold{1}}\|_2\cdot \left\|\Ib_K-\Gb_{(\I)}^{\dagger}\Gb_{(\I)}\right\|_2\cdot\big\|\gb^{[s]}\big\|_2\\
  &\overset{\pounds}{\leqslant} \sqrt{K}\cdot \left\|\Ib_K-\Gb_{(\I)}^{\dagger}\Gb_{(\I)}\right\|_2\cdot\big\|g^{[s]}\big\|_2\\
  &\overset{\textdollar}{\leqslant} 2\sqrt{K}\cdot\underbrace{\left\|\Ib_K-\Gb_{(\I)}^{\dagger}\Gb_{(\I)}\right\|_2}_{\err(\Gb_{(\I)})}\cdot\|\Ab\|_2\cdot\|\Ab\xb^{[s]}-\bb\|_2
\end{align*}
where $\pounds$ follows from the facts that $\|\gb^{[s]}\|_2\leqslant\|g^{[s]}\|_2$ and $\|\vec{\bold{1}}\|_2=\sqrt{K}$, and $\textdollar$ from \eqref{gr_ls} and sub-multiplicativity of matrix norms.
\end{proof}

In Theorem \ref{app_GC_thm}, we show the accuracy of the approximate gradient of iterative GC approaches based on sketching techniques that satisfy (with high probability) the $\ell_2$-subspace embedding property $\|\Ib_d-\Ub^\top\Sb^\top\Sb\Ub\|_2\leqslant\epsilon$ from \eqref{eq_form}. This then holds true for the GC method of Algorithm \ref{app_GC_alg}, by Theorem \ref{subsp_emb_thm_lvg} and the fact that we distributively emulate block leverage score sampling.

\begin{Thm}
\label{app_GC_thm}
Assume that the induced sketching matrix $\Sb$ from a GC scheme satisfies $\|\Ib_d-\Ub^\top\Sb^\top\Sb\Ub\|_2\leqslant\epsilon$ (with high probability). Then, the updated approximate gradient estimate $\gh^{[s]}$ at any iteration, satisfies (with high probability):
\begin{equation*}
\label{appr_GC_error}
  \big\|g^{[s]}-\gh^{[s]}\big\|_2 \leqslant 2\epsilon\cdot\|\Ab\|_2\cdot\|\Ab\xb^{[s]}-\bb\|_2.
\end{equation*}
Specifically, it satisfies \eqref{appr_GC_error_prop} with $\err(\Gb_{(\I)})=\epsilon/\sqrt{K}$.
\end{Thm}

\begin{proof}
Consider $\gh^{[s]}$ the approximated gradient in Algorithm \ref{alg_SD_iter_sketching} for linear regression with gradient $g^{[s]}$ \eqref{gr_ls}. It follows that
\begin{align*}
  \big\|g^{[s]}-\gh^{[s]}\big\|_2 &= \|2\Ab^\top(\Ab\xb^{[s]}-\bb)-2\Ab^\top(\Sb^\top\Sb)(\Ab\xb^{[s]}-\bb)\|_2 \\
  &= 2\|\Ab^\top(\Ib_N-\Sb^\top\Sb)(\Ab\xb^{[s]}-\bb)\|_2\\
  &\leqslant 2\|\Ab\|_2\cdot\|\Ib_N-\Sb^\top\Sb\|_2\cdot\|\Ab\xb^{[s]}-\bb\|_2\\
  &= 2\|\Ab\|_2\cdot\|\Ub^\top(\Ib_N-\Sb^\top\Sb)\Ub\|_2\cdot\|\Ab\xb^{[s]}-\bb\|_2\\
  &= 2\|\Ab\|_2\cdot\|\Ib_d-\Ub^\top\Sb^\top\Sb\Ub\|_2\cdot\|\Ab\xb^{[s]}-\bb\|_2\\
  &\overset{\flat}{\leqslant} 2\epsilon\cdot\|\Ab\|_2\cdot\|\Ab\xb^{[s]}-\bb\|_2
\end{align*}
where in $\flat$ we make use of the assumption that $\Sb$ satisfies \eqref{eq_form}. Our approximate GC approach therefore (with high probability) satisfies \eqref{appr_GC_error_prop}, with $\err(\Gb_{(\I)})=\epsilon/\sqrt{K}$.
\end{proof}

\section{Experiments}
\label{exper_sec}

In this section, we validate our theoretical results through numerical experiments. The minimum benefit of Algorithm \ref{alg_1_pseudocode} occurs when $\Pi_{\{K\}}$ is close to uniform. For this reason, we construct dataset matrices whose resulting sampling distributions and block leverage scores are non-uniform.

In our first experiment, we compare our iterative sketching approach, implemented as in Algorithm \ref{alg_SD_iter_sketching}, to other iterative sketching techniques and to regular SD (without sketching). For this experiment, we generated $\Ab\in\R^{2000\times40}$ following a $t$-distribution, and standard Gaussian noise was added to an arbitrary vector from $\image(\Ab)$ to define $\bb$. We considered $K=100$ blocks, thus $\tau=20$. The effective dimension $N=2000$ was reduced to $r=1000$, i.e., $q=50$. We compared the iterative approach with exact block leverage scores, i.e., $\beta=1$, against analogous approaches using the block-SRHT \cite{CMPH23} and Gaussian sketches, and regular SD.

For Figure \ref{varying_its_t_distr} we ran 600 iterations on six different instances for each approach, and varied $\xi$ for each experiment by logarithmic factors of $\xi^\times=2/\sigma_{\max}(\Ab)^2$. The average $\log$ residual errors $\log_{10}\big(\|\xb^\star-\xbh\|_2\big/\sqrt{N}\big)$ are depicted in Figure \ref{varying_its_t_distr}, and reported in Table \ref{log_res_err_table}. In Figure \ref{conv_fig_t_distr} we observe the convergence of the different approaches, in the case where $\xi\approx 0.42$. In this case, our method (block-lvg) outperforms the Gaussian sketching approach and regular SD. The fact that the performance of the block-SRHT is similar to Algorithm \ref{alg_1_pseudocode}, reflects the result of Proposition \ref{comparison_two_embds} in Appendix \ref{comp_SRHT_app}.

\begin{figure}[h]
  \centering
  \includegraphics[scale=.17]{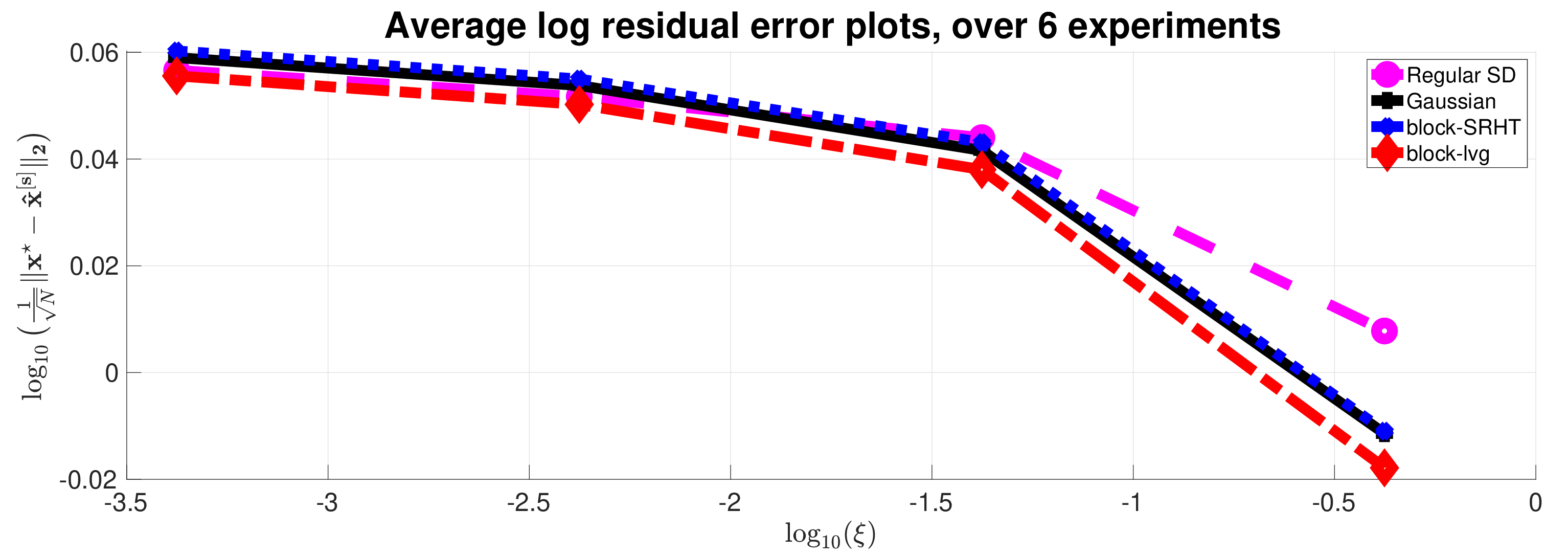}
  \caption{Residual error for varying $\xi_s$.}
  \label{varying_its_t_distr}
\end{figure}

\begin{center}
\begin{table}[h]
\centering
\begin{tabular}{ |p{1.6cm}||p{1cm}|p{1cm}|p{1cm}|p{1cm}| }
\hline
\multicolumn{5}{|c|}{\textbf{Average $\log$ residual error: $\log_{10}\big(\|\xb^\star-\xbh\|_2\big/\sqrt{N}\big)$}} \\
\hline
\textbf{$\log_{10}(\xi)$} & 0.0004 &
0.0042 & 0.0421 & 0.4207 \\
\hline
\hline
\textbf{Regular SD} & 0.0566 & 0.0517 & 0.0440 & 0.0078 \\
\hline
\textbf{Gaussian} & 0.0590 & 0.0538 & 0.0416 & -0.0114 \\
\hline
\textbf{block-SRHT} & 0.0603 & 0.0550 & 0.0431 & -0.0110 \\
\hline
\textbf{block-lvg} & 0.0556 & 0.0502 & 0.0380 & -0.0178 \\
\hline
\end{tabular}
\caption{Average $\log$ residual errors, for six instances of SD with fixed steps, when performing Gaussian sketching with updated sketches, iterative block-SRHT and iterative block leverage score sketching, and uncoded SD.}
\label{log_res_err_table}
\end{table}
\end{center}

\begin{figure}[h]
  \centering
  \includegraphics[scale=.17]{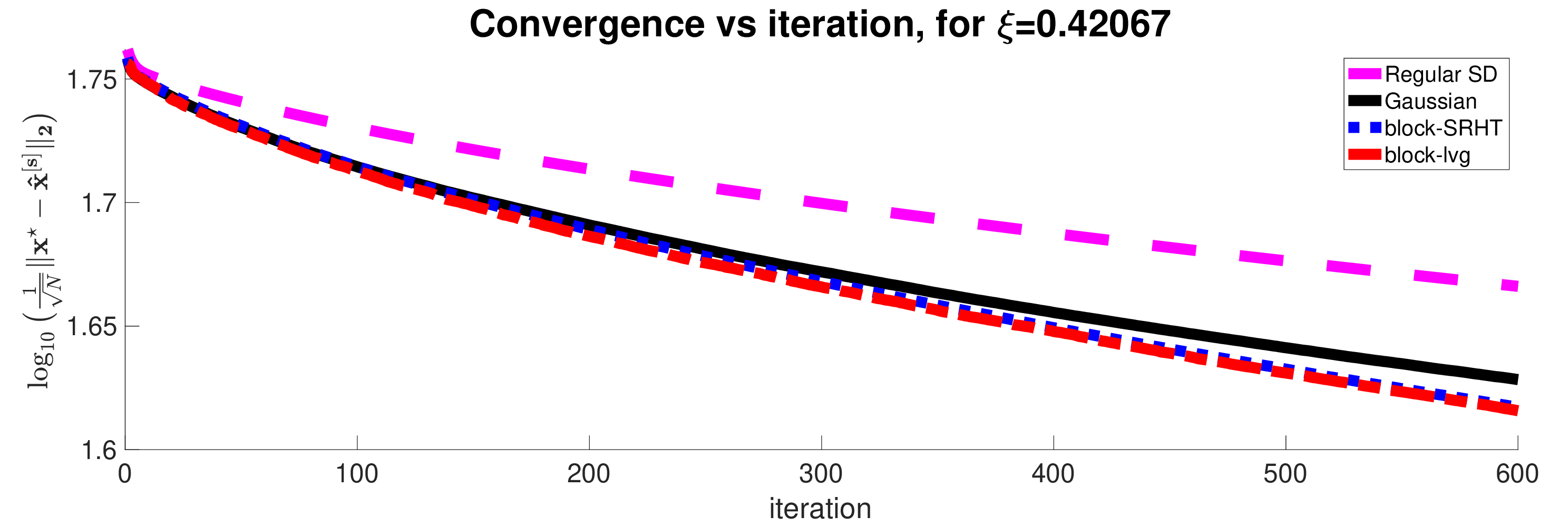}
  \caption{$\log$ residual error convergence.}
  \label{conv_fig_t_distr}
\end{figure}

We also considered the same experiment with $\Ab$ drawn from a $t$-distribution, with optimal step-size
$$ \xi_s^{\star} = \langle\Ab g^{[s]},\Ab\xb^{[s]}-\bb\rangle\big/\|\Ab g^{[s]}\|_2^2 $$
at each iteration. From Figure \ref{opt_step_t_distr}, we observe that our iterative sketching approach outperforms Gaussian sketching with updated sketches, and iterative sketching is superior to non-iterative. Furthermore, we validate Theorem \ref{SGD_nonunif_thm} and Lemma \ref{eq_opt_sols}, as our iterative sketching approach and SSD have similar convergence. Additionally, it was observed that in some cases our iterative sketching method would outperform regular SD (and SSD). We also compared our method with updated step-sizes $\xi_s^{\star}$, to iterative and non-iterative approaches according to the block leverage score sampling, block-SRHT, and Rademacher sketching methods \cite{Ach03,DMMS11,CMPH23}, in which our corresponding approach again produced more accurate final approximations. Further experiments (not shown) were performed on various distributions for $\Ab$, in which similar results to those of Figures \ref{varying_its_t_distr}, \ref{conv_fig_t_distr} and \ref{opt_step_t_distr} were observed.

\begin{figure}[h]
  \centering
  \includegraphics[scale=.17]{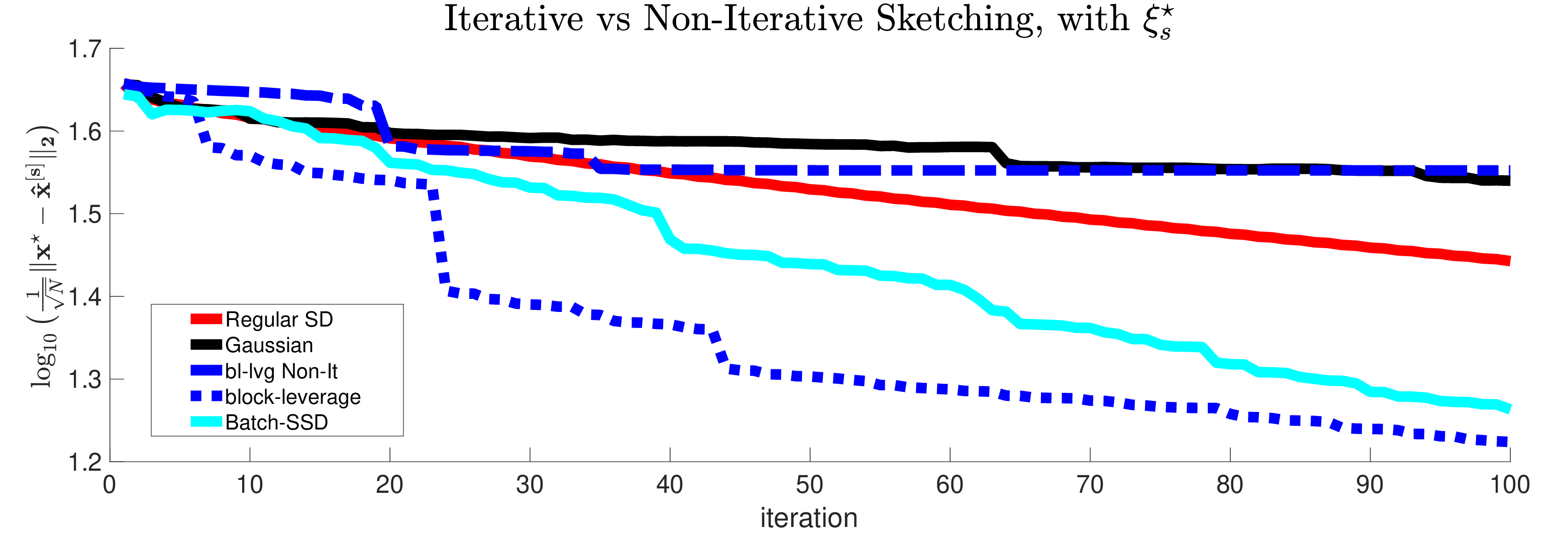}
  \caption{$\log$ convergence with $\xi_s^{\star}$.}
  \label{opt_step_t_distr}
\end{figure}

The quality of our approximate block leverage sampling distribution $\Pib_{\{K\}}$ was evaluated through an empirical mother runtime distribution taken from a distributed AWS experiment \cite{BP19a}. We considered the mother runtime distribution from the aforementioned experiment in which 500 computational servers were used to solve a large linear regression problem. The first step was to determine $\rh_{\{K\}}$ according to \eqref{appr_distr}, and then Algorithm \ref{alg_replicas_mod_R} was performed to determine $r_{\{K\}}$ so that $\sum_{i=1}^K r_i=500$, where $K=100$. We achieved satisfactory approximations $\Pib_{\{K\}}$ to $\Pi_{\{K\}}$, for varying ending times $T$ (measured in seconds), i.e., $\beta_{\Pib}\approx1$ for each $T$. In Table \ref{misestimations_table} we present the resulting $q(T)$, misestimation factors $\beta_{\Pib}$, and metric $\Delta_{\Pi,\Pib}$, for each $T$.

\begin{center}
\begin{table}[h]
\centering
\begin{tabular}{ |p{.9cm}||p{.6cm}|p{.6cm}|p{.6cm}|p{.6cm}|p{.6cm}|p{.6cm}|p{.6cm}| }
\hline
\multicolumn{8}{|c|}{\textbf{Expansion Networks Misestimation Factors and Metrics}} \\
\hline
\hline
{\small$t\gets T$} & 6 & 12 & 24 & 42 & 48 & 54 & 60 \\
\hline
\hline
$q(T)$ & {\small 142} & {\small 226} & {\small 230} & {\small 230} & {\small 262} & {\small 267} & {\small 499} \\ \hline
$\Delta_{\Pi,\Pib}$ & {\footnotesize .8859} & {\footnotesize .9078} & {\footnotesize .9079} & {\footnotesize .9079} & {\footnotesize .9079} & {\footnotesize .9079} & {\footnotesize .9079} \\ \hline
$\beta_{\Pib}$ & {\footnotesize .9407} & {\footnotesize .9720} & {\footnotesize .9727} & {\footnotesize .9727} & {\footnotesize .9727} & {\footnotesize .9727} & {\footnotesize .9727} \\ \hline
\end{tabular}
\caption{Misestimations with AWS server completion times.}
\label{misestimations_table}
\end{table}
\end{center}

\section{Conclusion and Future Work}
\label{concl_sec}

In this paper, we showed how one can exploit results from RandNLA to distributed CC, in the context of GC. By taking enough samples, or equivalently after waiting long enough, the approximation errors can be made arbitrarily small. In terms of CC, the advantages are that encodings correspond to a scalar multiplication, and no decoding step is required. Our proposed method has advantages relative to other CC approximation methods \cite{SH22,BWE20,CP18,CPE17,CMPH22,CMPH23,FD16,GW21,KKR19,RTTD17}, by incorporating information from our dataset into our scheme. These advantages include parallelization of block leverage score sketching. This is accomplished by leveraging the fact that non-straggling servers correspond to the sampled partitions of the proposed sketching algorithm.

Our proposed sketching also leads to faster algorithms. As observed experimentally in Section \ref{exper_sec}, our scheme outperforms Gaussian sketching and slightly outperforms block-SRHT sketching techniques. The improvements are also greater when adaptive step-sizes are considered. As was emphasized in the paper, iterative sketching has not been thoroughly studied in first order methods, which is what motivated our approach.

We quantified how much better the proposed solution is, both empirically and theoretically. The proposed method (Algorithm \ref{alg_SD_iter_sketching}) compares favorably to other sketching methods in terms of empirical convergence rate and approximation error. As contrasted to other approximate GC methods, our algorithm uses data dependent spectral approximations, that enables the algorithm to quantify the residual error of the gradient approximation at each step. In terms of convergence guarantees, we obtain the same convergence rate to related works, e.g., \cite{BWE20}, which matches the convergence rate of SSD.

Recently, papers have addressed cases where a certain subset of servers are persistent or adversarial stragglers \cite{CT22,BOUG22}. To adapt the theory developed here to the case where adversarial stragglers are present, would require modifying the replication across the non-persistent stragglers, in order to better approximate the sampling probability $\Pi_{\{K\}}$ through $\Pib_{\{K\}}$. This is worthwhile future work.

Even though we focused on leverage score sampling for linear regression, other sampling algorithms and problems could benefit by designing analogous replication schemes. One such problem is the column subset selection problem, which can be used to compute partial $\SVD$, $\QR$ decompositions, as well as low-rank approximations \cite{Mah16}. As for the sampling technique we studied, one can judiciously define a sampling distribution to approximate solutions to such problems \cite{BMD09}, which are known to be \textsf{NP-hard} under the Unique Games Conjecture assumption \cite{Civ14}.

Furthermore, existing block sampling algorithms can also benefit from the proposed expansion networks, e.g., $CR$-multiplication \cite{CPH20c} and $CUR$ decomposition \cite{OJXE18}. For instance, a coded matrix multiplication algorithm of minimum variance can been designed, where the sampling distribution proposed in \cite{CPH20c} is used to determine the replication numbers of the expansion network. In terms of matrix decompositions, the block leverage score algorithm of \cite{OJXE18} can be used to distributively determine an additive $\epsilon$-error decomposition of $\Ab$, in the CC setting. Another future direction is generalizing existing tensor product and factorization algorithms to block sampling, according to both approximate and exact sampling distributions, in order to make them practical for distributed environments.


\section*{Acknowledgment}

\noindent The research presented in this paper was partially supported by the U.S. National Science Foundation (NSF) under grants NSF-2217058; NSF-2133484; NSF-2246157; ECCS-2037304 and DMS-2134248, the NSF CAREER Award under grant CCF-2236829, the U.S. Army Research Office under grant W911NF2310343; the U.S. Army Research Office Early Career Award under grant W911NF-21-1-0242, the U.S. Department of Energy under grant DE-NA0003921, and by the Office of Naval Research under grant N00014-24-1-2164. We also thank the anonymous reviewers and Aditya Ramamoorthy for their constructive comments and feedback.

\appendices

\section{Proofs of Subsection \ref{lvg_sec}}
\label{lvg_app}

In this appendix, we provide the proof of Theorem \ref{subsp_emb_thm_lvg} and Corollary \ref{sub_opt_cor}. We first recall the following matrix Chernoff bound \cite[Fact 1]{Woo14}, and prove Proposition \ref{prop_Sp_lvg}.

\begin{Thm}[Matrix Chernoff Bound, {\cite[Fact 1]{Woo14}}]
\label{matr_Chern}
  Let $\Xb_1,\ldots,\Xb_q$ be independent copies of a symmetric random matrix $\Xb\in\R^{d\times d}$, with $\E[\Xb]=\bold{0}_{d\times d}, \|\Xb\|_2\leqslant \gamma$, $\|\E[\Xb^\top\Xb]\|_2\leqslant \sigma^2$. Let $\Zb=\frac{1}{q}\sum_{i=1}^q\Xb_i$. Then $\forall\epsilon>0$:
  \begin{equation*}
  \label{matr_Chern_expr}
    \Pr\Big[\|\Zb\|_2>\epsilon\Big]\leqslant2d\cdot\exp\left(-\frac{q\epsilon^2}{\sigma^2+\gamma\epsilon/3}\right).
  \end{equation*}
\end{Thm}

\begin{Prop}
\label{prop_Sp_lvg}
The sketching matrix $\Sbwt$ of Algorithm \ref{alg_1_pseudocode} with $\Pit_i\geqslant \beta\Pi_i$ for all $i$ and $\beta\in(0,1]$, guarantees
\begin{equation}
\label{conc_lvg_b}
  \Pr\Big[\|\Ib_d-\Ub^\top\Sbwt^\top\Sbwt\Ub\|_2>\epsilon\Big] \leqslant 2d\cdot e^{-q\epsilon^2\Theta(\beta/d)}
\end{equation}
for any $\epsilon>0$, and $q=r/\tau>d/\tau$.
\end{Prop}

\begin{proof}{[Proposition \ref{prop_Sp_lvg}]}
In order to use Theorem \ref{matr_Chern}, we first need to define an appropriate symmetric random matrix $\Xb$, whose realizations $\Xb_{\{q\}}$ correspond to the sampling procedure of Algorithm \ref{alg_1_pseudocode}, and $\Sbwt^\top\Sbwt=\frac{1}{q}\sum_{i=1}^q\Xb_i$. The realization of the matrix random variable are
$$ \Xb_i=\Ib_d-\left(\frac{\Ub_{\left(\K_\iota^i\right)}^\top\Ub_{\left(\K_\iota^i\right)}}{\Pit_\iota}\right)=\Ib_d-\left(\frac{\sum_{l\in\K_\iota^i}\Ub_{(l)}^\top\Ub_{(l)}}{\Pit_\iota}\right) $$
where $\K_\iota^i$ indicates the $\iota^{th}$ block of $\Ab\in\R^{N\times d}$, which was sampled at trial $i$. This holds for all $\iota\in\N_K$. The expectation of the symmetric random matrix $\Xb$ is
\begin{align*}
  \E[\Xb] &\overset{\ddagger}{=} \Ib_d-\left(\sum_{j=1}^K\Pit_j\cdot\frac{\Ub_{(\K_j)}^\top\Ub_{(\K_j)}}{\Pit_j}\right)\\
  &= \Ib_d-\sum_{j=1}^K\Ub_{(\K_j)}^\top\Ub_{(\K_j)}\\
  &\overset{\sharp}{=} \Ib_d-\sum_{l=1}^N \Ub_{(l)}^\top\Ub_{(l)}\\
  &= \Ib_d-\Ib_d\\
  &=\bold{0}_{d\times d}
\end{align*}
where $\ddagger$ follows from the fact that each realization $\Xb_i$ corresponding to each $\{\K_j^i\}_{j=1}^K$ of the random matrix is sampled with probability $\Pit_j$, and in $\sharp$ we simplify the expression in terms of rank-1 outer-product matrices. Furthermore, for $\{\ellt_l\}_{l=1}^N$ the corresponding \textit{approximate} leverage scores of $\Ab$
\begin{align*}
  \|\Xb_i\|_2 \ &\overset{\natural}{\leqslant} \|\Ib_d\|_2+\frac{\|\Ub_{\left(\K_\iota^i\right)}^\top\Ub_{\left(\K_\iota^i\right)}\|_2}{\Pit_\iota}\\
  &\overset{\diamond}{\leqslant} 1+\frac{\sum_{l\in\K_\iota^i}\ell_l}{\left(\sum_{l\in\K_\iota^i}\ellt_l\right)\big/d}\\
  &= 1+\frac{d\cdot\Pi_\iota}{\Pit_\iota}\\
  &= 1+\frac{d}{\beta}
\end{align*}
where in $\natural$ we use the triangle inequality on the definition of $\Xb_i$, and in $\diamond$ we use it on the sum of outer-products (the numerator of second summand).

We now upper bound the variance of the copies of $\Xb$:
\begin{align}
  \left\|\E\big[\Xb_i^\top\Xb_i\big]\right\|_2 &= \left\|\E\left[\left(\Ib_d-\left(\Ub_{(\K_\iota)}^\top\Ub_{(\K_\iota)}\big/\Pit_\iota\right)\right)^\top\left(\Ib_d-\left(\Ub_{(\K_\iota)}^\top\Ub_{(\K_\iota)}\big/\Pit_i\right)\right)\right]\right\|_2 \notag\\
  &=\left\|\Ib_d-2\cdot\E\left[(\Ub_{(\K_\iota)}^\top\Ub_{(\K_\iota)}\big/\Pit_\iota\right]+\E\left[\left(\Ub_{(\K_\iota)}^\top\Ub_{(\K_\iota)}\right)^2\big/\Pit_\iota^2\right]\right\|_2\notag\\
  &=\bigg\|2\Big(\Ib_d-\overbrace{\E\left[(\Ub_{(\K_\iota)}^\top\Ub_{(\K_\iota)}\big/\Pit_\iota\right]}^{=\Ib_d}\Big)-\Ib_d+\E\left[\left(\Ub_{(\K_\iota)}^\top\Ub_{(\K_\iota)}\right)^2\big/\Pit_\iota^2\right]\bigg\|_2 \notag\\
  &= \left\|\E\left[\left(\Ub_{(\K_\iota)}^\top\Ub_{(\K_\iota)}\right)^2\big/\Pit_\iota^2\right]-\Ib_d\right\|_2 \notag\\
  &= \left\|\left(\sum_{\iota=1}^K\Pi_\iota\cdot\left(\Ub_{(\K_\iota)}^\top\Ub_{(\K_\iota)}\right)^2\big/\Pit_\iota^2\right)-\Ib_d\right\|_2 \notag\\
  &\overset{\sharp}{\leqslant} \left\|\left(\sum_{\iota=1}^K\frac{1}{\beta}\left(\Ub_{(\K_\iota)}^\top\Ub_{(\K_\iota)}\right)^2\big/\Pit_\iota\right)-\Ib_d\right\|_2 \notag\\
  &= \left\|\left(\sum_{\iota=1}^K\frac{d}{\beta}\cdot\frac{\left(\Ub_{(\K_\iota)}^\top\Ub_{(\K_\iota)}\right)^2}{\|\Ub_{(\K_\iota)}\|_F^2}\right)-\Ib_d\right\|_2 \notag\\
  &\leqslant \left\|\left(\sum_{\iota=1}^K\frac{d}{\beta}\cdot\frac{\left(\Ub_{(\K_\iota)}^\top\Ub_{(\K_\iota)}\right)^2}{\|\Ub_{(\K_\iota)}\|_2^2}\right)-\Ib_d\right\|_2 \notag\\
  &\overset{\flat}{=} \left\|\frac{d}{\beta}\cdot\left(\sum_{\iota=1}^K\left(\Ub_{(\K_\iota)}^\top\Ub_{(\K_\iota)}\right)^2\right)-\Ib_d\right\|_2 \notag\\
  &\overset{\natural}{\leqslant} \left\|\frac{d}{\beta}\left(\sum_{\iota=1}^K\left(\Ub_{(\K_\iota)}^\top\Ub_{(\K_\iota)}\right)\left(\Ib_d-\Ub_{(\Kbar_\iota)}^\top\Ub_{(\Kbar_\iota)}\right)\right)-\Ib_d\right\|_2 \notag\\
  &= \bigg\|\frac{d}{\beta}\overbrace{\left(\sum_{\iota=1}^K\Ub_{(\K_\iota)}^\top\Ub_{(\K_\iota)}\right)}^{=\Ib_d}-\frac{d}{\beta}\left(\sum_{\iota=1}^K\left(\Ub_{(\K_\iota)}^\top\Ub_{(\K_\iota)}\right)\left(\Ub_{(\Kbar_\iota)}^\top\Ub_{(\Kbar_\iota)}\right)\right) -\Ib_d\bigg\|_2 \notag\\
  &= \bigg\|(d/\beta-1)\cdot\Ib_d-\frac{d}{\beta}\bigg(\sum_{\iota=1}^K\Ub_{(\K_\iota)}^\top\overbrace{\left(\Ub_{(\K_\iota)}\Ub_{(\Kbar_\iota)}^\top\right)}^{\bold{0}_{d\times d}}\Ub_{(\Kbar_\iota)}\bigg)\bigg\|_2 \notag\\
  &= \left\|(d/\beta-1)\cdot\Ib_d\right\|_2 \notag\\
  &= d/\beta-1 \notag
\end{align}
where in $\sharp$ we used the fact $\Pi_\iota/\Pit_\iota\leqslant1/\beta$, in $\flat$ that $\|\Ub_{(\K_\iota)}\|_2^2=1$, and in $\natural$ that $\Ub_{(\K_\iota)}^\top\Ub_{(\K_\iota)}=\Ib_d-\Ub_{(\Kbar_\iota)}^\top\Ub_{(\Kbar_\iota)}$ for each $\iota$.

According to Theorem \ref{matr_Chern}, we substitute $\gamma=1+d$ and $\sigma^2=d/\beta-1$ to get
\begin{align*}
  \frac{1}{q}\sum\limits_{i=1}^{q}\Xb_i &= \frac{1}{q}\sum\limits_{i=1}^{q}\left(\Ib_d-\frac{\Ub_{\left(\K_{j(i)}\right)}^\top\Ub_{\left(\K_{j(i)}\right)}}{\Pit_{j(i)}}\right)\\
  &= \Ib_d-\frac{1}{q}\left(\sum\limits_{i=1}^{q}\frac{\Ub_{\left(\K_{j(i)}\right)}^\top\Ub_{\left(\K_{j(i)}\right)}}{\Pit_{j(i)}}\right)\\
  &= \Ib_d-\Ub^\top\Sbwt^\top\Sbwt\Ub
\end{align*}
where the last equality follows from the definition of $\Sbwt$. Putting everything together into Theorem \ref{matr_Chern}, we get that
$$ \Pr\big[\|\Ib_d-\Ub^\top\Sbwt^\top\Sbwt\Ub\|_2>\epsilon\big]\leqslant2d\cdot e^{-q\epsilon^2\Theta(\beta/d)}. $$
\end{proof}

\begin{proof}{[Theorem \ref{subsp_emb_thm_lvg}]}
By substituting
$$ q=\Theta\left(\frac{d}{\tau}\log{(2d/\delta)}/(\beta\epsilon^2)\right) $$
in \eqref{conc_lvg_b} and taking the complementary event, we attain the desired statement.
\end{proof}

Before we prove Corollary \ref{sub_opt_cor}, we introduce the notion of \textit{block $\alpha$-balanced}, which is a generalization of \textit{$\alpha$-balanced} from \cite{PW16}. The sampling distribution $\Pi_{\{K\}}$ is said to be \textit{block $\alpha$-balanced}, if
\begin{equation*}
\label{bl_alpha_bal}
  \max_{i\in\N_K}\left\{\Pi_i\right\} \leqslant \frac{\alpha}{N/\tau}=\dfrac{\alpha}{K}
\end{equation*}
for some constant $\alpha$ independent of $K$ and $q$. Furthermore, in our context, if the individual leverage scores $\pi_{\{N\}}$ are $\alpha$-balanced for $\alpha$ independent of $N$ and $r$, then the block leverage scores $\Pi_{\{K\}}$ are block $\alpha$-balanced, as
\begin{equation}
\label{impl_bal_wts}  
  \max_{i\in\N_K}\left\{\Pi_i\right\} \leqslant \tau\cdot\max_{j\in\N_N}\left\{\pi_j\right\} \leqslant \tau\cdot\frac{\alpha}{N} = \frac{\alpha}{N/\tau} = \frac{\alpha}{K} .
\end{equation}

\begin{proof}{[Corollary \ref{sub_opt_cor}]}
From the proof of {\cite[Theorem 1]{PW16}}, we simply need to show that
$$ \big\|\E\big[\Sbwt^\top\big(\Sbwt\Sbwt^\top\big)^{-1}\Sbwt\big]\big\|_2\leqslant\eta\cdot\frac{r}{N} $$
for $\Sbwt$ a single sketch produced in Algorithm \ref{alg_1_pseudocode}, and an appropriate constant $\eta$ independent of $N$ and $r$. We assume that the individual leverage scores $\pi_{\{N\}}$ are $\alpha$-balanced, where $\alpha$ is a constant independent of $N$ and $r$. By \eqref{impl_bal_wts}, it follows that the block leverage scores $\Pi_{\{K\}}$ are block $\alpha$-balanced, i.e., $\Pi_i\leqslant \Pi_K \leqslant \frac{\alpha}{K}$ for all $i\in\N_{K-1}$.

A direct computation shows that
\begin{align*}
  \big(\Sbwt\Sbwt^\top\big)^{-1} &= \Big((\Db\cdot\Omb\otimes\Ib_\tau)\cdot(\Omb^\top\Db^\top\otimes\Ib_\tau)\Big)^{-1}\\
  &= \Big(\big(\Db\cdot\Omb\cdot\Omb^\top\Db^\top\big)\otimes\Ib_\tau\Big)^{-1}\\
  &= \big(\Db\cdot\Omb\cdot\Omb^\top\Db^\top\big)^{-1}\otimes\Ib_\tau
\end{align*}
and consequently
\begin{align*}
  \Sbwt^\top\big(\Sbwt\Sbwt^\top\big)^{-1}\Sbwt &= \big(\Omb^\top\cdot\Db^\top\otimes\Ib_\tau\big) \cdot \Big(\big(\Db\cdot\Omb\cdot\Omb^\top\cdot\Db^\top\big)^{-1}\otimes\Ib_\tau\Big)\cdot(\Db\cdot\Omb\otimes\Ib_\tau) \\
  &= \Big(\Omb^\top\cdot\Db^\top \cdot \big(\Db\cdot\Omb\cdot\Omb^\top\cdot\Db^\top\big)^{-1}\otimes\Ib_\tau\Big)\cdot(\Db\cdot\Omb\otimes\Ib_\tau) \\
  &= \underbrace{\Big(\Omb^\top\cdot\Db^\top \cdot \big(\Db\cdot\Omb\cdot\Omb^\top\cdot\Db^\top\big)^{-1}\cdot\Db\cdot\Omb\Big)}_{\coloneqq Z\in\R_{\geqslant0}^{K\times K}}\otimes\Ib_\tau \\
\end{align*}
where $Z=\sum_{\iota=1}^qZ_\iota$, for $\{Z_\iota\}_{\iota=1}^q$ rank-1 outer-product matrices of size $K\times K$ corresponding to each sampling trial, through $\Omb$. For each sampling trial, we know that the $i^{th}$ block is sampled with probability $\Pi_i$. Furthermore, if the $i^{th}$ block is sampled at iteration $\iota$, it follows that
\begin{align*}
  Z_\iota &= \eb_i\cdot\sqrt{\frac{1}{q\Pi_i}}\cdot\left(\eb_i^\top\cdot\left(\sqrt{\frac{1}{q\Pi_i}}\right)^2\cdot\eb_i\right)^{-1}\cdot\sqrt{\frac{1}{q\Pi_i}}\cdot\eb_i^\top \\
  &= \eb_i\cdot\sqrt{\frac{1}{q\Pi_i}}\cdot\left(\sqrt{q\Pi_i}\right)^2\cdot\sqrt{\frac{1}{q\Pi_i}}\cdot\eb_i^\top \\
  &= \eb_i\cdot\eb_i^\top
\end{align*}
hence
\begin{equation*}
  \E\left[\Sbwt^\top\big(\Sbwt\Sbwt^\top\big)^{-1}\Sbwt\right] = \E_\Omb\left[\sum_{j=1}^K\eb_j\cdot\eb_j^\top\right]\otimes\Ib_\tau = \left(\sum_{j=1}^K\E_\Omb\left[\eb_j\cdot\eb_j^\top\right]\right)\otimes\Ib_\tau = \Qb\otimes\Ib_\tau
\end{equation*}
where $\Qb=\diag\big(\{h_j\}_{j=1}^q\big)$, for $h_i=1-(1-\Pi_i)^q$ the probability that the $i^{th}$ block was sampled at least once. Since we assume that the blocks $\Ab_{\{K\}}$ are indexed in ascending order according their block leverage scores, i.e., $\Pi_i\leqslant\Pi_{i+1}$ for all $i\in\N_{K-1}$, it follows that $h_i\leqslant h_K$ for all $i$, thus
$$ \E\left[\Sbwt^\top\big(\Sbwt\Sbwt^\top\big)^{-1}\Sbwt\right] = \diag\left(\{h_j\}_{j=1}^q\right)\otimes\Ib_\tau \preceq h_K\cdot\Ib_N = \big(1-(1-\Pi_K)^q\big)\cdot\Ib_N \preceq q\Pi_K\cdot\Ib_N . $$
Consequently, since $\Pi_{\{K\}}$ are block $\alpha$-balanced, we have
\begin{equation*}
  \left\|\E\left[\Sbwt^\top\big(\Sbwt\Sbwt^\top\big)^{-1}\Sbwt\right]\right\|_2 \leqslant q\cdot\Pi_K \leqslant \alpha\cdot\frac{q}{K} = \alpha\cdot\frac{r}{N} .
\end{equation*}
This completes the proof, as we can take $\eta=\alpha$.
\end{proof}

\section{Concrete Example of Induced Sketching}
\label{app_example}

In this appendix, we give a simple demonstration of the induced sketching matrices, through our proposed GC approach of Algorithm \ref{app_GC_alg}. For simplicity, we will consider exact sampling, with $\Pi_{\{K\}}=\{3/20,3/20,4/20,5/20,5/20\}$ for $K=5$, an arbitrary large block size of $\tau$, and $m=20$. The optimal replication numbers resulting from this distribution are $r^{\star}_{\{K\}}=\{3,3,4,5,5\}$, hence $m=R$. In order to obtain a reduced dimension $r$ which is 60\% of the original $N$, we carry out $q=3$ sampling trials at each iteration.

Let the resulting index multisets corresponding to the encoded pairs $(\At_j,\bt_j)$ of the first four iterations, along with the resulting estimated gradients, be:
\begin{enumerate}
  \item $\Scal^{[1]}=\{1,4,5\} \quad \implies \quad \gh^{[1]}=\nabla_\xb L_\Sb\left(\Sbwt_{[1]},\Ab,\bb;\xb^{[1]}\right)=\gh_1^{[1]}+\gh_4^{[1]}+\gh_5^{[1]}$\\
  \item $\Scal^{[2]}=\{3,5,5\} \quad \implies \quad \gh^{[2]}=\nabla_\xb L_\Sb\left(\Sbwt_{[2]},\Ab,\bb;\xb^{[2]}\right)=\gh_3^{[2]}+\gh_5^{[2]}+\gh_5^{[2]}$
  \item $\Scal^{[3]}=\{2,4,5\} \quad \implies \quad \gh^{[3]}=\nabla_\xb L_\Sb\left(\Sbwt_{[3]},\Ab,\bb;\xb^{[3]}\right)=\gh_2^{[3]}+\gh_4^{[3]}+\gh_5^{[3]}$
  \item $\Scal^{[4]}=\{4,1,4\} \quad \implies \quad \gh^{[4]}=\nabla_\xb L_\Sb\left(\Sbwt_{[4]},\Ab,\bb;\xb^{[4]}\right)=\gh_4^{[4]}+\gh_1^{[4]}+\gh_4^{[4]}$.
\end{enumerate}
Considering $\Scal^{[1]}$ and $\Scal^{[2]}$ the first two iterations of this example, the index subsets are $\{\K_1^1,\K_4^2,\K_5^3\}$ and $\{\K_3^1,\K_5^2,\K_5^3\}$ respectively. Then, the corresponding induced block leverage score sketching matrices of Algorithm \ref{alg_1_pseudocode}, are:
\begin{enumerate}
  \item $\Sbwt_{[1]} = \begin{pmatrix} 1/\sqrt{3\Pi_1} & 0 & 0 & 0 & 0 \\ 0 & 0 & 0 & 1/\sqrt{3\Pi_4} & 0 \\ 0 & 0 & 0 & 0 & 1/\sqrt{3\Pi_5} \end{pmatrix} \otimes\Ib_\tau$\\
    {\white Neo}
  \item $ \Sbwt_{[2]} = \begin{pmatrix} 0 & 0 & 1/\sqrt{3\Pi_3} & 0 & 0 \\ 0 & 0 & 0 & 0 & 1/\sqrt{3\Pi_5} \\ 0 & 0 & 0 & 0 & 1/\sqrt{3\Pi_5} \end{pmatrix}\otimes\Ib_\tau$\\
    {\white Neo}
  \item $\Sbwt_{[3]} = \begin{pmatrix} 0 & 1/\sqrt{3\Pi_2} & 0 & 0 & 0 \\ 0 & 0 & 0 & 1/\sqrt{3\Pi_4} & 0 \\ 0 & 0 & 0 & 0 & 1/\sqrt{3\Pi_5} \end{pmatrix}\otimes\Ib_\tau$\\
    {\white Neo}
  \item $\Sbwt_{[4]} = \begin{pmatrix} 0 & 0 & 0 & 1/\sqrt{3\Pi_4} & 0 \\ 1/\sqrt{3\Pi_1} & 0 & 0 & 0 & 0 \\  0 & 0 & 0 & 1/\sqrt{3\Pi_4} & 0 \end{pmatrix}\otimes\Ib_\tau$
\end{enumerate}
each of which are of size $(3\tau)\times(5\tau)$.

\section{Comparison to the block-SRHT}
\label{comp_SRHT_app}

An alternative view point is that the matrix
$$ \Ubt_{\mathrm{exp}}\coloneqq\PsiB\cdot(\Sigb\Vb^\top)^{-1} $$
has uniform Frobenius block scores, which further justifies the proposed GC approach for homogeneous servers. This resembles the main idea behind the SRHT \cite{AC06,DMMS11} and different variants known as \textit{block-SRHT} \cite{CMPH22,BBGL22}, where a random projection is applied to the data matrix to \textit{flatten} its leverage scores, i.e., make them close to uniform. These two techniques for flattening the corresponding block scores are vastly different, and the proximity of the corresponding block scores are measured and quantified differently. In contrast to the block-SRHT, the quality of the approximation in our case depends on the misestimation factor $\beta_{\Pit}$, and is not quantified probabilistically.

We note that since $\PsiB$ has repeated blocks from the expansion, the scores we consider in Lemma \ref{Lemma_flat} are not the block leverage scores of $\PsiB$. The Frobenius block scores of $\Ubt_{\mathrm{exp}}$, are in fact the corresponding block leverage scores of $\Abt$, which are replicated in the expansion. Moreover, note that the closer $\beta_{\Pit}$ is to $1$, the closer the sampling distribution according to the Frobenius block scores of $\Ubt_{\mathrm{exp}}$; which we denote by $\Qt_{\{R\}}$, is to being exactly uniform. We denote the uniform sampling distribution by $\Uu_{\{R\}}$.

\begin{Lemma}
\label{Lemma_flat}
When $\Pit_{\{K\}}=\Pi_{\{K\}}$, the sampling distribution $\Qt_{\{R\}}$ is uniform. When $\Pit_\iota\geqslant\beta_{\Pit}\Pi_\iota$ for $\beta_{\Pit}=\min_{i\in\N_K}\{\Pi_i/\Pit_i\}\in(0,1)$ and all $\iota\in\N_K$, the resulting distribution $\Qt_{\{R\}}$ is approximately uniform, and satisfies  $d_{\Uu,\Qt}\leqslant1\big/(R\beta_{\Pit})$.
\end{Lemma}

\begin{proof}
Let $\Ub=\Big[\Ub_1^\top \ \cdots \ \Ub_K^\top\Big]^\top$ denote the corresponding blocks of $\Ub$ according the partitioning of $\D$. Without loss of generality, assume that the data points within each partition are consecutive rows of $\Ab$, and $\Ub_\iota\in\R^{\tau\times d}$ for all $\iota\in\N_K$.

From \eqref{expansion_matrix} and \eqref{scalar_enc_matrix}, it follows that
\begin{align*}
  \PsiB &= \Ebwt\cdot\Abt \\
  &= \left(\Eb\otimes\Ib_\tau\right)\cdot\left(\Gb\cdot\Ub\Sigb\Vb^\top\right) \\
  &= \left(\Eb\otimes\Ib_\tau\right) \cdot \left[\Ub_1^\top\big/\sqrt{q\Pit_1} \ \cdots \ \Ub_K^\top\big/\sqrt{q\Pit_K}\right]^\top \cdot \Sigb\Vb^\top \\
  &\eqqcolon \left(\Eb\otimes\Ib_\tau\right) \cdot \left[\Ubt_1^\top \ \cdots \ \Ubt_K^\top\right]^\top \cdot \Sigb\Vb^\top \\
  &\eqqcolon \overbrace{\bigg[ \underbrace{\Ubt_1^\top \ \cdots \ \Ubt_1^\top}_{r_1} \ \underbrace{\Ubt_2^\top \ \cdots \ \Ubt_2^\top}_{r_2} \ \cdots \ \underbrace{\Ubt_K^\top \ \cdots \ \Ubt_K^\top}_{r_K} \bigg]^\top}^{\Ubt_{\mathrm{exp}}\in\R^{R\tau\times d}} \cdot \Sigb\Vb^\top .
\end{align*}
Note that $\Ubt_{\mathrm{exp}}\Sigb\Vb^\top$ is not the $\SVD$ of $\PsiB$. For the normalizing factor of $\frac{q}{Rd}$:
\begin{align*}
  \Qt_\iota &= \frac{q}{Rd}\cdot\|\Ubt\|_F^2 = \frac{q}{Rd}\cdot\frac{\|\Ub_\iota\|_F^2}{q\Pit_\iota} = \frac{\Pi_\iota}{R\Pit_\iota} \leqslant \frac{1}{R\beta_{\Pit}}
\end{align*}
therefore
\begin{align*}
  \sum_{i=1}^R\big|\Qt_i-1/R\big|\overset{\triangle}{\leqslant} \frac{R}{R\beta_{\Pit}} = \frac{1}{\beta_{\Pit}}
\end{align*}
where $\triangle$ follows from the fact that $\big|\Qt_i-1/R\big|\leqslant1\big/\big(R\beta_{\Pit}\big)$ for each $i\in\N_R$. After normalizing by $1/R$ according to the distortion metric, we deduce that $d_{\Uu,\Qt}\leqslant1\big/\big(R\beta_{\Pit}\big)$.

In the case where $\Pit_{\{K\}}=\Pi_{\{K\}}$, we have
$$ \Qt_\iota = \frac{q}{Rd}\cdot\|\Ubt_\iota\|_F^2 = \frac{q}{Rd}\cdot\frac{\|\Ub_\iota\|_F^2}{q\Pi_\iota} = \frac{\Pi_\iota}{R\Pi_\iota} = \frac{1}{R} $$
for all $\iota\in\N_K$, thus $\Qt_{\{K\}}=\Uu_{\{K\}}$.
\end{proof}

Finally, in Proposition \ref{comparison_two_embds} we show when the block leverage score sampling sketch of Algorithm \ref{alg_1_pseudocode} and the block-SRHT of \cite{CMPH22} have the same $\ell_2$-subspace embedding guarantee. We first recall the corresponding result to Theorem \ref{subsp_emb_thm_lvg}, of the block-SRHT.

\begin{Thm}[{\cite[Theorem 7]{CMPH22}}]
\label{subsp_emb_thm_SRHT}
The block-SRHT $\Sb_{\Hbh}$ is an $\ell_2$-subspace embedding of $\Ab$. For $\delta>0$ and $q=\Theta\big(\frac{d}{\tau}\log(Nd/\delta)\cdot\log(2d/\delta)/\epsilon^2\big)$:
$$ \Pr\left[\|\Ib_d-\Ub^\top\Sb_{\Hbh}^\top\Sb_{\Hbh} \Ub\|_2\leqslant\epsilon\right]\geqslant 1-\delta . $$
\end{Thm}

\begin{Prop}
\label{comparison_two_embds}
Let $\beta=1$. For $\delta=e^{\Theta(1)}/(Nd)$, the sketches of Algorithm \ref{alg_1_pseudocode} and the block-SRHT of \cite{CMPH22} achieve the same asymptotic $\ell_2$-subspace embedding guarantee, for the same number of sampling trials $q$.
\end{Prop}

\begin{proof}
For $\delta=e^{\Theta(1)}/(Nd)$, the two sketching methods have the same $q$, and both satisfy the $\ell_2$-subspace embedding property with error probability $1-\delta$.
\end{proof}

\section{Contraction Rate of Block Leverage Score Sampling}
\label{cont_rate_app}

In this appendix we quantify the contraction rate of our method on the error term $\xb^{[s]}-\xb^\star$, which further characterizes the convergence of SD after applying our method. The contraction rate is compared to that of regular SD.

Recall that the contraction rate of an iterative process given by a function $h(x^{[s]})$ is the constant $\gamma\in(0,1)$ for which at each iteration we are guaranteed that $h(x^{[s+1]})\leqslant \gamma\cdot h(x^{[s]})$, therefore $h(x^{[s]})\leqslant \gamma_s\cdot h(x^{[0]})$. Let $\xi$ be a fixed step-size, $\Sbwt_{[s]}$ the induced sketching matrix of Algorithm \ref{alg_1_pseudocode} at iteration $s$, and define $\Bb_{SD}=\left(\Ib_d-2\xi\cdot\Ab^\top\Ab\right)$ and $\Bb_s=\left(\Ib_d-2\xi\cdot\Ab^\top\Sbwt_{[s]}^\top\Sbwt_{[s]}\Ab\right)$. We note that the contraction rates could be further improved if one also incorporates an optimal step-size. It is also worth noting that when weighting from Appendix \ref{weighting_sec} is introduced, we have the same contraction rate and straggler ratio.

\begin{Lemma}
\label{lem_exp_StS}
For $\Sbwt$ the sketching matrix of Algorithm \ref{alg_1_pseudocode}, we have $\E\left[\Sbwt^\top\Sbwt\right]=\Ib_N$.
\end{Lemma}

\begin{proof}
Similar to the proof of Proposition \ref{prop_Sp_lvg}, we define a symmetric random matrix $\Yb$, whose realizations correspond to the sampled and rescaled submatrices of Algorithm \ref{alg_1_pseudocode}. The realizations are
$$ \Yb_i = \frac{\Ib_{(\K_\iota^i)}^\top\Ib_{(\K_\iota^i)}}{q\Pit_\iota} = \frac{\sum_{l\in\K_\iota^i}\eb_l\eb_l^\top}{q\Pit_\iota} . $$
After $q$ sampling trials, we have $\Sbwt^\top\Sbwt=\sum_{i=1}^q\Yb_i$. It follows that
\begin{align*}
  \E\left[\Sbwt^\top\Sbwt\right] &= \E\left[\sum_{i=i}^q\Yb_i\right]\\
  &= \sum_{i=1}^q\E\left[\Yb_i\right]\\
  &= q\cdot \left(\sum_{j=1}^K\Pit_j\cdot\frac{\Ib_{(\K_j)}^\top\Ib_{(\K_j)}}{q\Pit_j}\right)\\
  &= \sum_{j=1}^K\Ib_{(\K_j)}^\top\Ib_{(\K_j)}\\
  &= \sum_{l=1}^N\eb_{(l)}\eb_{(l)}^\top\\
  &= \Ib_N .
\end{align*}
\end{proof}

\begin{Thm}
\label{contr_rate_thm}
The contraction rate of our GC approach through the expected sketch $\Sbwt_{[s]}$ at each iteration, is equal to the contraction rate of regular SD. Specifically, for $e_s\coloneqq\xb^{[s]}-\xb^\star$ the error at iteration $s$ and $\gamma_{SD}=\lambda_1(\Bb_{SD})$ the contraction rate of regular SD, we have $\big\|\E[e_{s+1}]\big\|_2\leqslant\gamma_{SD}\cdot\|e_s\|_2$.
\end{Thm}

\begin{proof}
For a fixed step-size $2\xi$, our SD parameter update at iteration $s+1$ is
$$ \xb^{[s+1]}\gets\xb^{[s]}-2\xi\cdot\Ab^\top\Sbwt_{[s]}^\top\Sbwt_{[s]}\left(\Ab\xb^{[s]}-\bb\right)\ , $$
where for regular SD we have $\Sbwt_{[s]}\gets\Ib_N$. At iteration $s+1$, the error $e_s$ of the previous iteration is not random, hence $\E[e_s]=e_s$. By substituting the expression of $\Bb_s$, it follows that
\begin{align}
  e_{s+1} &= \xb^{[s+1]}-\xb^\star \notag\\
  &= \left(\xb^{[s]}-2\xi\Ab^\top\Sbwt_{[s]}^\top\Sbwt_{[s]}\left(\Ab\xb^{[s]}-\bb\right)\right)-\xb^\star \notag\\
  &= \xb^{[s]}-2\xi\Ab^\top\Sbwt_{[s]}^\top\Sbwt_{[s]}\Ab\xb^{[s]}+2\xi\Ab^\top\Sbwt_{[s]}^\top\Sbwt_{[s]}\bb-\xb^\star \notag\\
  &= \Bb_s\xb^{[s]}-\left(\xb^\star-2\xi\Ab^\top\Sbwt_{[s]}^\top\Sbwt_{[s]}\bb\right) \notag\\
  &= \Bb_s\xb^{[s]}-\left(\xb^\star-2\xi\Ab^\top\Sbwt_{[s]}^\top\Sbwt_{[s]}\left(\Ab\xb^\star+\bb^\perp\right)\right) \notag\\
  &= \Bb_s\left(\xb^{[s]}-\xb^\star\right) - 2\xi\Ab^\top\Sbwt_{[s]}^\top\Sbwt_{[s]}\bb^\perp \notag\\
  &= \Bb_s e_s - 2\xi\Ab^\top\Sbwt_{[s]}^\top\Sbwt_{[s]}\bb^\perp \label{contr_rate_Bs}
\end{align}
and by Lemma \ref{lem_exp_StS}
\begin{align}
  \E\left[2\xi\Ab^\top\Sbwt_{[s]}^\top\Sbwt_{[s]}\bb^\perp\right] &= 2\xi\Ab^\top\E\left[\Sbwt_{[s]}^\top\Sbwt_{[s]}\right]\bb^\perp \notag\\
  &= 2\xi\Ab^\top\bb^\perp \notag\\
  &= \bold{0}_{d\times 1} \label{exp_null_term}
\end{align}
as $\bb^\perp$ lies in the kernel of $\Ab^\top$, and
\begin{align}
  \E\left[\Bb_s\right] &= \Ib_d-2\xi\Ab^\top\E\left[\Sbwt_{[s]}^\top\Sbwt_{[s]}\right]\Ab \notag\\
  &= \Ib_d-2\xi\Ab^\top\Ab \notag\\
  &= \Bb_{SD} . \label{exp_Bs_term}
\end{align}
From \eqref{contr_rate_Bs}, \eqref{exp_null_term} and \eqref{exp_Bs_term}, it follows that
\begin{align*}
  \E\left[e_{s+1}\right] &= \E\left[\Bb_s e_s\right] - \E\left[2\xi\Ab^\top\Sbwt_{[s]}^\top\Sbwt_{[s]}\bb^\perp\right]\\
  &= \E\left[\Bb_s \right] \cdot e_s\\
  &= \Bb_{SD}\cdot e_s .
\end{align*}
This results in the inequality
$$ \big\|\E[e_{s+1}]\big\|_2 \leqslant \lambda_1\big(\E[\Bb_s]\big)\cdot\|e_{s}\|_2 = \lambda_1(\Bb_{SD})\cdot\|e_{s}\|_2 $$
which implies the contraction rate of the expected sketch through Algorithm \ref{alg_1_pseudocode}:
$$ \gamma_{s+1}=\lambda_1(\Bb_{SD}). $$
By replacing $\Bb_s$ with $\Bb_{SD}$ in \eqref{contr_rate_Bs} and $\Sbwt_{[s]}\gets\Ib_N$, we conclude that the contraction rate of SD is $\gamma_{SD}=\lambda_1(\Bb_{SD})$.
\end{proof}

We conclude this appendix by stating the expected ratio of dimensions $r$ to $N$, and stragglers to servers.

\begin{Rmk}
The expected ratio of the reduced dimension $r$ to the original dimension $N$ is $q(T)\tau/N$, and the expected straggler to servers ratio is $(m-q)/m$. Since we stop receiving computations at a time instance $T$; we expect that $q\gets q(T)$ computations are received, hence there are $m-q$ stragglers. Therefore, the straggler to servers ratio is $(m-q)/m$. The expected ratio between the two dimensions is immediate from the fact that $r=q\tau$ is the new reduced dimension, in the case where $q\leqslant K$.
\end{Rmk}

\section{Weighted Block Leverage Score Sketching}
\label{weighting_sec}

So far, we have considered sampling with replacement according to the normalized block leverage scores, to reduce the effective dimension $N$ of $\Ab$ and $\bb$ to $r=q\tau$. In this appendix, we show that by \textit{weighting} the sampled blocks according to the sampling which has taken place through $\Omb$ for the construction of the sketching matrix $\Sbwt$, we can further compress the data matrix $\Ab$, and get the same results when first and second order optimization methods are used to solve \eqref{mod_OLS}. The weighting we propose is more beneficial with non-uniform distributions, as we expect the sampling with replacement to capture the importance of the more influence blocks. The \textit{weighted sketching matrix} $\Sbw$ we propose, is a simple extension to $\Sbwt$ of Algorithm \ref{alg_1_pseudocode}. For simplicity, we do not consider iterative sketching, though similar arguments and guarantees can be derived.

The main idea is to not keep repetitions of blocks which were sampled multiple times, but rather weigh each block by the number of times it was sampled. By doing so, we retain the \textit{weighted sketch} $\Sbw\Ab$ of size $\qbar\tau\times d$, for $\qbar$ the number of distinct blocks that were sampled.\footnote{In practice, for highly non-uniform $\Pi_{\{K\}}$ we expect $\qbar\tau\ll r=q\tau$. The sketch $\Sbw\Ab$ could therefore be stored in much less space than $\Sbwt\Ab$, and the system of equations $\Sbw(\Ab-\bb)=\bold{0}_{\qbar\tau\times1}$ could have significantly fewer equations than $\Sbwt(\Ab-\bb)=\bold{0}_{q\tau\times1}$.} Additionally, the gradient and Hessian of $L_\Sb\big(\Sbw,\Ab,\bb;\xb\big)$ are respectively equal to those of $L_\Sb\big(\Sbwt,\Ab,\bb;\xb\big)$, and are unbiased estimators of the gradient and Hessian of $L_{ls}(\Ab,\bb;\xb)$.

To achieve the weighting algorithmically, we count how many times each of the distinct $\qbar$ blocks were sampled, and at the end of the sampling procedure we multiply the blocks by their corresponding \textit{weight}. We initialize a weight vector $\wb=\bold{0}_{1\times K}$, and in the sampling procedure whenever the $i^{th}$ partition is drawn, we update its corresponding weight: $\wb_i\gets\wb_i+1$. It is clear that once $q$ trials are been carried out, we have $\|\wb\|_1=q$ for $\wb\in\N_0^{1\times K}$.

Let $\Scal$ denote the index multiset observed after the sampling procedure of Algorithm \ref{alg_1_pseudocode}, and $\Sbar$ the set of indices comprising $\Scal$. That is, $\Scal$ has cardinality $q$ and may have repetitions, while $\Sbar=\N_K\cap\Scal$ has cardinality $\qbar=|\Sbar|\leqslant q$ with no repetitions. We denote the ratio of the two sets by $\zeta\coloneqq q/\qbar=\|\wb\|_1/\|\wb\|_0\geqslant1$, which indicates how much further compression we get by utilizing the fact that blocks may be sampled multiple times. The corresponding weighted sketching matrix $\Sbw$ of $\Sbwt$ is then
\begin{equation}
\label{expr_Sw}
  \Sbw = \overbrace{\diag\left(\Big\{\sqrt{\wb_j\big/\big(q\Pi_j}\big)\Big\}_{j\in\Sbar}\right)}^{\Wbb_{1/2}\in\R_{\geqslant0}^{\qbar\times\qbar}}\cdot\Ib_{(\Sbar)}\otimes\Ib_\tau \in \R^{\qbar\tau\times N}
\end{equation}
where $\Ib_{(\Sbar)}\in\{0,1\}^{\qbar\times K}$ is the restriction of $\Ib_{K}$ to the rows indexed by $\Sbar$.

For simplicity, assume that the sampling matrices which are devised for $\Sbwt$ and $\Sbw$ follow the ordering of the sampled blocks in order of the samples, i.e., if $\left(\Omb_{(i)}\right)_j=1$ then $\left(\Omb_{(i+1)}\right)_l=1$ for $l\geqslant j$, and equivalently $\Scal_l\leqslant\Scal_{l+1}$ and $\Sbar_l<\Sbar_{l+1}$ for all valid $l$. We restrict the sampling matrix $\Omb\in\{0,1\}^{q\times K}$ to its unique rows, by applying $\Ombb\in\{0,1\}^{\qbar\times q}$:
$$ \Ombb_{ij} = \begin{cases} 1 \quad \text{ for } i=1 \text{ and } j=\Scal_1 \\ 1 \quad \text{ if } \Scal_j>\Scal_{j-1} \text{ for } j\in\N_q\backslash\{1\} \\ 0 \quad \text{ otherwise } \end{cases} $$
to the left of the sampling matrix $\Omb$ of Algorithm \ref{alg_1_pseudocode}, i.e., $\Ombw \coloneqq \Ombb\cdot\Omb\in\{0,1\}^{\qbar\times K}$. This sampling matrix then satisfies
$$ \left(\Ombw\right)_{ij} = \begin{cases} 1 \quad \text{ if } j=\Sbar_j \\ 0 \quad \text{ otherwise } \end{cases} \text{ for } j\in\N_{\qbar} . $$
Let $\wbt=\wb_{|_{\Sbar}}\in\Z_+^{1\times \qbar}$ be the restriction of $\wb$ to its nonzero elements; hence $\|\wbt\|_1=\|\wb\|_1=|\Scal|=q$, and define the rescaling diagonal matrix $\Wbt_{1/2}=\diag\left(\big\{\sqrt{\wbt_i}\big\}_{i=1}^{\qbar}\right)$. We then have the following relationship
\begin{align}
  \Sbw &= \overbrace{\left(\Wbt_{1/2}\cdot\Ombb\cdot\Db\cdot\Ombb^\top\right)}^{=\Wbb_{1/2}}\cdot\Ombw\otimes\Ib_{\tau} \notag\\
  &= \big(\Wbb_{1/2}\otimes\Ib_\tau\big)\cdot\big(\Ombw\otimes\Ib_\tau\big) \label{Sb_WSp_expr_lvg}
\end{align}
when $\Sbar=\N_K\cap\Scal$.

As previously noted, $\Sbw$ has $\zeta$ times less rows than $\Sbwt$. Hence, the required storage space for the sketch $\Sbw\Ab$ drops by a multiplicatively factor of $\zeta$, and the required operations are reduced analogously, according to the computation. The weighted sketching matrix $\Sbw$ has the following guarantees, which imply that the proposed weighting will not affect first or second order iterative methods which are used to approximate \eqref{mod_OLS}.

\begin{Prop}
\label{prop_wght_lr}  
The resulting gradient and Hessian of the modified least squares problem \eqref{mod_OLS} when sketching with $\Sbwt$ of Algorithm \ref{alg_1_pseudocode}, are respectively identical to the resulting gradient and Hessian when sketching with $\Sbw$ presented in \eqref{expr_Sw} and \eqref{Sb_WSp_expr_lvg}.
\end{Prop}

\begin{proof}
From Algorithm \ref{alg_1_pseudocode}, the assumption on the ordering of the elements in $\Scal$ and $\Sbar$, and the construction of $\Sbwt$, we have
\begin{align*}
  \Sbw^\top\Sbw &=  \Big(\big(\Ombw^\top\Wbt_{1/2}^\top\big)\otimes\Ib_\tau\Big) \cdot \Big(\big(\Wbt_{1/2}\cdot\Ombw\big)\otimes\Ib_\tau\Big)\\
  &= \Big(\big(\Omb^\top\Db^\top\big)\otimes\Ib_\tau\Big) \cdot \Big(\big(\Db\cdot\Omb\big)\otimes\Ib_\tau\Big)\\
  &= \Sbwt^\top\Sbwt .
\end{align*}
Let $\Tcal=\biguplus_{j\in\Scal}\K_j$ and $\Tbar=\bigsqcup_{j\in\Sbar}\K_j$, thus $\Tbar$ is contained in $\Tcal$ when both are viewed as multisets. Considering the objective function $L_\Sb(\Sb,\Ab,\bb;\xb)$ of \eqref{mod_OLS}, the equivalence of gradients is observed through the following computation
\begin{align*}
  \nabla_{\xb}L_\Sb(\Sbwt,\Ab,\bb;\xb) &= 2\Ab^\top\left(\Sbwt^\top\Sbwt\right)\left(\Ab\xb-\bb\right)\\
  &= 2\sum\limits_{l\in\Tcal}\Ab_{(l)}^\top \cdot\Db_{ll}^2 \cdot  \left(\Ab_{(l)}\xb-\bb_l\right)\\
  &= 2\sum\limits_{j\in\Tbar}\wbt_j\cdot\Ab_{(j)}^\top \cdot\Db_{jj}^2 \cdot  \left(\Ab_{(j)}\xb-\bb_j\right)\\
  &= 2\sum\limits_{j\in\Tbar}\Ab_{(j)}^\top \cdot\big(\Wbt_{1/2}\big)_{jj}^2 \cdot  \left(\Ab_{(j)}\xb-\bb_j\right)\\
  &= 2\Ab^\top\left(\Sbw^\top\Sbw\right)\left(\Ab\xb-\bb\right)\\
  &= \nabla_{\xb}L_\Sb(\Sbw,\Ab,\bb;\xb) .
\end{align*}

Recall that the Hessian of the least squares objective function \eqref{x_star_lr} is $\nabla_{\xb}^2L_{ls}(\Ab,\bb;\xb)=2\Ab^\top\Ab$. Considering the modified objective function \eqref{mod_OLS} and our sketching matrices, it follows that
\begin{align*}
  \nabla_{\xb}^2L_\Sb\left(\Sbwt,\Ab,\bb;\xb\right) &= 2\Ab^\top\left(\Sbwt^\top\Sbwt\right)\Ab\\
  &= 2\sum\limits_{l\in\Tcal}\Ab_{(l)}^\top \cdot\Db_{ll}^2 \cdot \Ab_{(l)}\\
  &= 2\sum\limits_{j\in\Tbar}\wbt_j\cdot\Ab_{(j)}^\top \cdot\Db_{jj}^2 \cdot \Ab_{(j)}\\
  &= 2\sum\limits_{j\in\Tbar}\Ab_{(j)}^\top \cdot\big(\Wbt_{1/2}\big)_{jj}^2 \cdot \Ab_{(j)}\\
  &= 2\Ab^\top\left(\Sbw^\top\Sbw\right)\Ab\\
  &= \nabla_{\xb}^2L_\Sb\big(\Sbw,\Ab,\bb;\xb\big)
\end{align*}
which completes the proof.
\end{proof}

\begin{Cor}
\label{unb_est_Sw}
At each iteration, the gradient and Hessian of the weighted sketch system of equations $\Sbw(\Ab-\bb)=\bold{0}_{\qbar\tau\times1}$, are unbiased estimators of the gradient and Hessian of the original system $(\Ab-\bb)=\bold{0}_{N\times1}$.
\end{Cor}

\begin{proof}
Denote the gradient and Hessian of the weighted sketch at iteration $s$ by $\gh_\wb^{[s]}$ and $\Hh_\wb^{[s]}$ respectively. By Proposition \ref{prop_wght_lr} we know that $\gh_\wb^{[s]}=\gh^{[s]}$, and by Theorem \ref{SGD_nonunif_thm} that $\E\left[\gh^{[s]}\right]=g^{[s]}$. Hence $\E\left[\gh_\wb^{[s]}\right]=g^{[s]}$.

Following the same notation as in the proof of Theorem \ref{SGD_nonunif_thm}, the Hessian $\Hh^{[s]}=\nabla_{\xb}^2L_\Sb\left(\Sbwt_{[s]},\Ab,\bb;\xb^{[s]}\right)$ is 
$$ \Hh^{[s]} = 2\sum_{i\in\I^{[s]}}\frac{1}{q\Pit_i}\Ab_i^\top\Ab_i $$
thus
\begin{align*}
  \E\left[ \Hh^{[s]} \right] &= 2\E\left[\sum_{i\in\I^{[s]}}\frac{1}{q\Pit_i}\Ab_i^\top\Ab_i\right]\\
  &= 2\sum_{i\in\I^{[s]}}\sum_{j=1}^K\Pit_j\frac{1}{q\Pit_j}\Ab_j^\top\Ab_j\\
  &= 2q\cdot\sum_{j=1}^K\frac{1}{q}\Ab_i^\top\Ab_i\\
  &= 2\Ab^\top\Ab
\end{align*}
which is precisely the Hessian of \eqref{x_star_lr}. By Proposition \ref{prop_wght_lr}, it follows that $\E\left[ \Hh_\wb^{[s]} \right]=2\Ab^\top\Ab$, which completes the proof.
\end{proof}

Geometrically, from the point of view of adding vectors, the partial gradients of the partitions sampled will be scaled accordingly to their weights. Therefore, the partial gradients $\gh_i$ with higher weights have a greater influence in the direction of the resulting gradient $\gh$. This was also the fundamental idea behind our sketching and GC techniques, as the partitions sampled multiple times are of greater importance.

Next, we quantify the expected dimension of the weighted sketch $\Sbw\Ab$. This shows the dependence on $\Pit_{\{K\}}$, and further justifies that we attain a higher compression factor $\zeta$ when the block leverage scores are non-uniform. We note that when the data points are i.i.d. and $N$ is significantly larger than $d$, then the leverage scores will be approximately equal; and use of these scores to determine dominant rows will offer little or no benefit.

\begin{Thm}
\label{thm_exp_dim}
The expected reduced dimension of $\Sbw\Ab$ is 
$$ \Big(K-\sum_{i=1}^K(1-\Pit_i)^q\Big)\cdot\tau $$
which is maximal when $\Pit_{\{K\}}$ is uniform.
\end{Thm}

\begin{proof}
It suffices to determine the expected number of distinct blocks $\Ab_i$ which are sampled after $q$ trials when carrying out Algorithm \ref{alg_1_pseudocode}. The probability of not sampling $\Ab_i$ at a given trial is $\big(1-\Pit_i\big)$, hence not sampling $\Ab_i$ at any trial occurs with probability $\big(1-\Pit_i\big)^q$, since the trials are identical and independent. Thus, the expected number of distinct blocks being sampled is\begin{align*}
  \E[\qbar] &= \sum_{i=1}^K1\cdot\Pr\big[\Ab_i\text{ was sampled at least once sampled}\big]\\
  &= \sum_{i=1}^K\left(1-\cdot\Pr\big[\Ab_i\text{ was not sampled at any trial}\big]\right)\\
  &= \sum_{i=1}^K\left(1-\left(1-\Pit_i\right)^q\right)\\
  &= K-\sum_{i=1}^K\left(1-\Pit_i\right)^q .
\end{align*}
Thus, the expected reduced dimension is $\tau\cdot\E[\qbar]$.

Let $Q\big(\Pit_{\{K\}}\big)\coloneqq\sum_{i=1}^K(1-\Pit_i)^q$, and introduce the Lagrange multiplier $\lambda>0$ to the constraint $R\big(\Pit_{\{K\}}\big)=\left(\sum_{i=1}^K\Pit_i-1\right)$, to get the Lagrange function
\begin{align*}
  \lag\big(\Pit_{\{K\}},\lambda\big) &\coloneqq Q\big(\Pit_{\{K\}}\big)+\lambda\cdot R\big(\Pit_{\{K\}}\big)\\
  &= \sum_{i=1}^K\left(\lambda\cdot\Pit_i+\big(1-\Pit_i\big)^q\right)-\lambda
\end{align*}
for which
\begin{equation}
\label{partial_Pi}
  \frac{\partial \lag\big(\Pit_{\{K\}},\lambda\big)}{\partial\Pit_i} = \lambda-q(1-\Pit_i)^{q-1} = 0
\end{equation}
\begin{equation}
\label{partial_set_0}
   \implies \ \Pit_i=1-(\lambda/q)^{1/(q-1)} \ \ \text{ and } \ \ \lambda=q(1-\Pit_i)^{q-1}
\end{equation}
for all $i\in\N_K$, and
\begin{equation}
\label{partial_lambda}
  \frac{\partial \lag\big(\Pit_{\{K\}},\lambda\big)}{\partial\lambda} = \sum_{i=1}^K\Pit_i-1=0 .
\end{equation}
Note that the uniform distribution $\Uu_{\{K\}}=\big\{\Pit_i=1/K\big\}_{i=1}^K$ is a solution to \eqref{partial_lambda}. We will now verify that $\Uu_{\{K\}}$ satisfies \eqref{partial_Pi}. From \eqref{partial_set_0}, for $\Pit_{\{K\}}\gets\Uu_{\{K\}}$ we have $\lambda=q(1-1/K)^{q-1}>0$, which we substitute into \eqref{partial_Pi}:
\begin{equation*}
  \lambda-q(1-\Pit_i)^{q-1} = q(1-1/K)^{q-1} - q(1-1/K)^{q-1}=0 .
\end{equation*}
Hence, $\Uu_{\{K\}}$ is the solution to both \eqref{partial_Pi} and \eqref{partial_lambda}.

By the second derivative test, since $\partial^2\lag\big(\Pit_{\{K\}}\big)\big/\partial\Pit_i^2 = q(q-1)(1-\Pit_i)^{q-2}$ is positive for $\Pit_i=1/K$, we conclude that $Q\big(\Uu_{\{K\}}\big)\leqslant Q\big(\Pit_{\{K\}}\big)$ for any $\Pit_{\{K\}}\neq\Uu_{\{K\}}$. This implies that $\E[\qbar]$ is maximal when $\Pit_{\{K\}}=\Uu_{\{K\}}$, and so is the expected reduced dimension of $\Sbw\Ab$.
\end{proof}

We further note that $\E[\qbar]$ from the proof of Theorem \ref{thm_exp_dim}, is trivially minimal in the degenerate case where $\Pit_\iota=1$ for a single $\iota\in\N_{K}$, and $\Pit_j=0$ for every $j\in\N_{K}\backslash\{\iota\}$. This occurs in the case where $\Ab_j=\bold{0}_{\tau\times d}$ for each $j$, and $\qbar$ is therefore exactly
$$ K-\sum_{j\neq\iota}(1-\Pit_j)^q = K-\sum_{j\neq\iota}1^q = K-(K-1) = 1 . $$


\section{Table of Notational Conventions and $\ell_2$-subspace embedding Condition}
\label{notation_app}

In this appendix, we collect the main notational conventions used in this paper, and show how to derive the $\ell_2$-subspace embedding condition \eqref{eq_form} from \eqref{subsp_emb_Ax}. This definition of an $\ell_2$-subspace embedding and the corresponding condition, can also be found in \cite{Woo14}.

Let $\Zb=\Ib_d-(\Sb\Ub)^\top(\Sb\Ub)$, and $\yb\in\R^d$ be an arbitrary vector. In the bounds \eqref{subsp_emb_lower} and \eqref{subsp_emb_upper}, note that $\epsilon\|\yb\|_2^2>0$ and $\yb^\top\Zb\yb = -\big(\yb^\top(-\Zb)\yb\big)$, so we can work with the absolute value $\big|\yb^\top\Zb\yb\big|$ in order to take care of both inequalities. That is, together both imply that $0\leqslant\big|\ybt^\top\Zb\ybt\big|\leqslant\epsilon\|\ybt\|^2$, where for now we let $\ybt$ be a unit vector. Then, this requirement is satisfied for all vectors $\yb$ by scaling, which implies $\|\Zb\|\leqslant \epsilon$.

\begin{table}
\centering
\begin{tabular}{ |p{2.65cm}||p{5.2cm}| }
\hline
\multicolumn{2}{|c|}{\textbf{LIST OF NOTATION}} \\
\hline
\hline
  $m$ & number of computational nodes \\ \hline
  $q$ & number of sampled blocks/responsive nodes \\ \hline
  $\N_n$ & set of integers $\{1,2,\ldots,n\}$ \\ \hline
  $[s]$ & used as a superscript or subscript, to denote the iteration number `$s$' of the iterative process\\ \hline
  $X_{\{n\}}$ & set $\{X_1,\ldots,X_n\}$ \\ \hline
  $\Ab$ & dataset matrix of size $N\times d$ \\ \hline
  $\Ab_{(i)}/\Ab^{(i)}$ & $i^{th}$ row/column of $\Ab$ \\ \hline
  $\bb$ & corresponding $d\times 1$ label vector to $\Ab$ \\ \hline
  $\Sb$ & arbitrary sketching matrix of size $r\times N$ \\ \hline
  $\Sbwt$ & sketching matrix of Algorithm \ref{alg_1_pseudocode} \\ \hline
  $\Omb$ & row sampling matrix of size $r\times N$ \\ \hline
  $\Ombwt$ & block sampling matrix for $\Sbwt$ \\ \hline
  $\D$ & dataset corresponding to $\Ab,\bb$ \\ \hline
  $\D_i$ & sub-dataset corresponding to $\Ab_i,\bb_i$ \\ \hline
  $N$ & number of data points in $\D$ \\ \hline
  $K$ & number of partitions of the data matrix $\Ab$ \\ \hline
  $\tau$ & size of the data partitions, $\tau=N/K$ \\ \hline
  {\footnotesize$\Abt = \Big[\At_1^\top \ \cdots \ \At_K^\top\Big]^\top$} & encoding of blocks $\Ab_{\{K\}}$ \\ \hline
  {\footnotesize$\bbt = \Big[\bt_1^\top \ \cdots \ \bt_K^\top\Big]^\top$} & encoding of blocks $\bb_{\{K\}}$ \\ \hline
  $L_{ls}(\Ab,\bb;\xb)$ & least squares objective function $\|\Ab\xb-\bb\|_2^2$ \\ \hline
  $L_\Sb(S,\Ab,\bb;\xb)$ & modified least squares $\|S(\Ab\xb-\bb)\|_2^2$, for $S$ a specified encoding or sketching matrix \\ \hline
  $\xb^\star$ & optimal solution to \eqref{x_star_lr} \\ \hline
  $\xbh$ & approximated solution to \eqref{x_star_lr} through \eqref{mod_OLS} \\ \hline
  $\xbt$ & optimal solution to \eqref{nonit_OLS} \\ \hline
  $\I^{[s]}$ & index set of the $q$ fastest servers at iteration $s$ \\ \hline
  $\Scal^{[s]}$ & index multiset of sampled blocks at iteration $s$ \\ \hline
  $g^{[s]}$ & gradient $g^{[s]} = 2\Ab^\top(\Ab\xb^{[s]}-\bb)$ \eqref{gr_ls} \\ \hline 
  $g_i^{[s]}$ & partial gradient corresponding to $\D_i$ \eqref{part_gr_ls} \\ \hline
  $\gh^{[s]}$ & approximated gradient at iteration $s$ \eqref{appr_gr} \\ \hline
  $\gh_i^{[s]}$ & appr. partial gradient of $\D_i$ at iteration $s$ \eqref{gt_unif_id} \\ \hline
  $\xi_s$ & SD step-size at iteration $s$ \\ \hline
  $\epsilon$ & $\ell_2$-subspace embedding error \\ \hline
  $\delta$ & $\ell_2$-subspace embedding failure probability \\ \hline
  $\Pi_{\{K\}}$ & block leverage score sampling distribution \\ \hline
  ${\tilde{\Pi}_{\{K\}}}^{{\white |}}$ & approximate distribution to $\Pi_{\{K\}}$ \\ \hline
  ${\bar{\Pi}_{\{K\}}}^{{\white |}}$ & induced sampling distribution through expansion networks \\ \hline
  $\beta_{\Pib}$ & misestimation factor of $\Pib^{{\white |}}$ to $\Pi$ \\ \hline
  $d_{\Pi,Q}$ & distortion metric between $\Pi_{\{K\}}$ and $Q_{\{K\}}$ \\ \hline
  $\Delta_{\Pi,\Pib}$ & normalized distortion metric \eqref{appr_opt_sol} \\ \hline  
  $F(t)$ & mother runtime distribution of the computational network \\ \hline
  $\Ft(t)$ & runtime distribution of computational tasks \\ \hline
  $\phi(t)$ & survival function $\phi(t)\coloneqq1-\Ft(t)$ \\ \hline
  $T$ & ending time, at which the central server stops receiving computations \\ \hline
  $q(T)$ & number of received computations at time $T$, i.e., $q(T)\coloneqq\floor{\Ft(T)\cdot m}$ \\ \hline
  $\rh_{\{K\}}$ & approximate replication numbers through \eqref{appr_distr} \\ \hline
  $r^{\star}_{\{K\}}$ & optimal replication numbers to \eqref{int_program} \\ \hline
  $\rt_{\{K\}}$ & approximate replication numbers \\ \hline
  $r_{\{K\}}$ & replication numbers produced by Algorithm \ref{alg_replicas_mod_R} \\ \hline
  $R$ & $R=\sum_{i=1}^K r_i$ \\ \hline
  $\Ebwt$ & expansion matrix \eqref{expansion_matrix} \\ \hline
  $\Gb$ & encoding generator matrix \eqref{scalar_enc_matrix} \\ \hline
\end{tabular}
\caption{Common notation list used throughout the paper.}
\label{notation_table}
\end{table}


\bibliographystyle{IEEEtran}
\bibliography{refs_all}
\balance


\end{document}